\documentclass{article}
\usepackage{Style/arxiv}
\usepackage[utf8]{inputenc} 
\usepackage[T1]{fontenc}    
\usepackage{hyperref}       
\usepackage{url}            
\usepackage{booktabs}       
\usepackage{amsfonts}       
\usepackage{nicefrac}       
\usepackage{microtype}      
\usepackage{lipsum}		
\usepackage{graphicx}
\usepackage[square,sort,comma,numbers]{natbib}
\usepackage{doi}
\AtBeginDocument{%
  }

\usepackage{algpseudocode}
\usepackage{algorithm}
\usepackage{tikz-cd}
\usepackage{amsmath}
\usepackage{amsthm}
\newtheorem{theorem}{Theorem}
\usepackage{xr}
\usepackage{float}
\usepackage{placeins}
\usepackage{wrapfig}
    \usepackage{caption}
    \usepackage{subcaption}
\usepackage{listings}
\usepackage{caption} 
\newcommand{\diff}{\mathop{}\!\text{d}}
\newtheorem{remark}{Remark}

\title{A Langevin sampling algorithm inspired by the Adam optimizer}

\author{ 
Benedict Leimkuhler        \\
	School of Mathematics\\
	University of Edinburgh\\
	Edinburgh, UK EH9 3FD \\
\texttt{b.leimkuhler@ed.ac.uk} 
\And
Ren\'e Lohmann\\
	School of Mathematics\\
	University of Edinburgh\\
	Edinburgh, UK EH9 3FD \\
    \texttt{r.lohmann@ed.ac.uk}
	\And
Peter A. Whalley \\
	Department of Statistics\\
	ETH Z\"{u}rich\\
        Z\"{u}rich, Switzerland\\
	\texttt{pwhalley@ethz.ch}
}

\renewcommand{\shorttitle}{A Langevin sampling algorithm inspired by the Adam optimizer}

\hypersetup{
pdftitle={A Langevin sampling algorithm inspired by the Adam optimizer},
pdfauthor={Benedict Leimkuhler,Rene Lohmann, Peter Whalley},
}

\begin{document}
\maketitle
\begin{abstract}
We present a framework for adaptive-stepsize MCMC sampling based on time-rescaled Langevin dynamics, in which the stepsize variation is dynamically driven by an additional degree of freedom. 
The use of an auxiliary relaxation equation allows accumulation of a moving average of a local monitor function and provides for precise control of the timestep while circumventing the need to modify the drift term in the physical system. Our algorithm is straightforward to implement and can be readily combined with any off-the-peg fixed-stepsize Langevin integrator.  As a particular example, we consider control of the stepsize by monitoring the norm of the log-posterior gradient, which takes inspiration from the Adam optimizer, the stepsize being automatically reduced in  regions of steep change of the log posterior and increased on plateaus, improving numerical stability and convergence speed. As in Adam, the stepsize variation depends on the recent history of the monitor function, which enhances stability and improves accuracy compared to more immediate control approaches. We demonstrate the potential benefit of this method--both in accuracy and in stability--in numerical experiments including Neal's funnel and a Bayesian neural network for classification of  MNIST data. 
\end{abstract}

\keywords{Sampling methods, computational statistics, Adam, Langevin dynamics,  adaptive or variable stepsize, Bayesian sampling, neural network }

\section{Introduction}
Monte Carlo sampling schemes are ubiquitous in modern-day science, engineering, and finance. They are used to quantify risk and uncertainty, parameterize statistical models, calculate thermodynamic quantities in physical models, and explore protein configurational states. The aim in sampling is to generate independent and identically distributed (i.i.d.) realizations $x_i$ of a random variable $X \in \mathbb{R}^d$ distributed according to a given probability law, which we assume is defined in terms of a positive, smooth density $\pi$.   Such samples can then be used to estimate probabilities of certain events or expectations of functions, to calculate uncertainties, assess or compare models, or to explore optimality of parameterizations in machine learning applications. One of the most powerful and widely used categories of sampling methods are Markov Chain Monte Carlo schemes (MCMC \cite{MCMCmachinelearning,HMC_Neal}), which generate the samples using Markov Chains that are ergodic with respect to the target probability measure.   The original MCMC scheme is the Metropolis-Hastings (MH) algorithm \cite{1953JChPh..21.1087M,1970Bimka..57...97H} which uses random proposals in conjunction with an accept-reject procedure (the MH-criterion) to generate the Markov Chain.   The MH framework is very general and allows many alternatives for proposal generation.  The use of Metropolis correction may, however,  add substantial computational burden and the rejection steps can slow convergence to the target distribution. Even with a warm-start assumption the convergence rate is known to scale polynomially with respect to dimension \cite{roberts1998optimal,wu2022minimax}.



In this article we consider Langevin dynamics-based methods.  We assume the target probability density can be defined in terms of a $C^2(\mathbb{R}^d)$ energy function $U:  \mathbb{R}^d \to \mathbb{R}$ which grows sufficiently rapidly as $x\rightarrow \infty$ so that 
\[
\pi(x) =  Z^{-1}\exp(-\beta U(x))
\]
is  Lebesgue integrable on $\mathbb{R}^d$.  Here $\beta$ is the reciprocal temperature and $Z \equiv \int_{\mathbb{R}^d} \exp(-\beta U(x)) dx$ is a normalizing constant so that $\pi$ integrates to one.  In the ``overdamped'' form of Langevin dynamics the It\^{o} stochastic differential equation
\begin{equation} \label{eq:odl}
{\rm d} x = -\nabla U(x) {\rm d}t + \sqrt{2 \beta^{-1}} {\rm d}W_t
\end{equation}
is used to generate the paths, where $(W_{t})_{t \geq 0}$ is a $d$-dimensional standard Brownian motion.  In practice, it is found that introducing a momentum vector $p$ can enhance the efficiency of the sampling process and the ``underdamped'' Langevin dynamics system is often used instead:
\begin{eqnarray} \label{eq:udl_1}
{\rm d} x & = & p {\rm d} t,\\
{\rm d} p & = & -\nabla U(x) {\rm d}t - \gamma p {\rm d}t + \sqrt{2\gamma \beta^{-1}} {\rm d}W_t. \label{eq:udl_2}
\end{eqnarray}

Convergence analysis of \eqref{eq:odl} and \eqref{eq:udl_1}-\eqref{eq:udl_2} is well understood in the continuous setting \cite{pavliotis2014stochastic}, but these continuous systems need to be replaced by discrete processes in practical applications. For this purpose a numerical discretization is introduced.  For example, \eqref{eq:odl} can be discretized using the Euler-Maruyama method:
\[
x_{n+1} = x_n - \Delta t \nabla U(x_n) + \sqrt{2\Delta t \beta^{-1}} \xi_{n+1},
\]
where $\Delta t > 0$ is the stepsize and $\xi_{n+1} \sim \mathcal{N}(0,I_{d})$.
Similar discretizations may be introduced to solve \eqref{eq:udl_1}-\eqref{eq:udl_2}.  For example one method maps
$(x_n ,p_n)$ to $(x_{n+1},p_{n+1})$ by
\begin{eqnarray*}
p_{n+1} & = &  cp_{n} - \Delta t \nabla U(x_n) + \sqrt{(1-c^2)\beta^{-1}}\xi_{n+1},\\
x_{n+1} & = & x_n + \Delta t p_n,
\end{eqnarray*}
where $c=\exp(-\Delta t \gamma)$ and $\xi_{n+1} \sim \mathcal{N}(0,I_{d})$; this is referred to as the ``OBA'' method using the naming convention first introduced in \cite{LeMa2013}.
While the momenta must be carried forward during computations, only the  sequence of $x$ variables needs to be stored.  

By eliminating the stop-and-go aspect of MH methods (i.e., rejected steps), SDE schemes promise faster convergence, but they have some important drawbacks. First, a numerical method such as those mentioned above introduces bias with respect to the target distribution.  Effectively we sample a perturbed distribution with density $\hat{\pi}^{\Delta t}$ replacing $\pi$.  The bias can be controlled by choosing a sufficiently small discretization stepsize $\Delta t$, where the quality of the numerical integrator governs the size of the bias and its stepsize-dependent scaling \cite{LeMaSt2016}.  Of course, the stepsize cannot be reduced arbitrarily, as smaller steps lead to more strongly correlated successive states. This, in turn, requires longer trajectories to achieve the same level of exploration, increasing the computational cost.

The second well-known drawback of LD methods is that such methods tend to suffer from numerical instabilities in cases where $\nabla U$ has a large Lipschitz constant, forcing the stepsize to be reduced.  The stepsize restriction will be governed by the largest curvature; in the case of multimodal target distributions each basin may have a different Hessian eigenstructure and thus introduce different stability constraints.  As the sampling path ventures from the vicinity of one minimum to the vicinity of another, the stepsize restrictions change.  In common practice a single fixed stepsize is used which must be chosen  small enough to mitigate these issues. This comes at the cost of requiring a larger number of integrator steps to generate a new (sufficiently decorrelated) sample, thus increasing the computational cost in comparison to what would seem intuitive. 

In high-dimensional sampling applications, individual integrator steps may be extremely expensive due to the costly evaluation of $\nabla U$.
Nowhere is this more apparent than in machine learning, where a reduction of computing cost per iteration is typically achieved by replacing the evaluation of the true gradient $\nabla U$ with a cheaper stochastic approximation (usually realized by data subsampling, see e.g., \cite{RobinsMonro,SGLD,SGHMC,DingPaper,shaw2025randomised}). This, however, adds additional perturbations to the dynamics which can again be controlled by decreasing the stepsize (see e.g., \cite{BenAmosCovControlled} for a simple model of how gradient noise relates to stepsize). 
Thus, substantial effort has been made to design novel LD-based MCMC methods with improved accuracy, stability, or efficiency, see, for example \cite{BBK,quasisymplecticMelchionna,BuPa2007,adaptiveLangevin,LeMa2013,SymmetricMiniBatchSplitting,SGLD,SGHMC,DingPaper,DaRi2020,shaw2025random, jumps2025}. It should be mentioned that there are also dynamics-based MCMC schemes that come with an MH-correction step, such as Hamiltonian Monte Carlo (HMC) and relatives \cite{HMC_original,HMC_Neal,HOROWITZgHMC,randomHMC,RiemannHMC,MALA,MALT}. However, due to the expensive evaluation of $U$ and potentially low acceptance rates, they are often deemed too inefficient for large scale applications \cite{SGLD,SGHMC,betancourt_HMC_subsampling}.  Efforts to address the issue \cite{AMAGOLD,minibatchMH} may sacrifice  stability or robustness compared to standard procedures and have not been widely adopted.  This article proposes an unadjusted LD-based sampling scheme and is thus in line with the philosophy of foregoing the MH-criterion.

One approach to tackle some of the mentioned issues comes in the form of adaptive stepsize methods, which are widely used in integration of deterministic systems, e.g., $\frac{\diff }{dt} x = \phi(x) $.  Variable stepsize procedures try to estimate the local discretization error being introduced at each step and adjust the stepsize down or up to maintain a certain prescribed local error tolerance (see, e.g., Chapter 3 of \cite{Lambert}).  The error estimate may be based on finite difference approximation, or extrapolation, or the use of  specialized `embedded' integration schemes (e.g., Runge-Kutta Fehlberg methods \cite{RKF}).  These techniques have also been adapted for weak approximation (sampling) using stochastic differential equations \cite{Ro2004,Va2010,Sz2001}. Implementations of this approach are in widespread use in modern simulation software. 
The molecular dynamics package OpenMM \cite{OpenMM} uses a variable stepsize scheme in which local errors are estimated on-the-fly via simple Euler discretization (or Euler-Maruyama discretization for their LD integrator; see \cite{OpenMM}). Additionally, there has been recent developments of adaptive stepsize methods, which allow for high strong order approximation in \cite{jelinvcivc2024single} and \cite{foster2023convergence}. We also mention the development of recent methodology for local adaptation of stepsize and other parameters within Metropolis-based algorithms, for example, HMC and more general PDMP-based algorithms (see \cite{bou2024incorporating,bou2024gist,chevallier2025towards}).

An alternative approach to variable stepsize changes the stepsize through Sundman time transformation:
\begin{equation}\label{eq:Sundman_transform}
\frac{\diff x}{\diff \tau} = R(x)\phi(x), \quad \frac{\diff t}{\diff \tau}=R(x), 
\end{equation}
with a scalar-valued function $R:\mathbb{R}^{d} \to \mathbb{R}$, which is uniformly bounded such that $0<m<R(x)<M<\infty$ for all $x \in \mathbb{R}^{d}$.
These types of transformations are commonly used in classical mechanics (see the recent work \cite{Ca2022} and the references therein, as well as their illustrative examples). The Sundman transform can be used to turn fixed-stepsize numerical integrators into adaptive stepsize schemes improving integration stability and efficiency.  Specifically, one may discretize the equation \eqref{eq:Sundman_transform} using a fixed step $\Delta \tau$ in the `fictive' time variable $\tau$, then interpret this as equivalent to variable steps in `real' time according to (at timestep $n$)
\[
\Delta t_n \approx R(x_n) \Delta \tau.
\]
Recently, the idea was applied to Langevin dynamics in \cite{alix} and more general Markov processes \cite{bertazzi2025sampling}.  In \cite{alix}, a suggested transform kernel was given by $R(x)=\tilde{R}\left(\|\nabla U(x)\|^{-1}\right)$, with a boundedness-ensuring function $\tilde{R}: \mathbb R_{+} \to \left[m,M \right] $ with $\lim_{u\to \infty}\tilde{R}(u)=m<M=\lim_{u\to 0}\tilde{R}(u)$. As in the ODE case, the corresponding numerical schemes accomplish enhanced stability and efficiency on the one- and two-dimensional test examples considered. The framework we present in this article builds on this idea, but we introduce an alternative mechanism for stepsize adaptation. Rather than using the current state of the variable of interest $x$ to adjust the stepsize, we introduce an auxiliary variable $\zeta \in \mathbb{R}$, evolving via a suitably chosen dynamical equation 
\begin{equation}
\frac{\diff}{\diff \tau} \zeta = f(x,p,\zeta)
\label{eq:zeta_dynamics}.
\end{equation}
We focus here on the following natural choice:
\[
f(x,p,\zeta) = -\alpha \zeta +  g(x,p),
\]
with parameter $\alpha>0$ and a monitor function $g$. This type of dynamics effectively computes a moving average of $g(x,p)$ over the recent history, with the driving function $g$ determining on what basis the stepsize is to be changed. We then apply a Sundman transformation which is expressed as a function of $\zeta$
\[
{\rm d}t = \psi(\zeta){\rm d} \tau.
\]  
As the time-rescaling alters the rate at which samples are acquired, these samples cannot be used directly for computing Gibbs-Boltzmann averages, but as we shall see it is straightforward to reweight the data in order to calculate any such quantities.  Employing a bounded Sundman transformation $\psi$ controls the stability of the reweighting process. 

For both the introduction of the general dynamics \eqref{eq:zeta_dynamics} and the choice of $g$, we draw inspiration from optimization, where numerical challenges arise that mirror those associated to MCMC samplers. 
A classic, general purpose optimization method is the stochastic gradient descent (SGD) method of Robins and Monro \cite{RobinsMonro}, and it remains a popular choice due to its simplicity and robustness in practice. In recent years, there has been extensive work on adaptive-stepsize variants of SGD \cite{AdaGrad,RMSProp,ADADELTA,Adam,NAdam,AMSGrad,AdaFom}. Most of these schemes are modifications or extensions of the methods AdaGrad \cite{AdaGrad,AdaGrad2} or RMSProp \cite{RMSProp}. 
By far the most widely known and used of these is the Adam optimizer (short for adaptive moment estimation \cite{Adam}), which uses estimates of mean and variance of the gradients (computed as moving averages) to adjust the stepsize, reducing it in relation to the steepness of the landscape. Adam is known to improve efficiency in certain cases \cite{Adam}, often reaching the minimum in many fewer steps than SGD, and  is popular for model training in various settings such as natural language processing \cite{Adam_example_languagemodel,Adam_example_languagemodel2,Adam_example_languagemodel3}, Bayesian neural networks \cite{Adam_example_BNN,Adam_example_BNN2}, and computer vision \cite{Adam_example_computervision,Adam_example_computervision2,Adam_example_computervision3}. 

By combining the methodology of time-transformed SDEs with Adam's approach to adjust the stepsize based on a moving average over recent history, we derive a flexible framework for adaptive-stepsize sampling, which we call ''SamAdams'' (\underline{sam}pling with \underline{ada}ptive \underline{m}oderated \underline{s}tepsize). The procedure can
be combined with state-of-the art integrators for Langevin dynamics. In particular, it is possible to transform the splitting integrators mentioned above into adaptive algorithms. Adaptivity improves numerical stability compared to relying on a constant stepsize, as demonstrated in Fig. \ref{fig:star_trajectory} where a fixed stepsize Langevin integrator is compared with its adaptive counterpart. 
\begin{figure}[ht!]
\begin{center}
\includegraphics[width=1\textwidth]{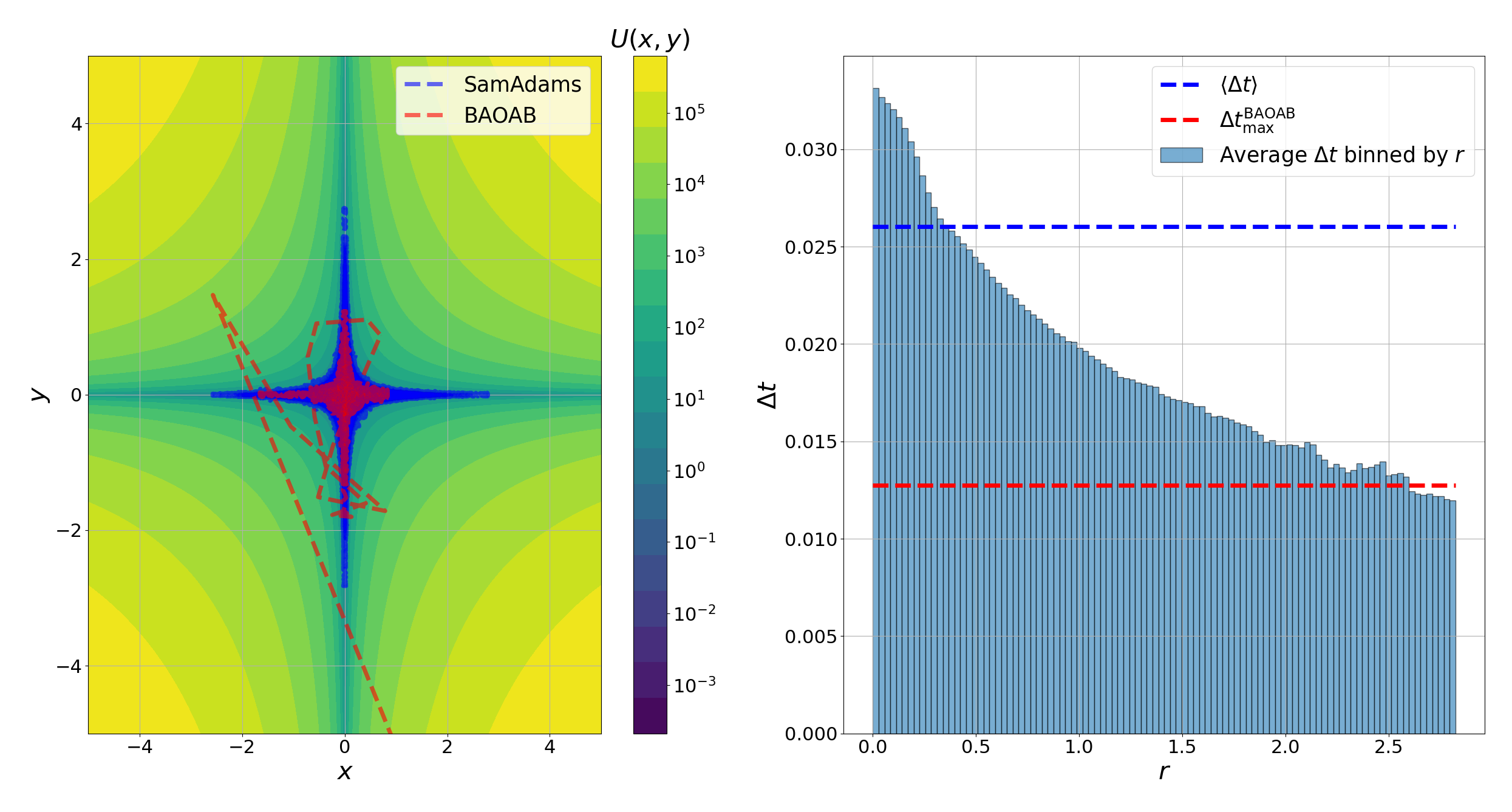}
\caption{Sampling trajectories of a constant-stepsize integrator (BAOAB) and our adaptive-stepsize scheme (SamAdams) on a star-shaped landscape $U(x,y)= x^2 + 1000x^2y^2 + y^2$. \textbf{Left:} Potential $U(x,y)$ with trajectories. BAOAB was run at the mean stepsize used by SamAdams (obtained by averaging over all iterations). \textbf{Right:} Stepsize values $\Delta t$ used by SamAdams are binned by distance to the origin $r=\sqrt{x^2+y^2}$ together with the mean stepsize (blue dashed line) and the maximum stable stepsize for BAOAB (red dashed line).  SamAdams uses a small stepsize only at the outer points of the stable domain.}
\label{fig:star_trajectory}
\end{center}
\end{figure}
While the fixed stepsize method becomes unstable in the tips of the star-like potential, SamAdams remains stable by automatically reducing its stepsize  in those regions. As evident by the $\Delta t$ histograms, the smallest $\Delta t$ values are adopted in the tips of the star. There, they are similar in size to the stability threshold of BAOAB, which neatly confirms the understanding that the stability of a constant-stepsize scheme is determined by those landscape regions of largest steepness and curvature where the forces and force fluctuations are largest.  Since most of the probability mass is located close to the origin where the experienced forces are small, the adaptive scheme is able to use larger stepsizes during most of the simulation, enabling a larger mean stepsize $\langle \Delta t \rangle$, with consequent increase in computational efficiency. The ability to use larger stepsizes while maintaining sampling quality would greatly benefit those sampling applications where the sampling error is dominated by a lack of exploration of the loss landscape $U$ (rather than other sources of error such as discretization bias or model specification). Prominent examples are large-scale Bayesian neural networks \cite{BNNs1,BNNs2,BNNs3,BNNs4,BNNs5} or molecular dynamics simulations \cite{MD1,MD2,MD3,MD4,OpenMM}.

For the purposes of disambiguation we mention that Adam and variants further improve optimizer efficiency by independently rescaling the coordinates of the system through the mechanism of individual timesteps.  We don't address this important aspect here but instead focus on the use of a moving average to stabilize the timestep selection in a pure Langevin dynamics framework. Some sampling methods that incorporate anisotropic coordinate transformation in order to enhance performance are RMHMC \cite{RiemannHMC}, the ensemble quasi-Newton method \cite{EQN}, and the recently proposed AdamMCMC method \cite{Bi2024_adammcmc}.

The rest of this article is structured as follows. Section \ref{sec:SundmanSDEs} discusses the Sundman-transformed SDE our new framework is built on and a reweighting scheme for physical observables from trajectories that evolve in rescaled time. The following section is addressed to numerical discretization of the equations of motion.  In Section \ref{sec: SamAdams} we discuss designing a Sundman kernel so that the transformed SDE adopts a stepsize adaptation resembling the device in Adam.  Finally, Section \ref{sec:numerics} contains our numerical results.

\section{Sundman-transformed SDEs and Averaging}\label{sec:SundmanSDEs}

The notion of rescaling of time is a familiar one in studies of gravitational
$N$-body problems, where it is used in analytical as well as numerical
treatments. Let an autonomous ordinary differential equation be given
of the form 
\begin{equation}
\frac{\diff x_{t}}{\diff t} = \phi(x_t).\label{eq:ode}
\end{equation}
We use the notation $x_{t}:=x(t)$ to compactly indicate the independent variable. In particular, the classical Sundman transformation\cite{Sundman1912} replaces $t$ by
a new variables $\tau$ which is defined by an ordinary differential
equation of the form 
\[
\frac{\diff t}{\diff \tau}=\psi(x_{\tau}),
\]
with $x_{\tau}:=x_{t(\tau)}$, so that \eqref{eq:ode} becomes, using chain rule, 
\begin{align*} 
\frac{\diff x_{\tau}}{\diff \tau }  = & \frac{\diff x_{t}}{\diff t} \frac{\diff t}{\diff \tau } = \psi(x_{\tau}) \phi(x_{\tau}).
\end{align*} 
This type of time-rescaling can be used to facilitate numerical integration
of ODEs. As explained in \cite{Stoffer1995,Ben_adaptive_verlet},
if the Sundman transformation is suitably chosen to normalize the system or at
least lessen the variation in the magnitude of $\phi$ and its derivatives, the rescaled
system can often be integrated using fixed stepsize with the result
that errors or instabilities associated with directly integrating
\eqref{eq:ode} are eliminated or reduced.

Time-rescaling changes the effective frequencies of a system with
oscillatory components. For example if we introduce a simple constant
Sundman transformation ${\rm d}t/{\rm d}\tau=a$ into a 1D harmonic
oscillator with frequency $\omega$, $\diff q_{t}/\diff t=p_t$; $\diff p_{t}/ \diff t=-\omega^{2}q_t$,
we obtain the system $\diff q_{\tau}/\diff \tau=ap_{\tau}$; $\diff p_{\tau}/\diff \tau=-a\omega^{2}q_{\tau}$. This
new system has frequency $a\omega$. The same effect can be
obtained by rescaling only the momentum equation or only the position
equation. Since the stable timestep for integration of oscillatory
dynamics typically depends on the fast frequencies we can potentially improve stability efficiently by adjusting the time-rescaling dynamically.
In nonlinear systems and systems with multiple oscillatory modes,
a configuration-dependent Sundman transformation can be used to modify
the dynamics so that all components are propagated with an effective
smaller stepsize in regions of high oscillation. For example, the Adaptive
Verlet method \cite{Ben_adaptive_verlet} uses just such a time-rescaling,
treating the local time-rescaling factor  as an additional
dependent variable of the system and evolving this in various ways
to enhance numerical performance.

A simpler form of adaptivity that, in our experience, often works well, is to use the Sundman transformation to
adjust the stepsize at the beginning of each step using $\Delta t_{n}:=\Delta\tau R(x_{n})$.
$\Delta\tau$ represents a fixed stepsize for the rescaled system
and $R(x)$ indicates how the step should be adjusted depending on current configuration $x$. As in \cite{alix}, by applying the Sundman transform to an SDE and then discretizing with fixed stepsize,  we can obtain an adaptive stepsize sampling method in the original time variable.


Using the Sundman transform $\diff t / \diff \tau = R(x_{\tau})$, we can write down 
a  time-rescaled version of \eqref{eq:udl_1}-\eqref{eq:udl_2}: 
\begin{align}
\diff x_{\tau}    &= R(x_{\tau}) p_{\tau} \diff \tau,  \\
\diff p_{\tau}    &= - R(x_{\tau})\nabla U(x_{\tau})\diff \tau - \gamma  R(x_{\tau}) p_{\tau} \diff \tau + \sqrt{\frac{2\gamma R(x_{\tau})}{\beta}}\diff W_{\tau}.  
\end{align} 
In \cite{alix}, it was shown that the canonical distribution $\rho_{\beta}$ is no longer invariant under the process. To correct the measure that article introduced a drift correction 
\[
\Delta(x) := \frac{\gamma}{\beta} \nabla R(x).
\]
This method adds some complexity in the form of an additional gradient computation.

Here, we propose a different way of introducing the timestep adaptation. Rather than letting the transform function directly depend on the configuration $x$, we introduce an artificial dynamical control variable $\zeta \in \mathbb{R}$, and a Sundman rescaling function $\psi$ such that the time adaption is governed by $\diff t / \diff \tau = \psi(\zeta_{\tau})$ with $\diff \zeta_{\tau}=f(x_{\tau},p_{\tau},\zeta_{\tau})\diff \tau$ for some scalar-valued function $f$ on the augmented phase space variables $(x,p,\zeta)$. 

It makes intuitive sense to choose the function $f$ as the sum of dissipative and driving terms, i.e., 
\begin{align}\label{eq: zeta_dynamics}
f(x,p,\zeta) = - \alpha \zeta +  g(x,p) ,
\end{align}
with hyperparameter $\alpha>0$ representing the {\em attack rate} of a relaxation process.  The {\em monitor} or {\em driving function} $g(x,p)$ can be any positive smooth function and ultimately decides on what basis a small time increment $\diff t$ is varied. The rescaling function $\psi$, over a crucial interval, mimics a reciprocal power of $\zeta$.  If the monitor function takes on larger values, $\zeta$ will tend to increase and the time interval $\diff t$ will shrink (implying smaller stepsizes during numerical integration). If the monitor function is small (meaning the solution is locally smooth and easy to integrate), the rescaling will lead to an increase in stepsize. This corresponds to the paradigm to take {\em as large a stepsize as possible but as small as necessary}. 
Making the stepsize adaptation depend on the recent history of a driving function is one reason why the Adam optimizer became so successful.

Since one can employ many different monitor functions $g(x,p)$ and Sundman transform kernels $\psi(\zeta)$, we obtain a flexible framework for adaptive stepsize sampling, given by\footnote{A version of this system for generic SDEs is given in Appendix \ref{sup:sec:general_SamAdams}.}
\begin{align}
\diff x_{\tau}    &= \psi(\zeta_{\tau}) p_{\tau} \diff \tau,   \label{eq: full_framework1}  \\  
\diff p_{\tau}    &= - \psi(\zeta_{\tau})\nabla U(x_{\tau})\diff \tau - \gamma  \psi(\zeta_{\tau}) p_{\tau} \diff \tau + \sqrt{2\gamma \beta^{-1}\psi(\zeta_{\tau}) }\diff W_{\tau}, \label{eq:p_transformed}\\ 
\diff \zeta_{\tau}&= - \alpha \zeta_{\tau}\diff \tau +  g(x_{\tau},p_{\tau}) \diff \tau,   \label{eq: zeta_eq}\\ 
\diff t           &= \psi(\zeta_{\tau})\diff \tau. \label{eq: full_framework2}   
\end{align} 
Note that in the Adam optimizer, the stepsizes are adapted based on moving averages (see Appendix \ref{sup:sec:adam_derivation} for a detailed discussion). The dynamics \eqref{eq: zeta_eq} also effectively computes a moving average of the driving function $g$ over the recent history, exponentially weighted with rate $\alpha$, which is then used to modify the stepsize $\Delta t$ (see also Appendix \ref{sup:sec:stepsize_variation}). Due to the stepsize adaptation through moving averages and the fact that, as we will show in Sec. \ref{sec: SamAdams}, one can design transform kernel $\psi$ and driving function $g$ in a way to resemble the Adam equations, the dynamics \eqref{eq: full_framework1}-\eqref{eq: full_framework2} can be understood as an Adam-inspired sampling framework. We name this framework \textbf{SamAdams} (sampling with adaptively moderated stepsizes). Although we will focus on particular choices for $\psi$ and $g$, it is important to emphasize that the framework is very general.  We next address the ergodicity of process \eqref{eq: full_framework1}-\eqref{eq: full_framework2}, how to obtain canonical averages from it, and how to simulate it in practice.

\subsection{Ergodic Averages}\label{sec:computing_averages}
Due to the time-rescaling and the additional $\zeta$-dynamics, the invariant measure of \eqref{eq: full_framework1}-\eqref{eq: full_framework2} no longer coincides with the invariant measure of underdamped Langevin dynamics \eqref{eq:udl_1}-\eqref{eq:udl_2}, i.e., the canonical measure $\pi_{\beta}$. 
However, under mild assumptions on $U$ and $\psi(u)$, we can correct this error by reweighting the samples with the corresponding values of the transform kernel $\psi(\zeta_{\tau})$. To see this, assume the dynamics \eqref{eq: full_framework1}-\eqref{eq: full_framework2} is ergodic with invariant measure $\Pi_{\tau}$. Assume we have samples $(x_{\tau_{i}},p_{\tau_{i}},\zeta_{\tau_{i}})$ from the solution of \eqref{eq: full_framework1}-\eqref{eq: full_framework2} at times $\tau_{i}:=i\Delta \tau$, $i\in\mathbb{N}$, for some stepsize $\Delta \tau>0$.  It follows that under mild assumptions on an observable $\phi:\mathbb{R}^{2d} \to \mathbb{R}$,
\begin{align*} 
\lim_{N \to \infty} \frac{\sum^{N}_{i=1}\phi(x_{\tau_i},p_{\tau_i})\psi(\zeta_{\tau_i})}{\sum^{N}_{i=1}\psi(\zeta_{\tau_i})} &= \lim_{N \to \infty}\frac{\frac{1}{N}\sum^{N}_{i=1}\phi(x_{\tau_i},p_{\tau_i})\psi(\zeta_{\tau_i})}{\frac{1}{N}\sum^{N}_{i=1}\psi(\zeta_{\tau_i})} \\ 
&= \frac{\mathbb{E}_{\Pi_{\tau}}(\phi \psi)}{\mathbb{E}_{\Pi_{\tau}}(\psi)} \\ 
&= \lim_{T \to \infty} \frac{\frac{1}{\tau(T)}\int^{\tau(T)}_{0}\phi(x_{\tau},p_{\tau})\psi(\zeta_{\tau})d\tau}{\frac{1}{\tau(T)}\int^{\tau(T)}_{0}\psi(\zeta_{\tau})d\tau}, \\ 
&= \lim_{T \to \infty} \frac{1}{T}\int^{T}_{0} \phi(x_{\tau(t)}, p_{\tau(t)})dt\\ 
&= \mathbb{E}_{\pi_{\beta}}(\phi), 
\end{align*} 
where the second and third lines follow from ergodicity and the fourth line from identifying $\psi\diff\tau=\diff t$. 
We remark that the process \eqref{eq: full_framework1}-\eqref{eq: full_framework2} can be shown to be ergodic under minimal
assumptions as in \cite{sachs2017langevin} which
treats a more general case. Similar techniques can be found in \cite{mattingly2002ergodicity,phillips2024numerics}. In particular, under sufficient smoothness, the assumption that the time-rescaling $\psi$ is uniformly bounded from below and above and the force is convex outside a ball, one can show ergodicity of the time-rescaled process \eqref{eq: full_framework1}-\eqref{eq: full_framework2} by considering the family of Lyapunov functions considered in \cite{mattingly2002ergodicity}. Resulting in an ergodicity result with a convergence rate which depends on the uniform lower bound on $\psi$. More sophisticated techniques would be required to show an improved convergence rate of the Adam sampler, akin to the adaptive stepsize results available in the optimization literature (see, for example, \cite{defossez-2022} for Adam or Adagrad). 

We have thus shown how to obtain canonical averages as time averages over a single trajectory. In practice, when employing constant-stepsize MCMC to sample $\pi_{\beta}$, rather than obtaining averages through a time average along a single long trajectory, one often draws multiple trajectories in parallel (allowing for parallel computation) and approximates 
\begin{equation}
\mathbb{E}_{\pi_{\beta}}(\phi(x_{t},p_{t}))\approx\frac{1}{N}\sum_{i=1}^N\phi(x^i_{t},p^i_{t}),  
\end{equation}
for $N$ independent trajectories and $t$ large enough to have $\text{Law}(x_t,p_t)\approx \pi_{\beta}$. The superscript refers to the trajectory index. The same is possible for the time-rescaled dynamics \eqref{eq: full_framework1}-\eqref{eq: full_framework2}. For sufficiently large $\tau(t)$, we have that $\text{Law}(x_{\tau}, p_{\tau},\zeta_{\tau})\approx\Pi_{\tau}$. 
In this case, we have for $N$ independent trajectories, $N$ large enough,
\begin{equation}
\frac{\sum_{i=1}^{N}\phi(x^i_{\tau}, p^i_{\tau})\psi(\zeta^i_{\tau})}{\sum_{i=1}^{N}\psi(\zeta^i_{\tau})}\approx\frac{\mathbb{E}_{\Pi_{\tau}}(\phi \psi)}{\mathbb{E}_{\Pi_{\tau}}(
\psi
)}=\mathbb{E}_{\pi}(\phi).
\end{equation} 
Thus, when simulating the time-rescaled dynamics, one can use both time- and trajectory- averages to approximate canonical averages just like in constant-stepsize MCMC. Obtaining canonical averages can thus be illustrated by the diagram below. 
\begin{figure}[h]
\begin{center}
\includegraphics[width=5in]{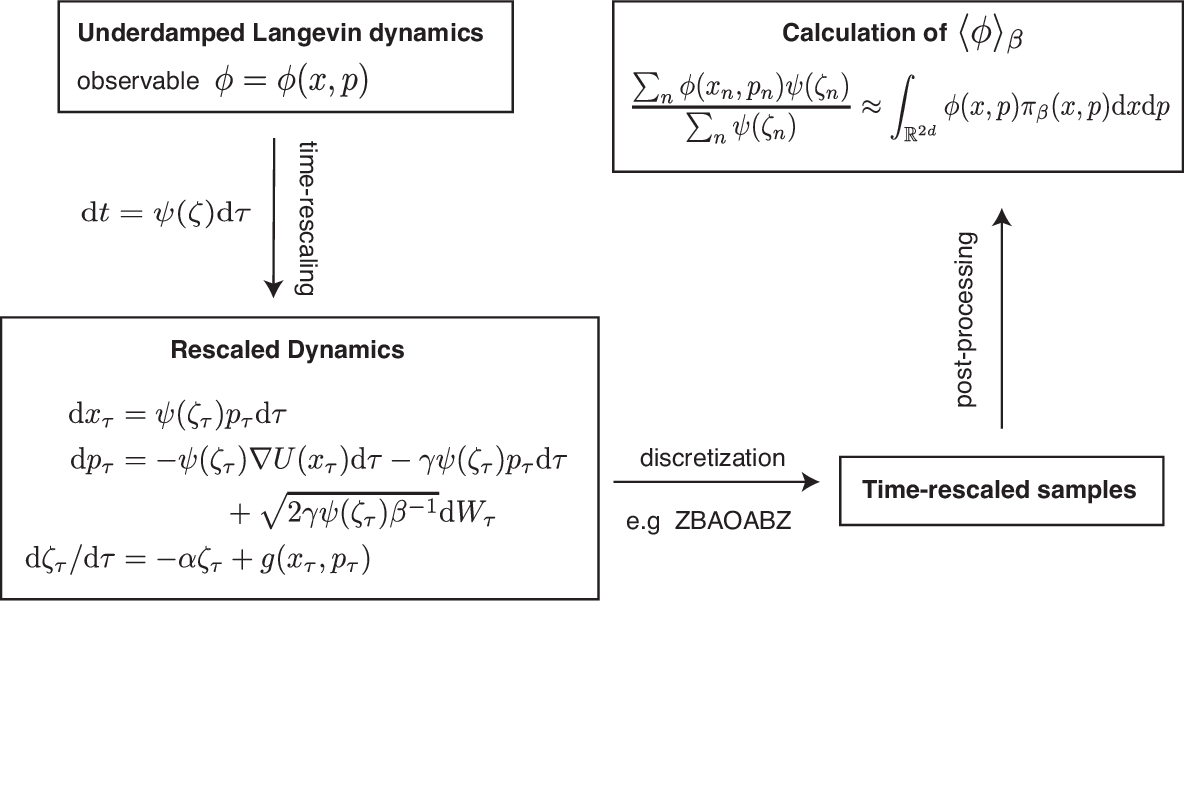}
    \caption{SamAdams sampling procedure.}
\end{center}
\end{figure}
The subscript $n$ in the upper right-hand corner of the diagram can either denote different points in time along a single trajectory (time-average) or a trajectory index when averaging over samples of different trajectories taken at the same time (trajectory-average). 

\section{Numerical Integration}\label{sec: integration}
In order to simulate SamAdams, we need to discretize the continuous-time process \eqref{eq: full_framework1}
to \eqref{eq: full_framework2} with a suitable numerical integrator. For convenience we adopt the framework of \cite{Sk2002,BuPa2007,Bou-RabeeOwhadi2010,LeMa2013} in which a symplectic splitting of the Hamiltonian part of the underdamped Langevin system is composed with a map that exactly preserves the momentum distribution.   A frequent
choice for the stochastic part is
\[
\Phi_{\Delta t}^{{\rm O}}(x,p)=\bigg(x,\ \exp(-\gamma\Delta t)p+\sqrt{(1-\exp(-2\gamma\Delta t))\beta^{-1}}\xi\bigg),
\]
where $\xi$ is a random vector with each component independently drawn
from ${\mathcal N}(0,1)$. Since $\Phi_{\Delta t}^{{\rm O}}$ preserves
the momentum distribution regardless of $\Delta t$, it is reasonable
to introduce this in a way which parallels our previous derivation.
Resolving the Hamiltonian part can be done by splitting the dynamics
into a drift at constant momentum $\diff x_{t}/\diff t=p_t$; $\diff p_{t} / \diff t=0$ and a momentum
kick $\diff x_{t} /\diff t=0$; $\diff p_{t}/\diff t=-\nabla U(x_t)$ using the two propagators 
\begin{align*}
\Phi_{\Delta t}^{{\rm A}}(x,p) & =(x+\Delta tp,\ p),  \\
\Phi_{\Delta t}^{{\rm B}}(x,p) & =(x,\ p-\Delta t\nabla U(x)).
\end{align*}
Different integrators for Langevin dynamics can be obtained by composing
the A-, B-, and O-maps in different ways, e.g., giving symmetric OBABO \cite{BuPa2007} and BAOAB\cite{LeMa2013} schemes; for example $\Phi^{\text{OBABO}}_{\Delta t}:=\Phi_{\Delta t/2}^{{\rm O}}\circ\Phi_{\Delta t/2}^{{\rm B}}\circ\Phi_{\Delta t}^{{\rm A}}\circ\Phi_{\Delta t/2}^{{\rm B}}\circ\Phi_{\Delta t/2}^{{\rm O}}$.
In \cite{LeMa2013,LeMaSt2016} these integrators have been shown to provide second order approximation of averages with respect to the Gibbs-Boltzman  distribution.  Because of the special form of these composition methods, only a single evaluation of the expensive gradient term is needed at each step, improving efficiency.  Alternatives such as BABO (half step of B, whole step of A, half step of B, whole step of O) sacrifice the symmetry but maintain the same second order accuracy with respect to the invariant measure and similarly require only a single gradient per step.   Another good prospective Langevin scheme is the symmetric UBU discretization \cite{ubu}.

What remains is to incorporate the stepsize adaptation described by \eqref{eq: zeta_eq} and \eqref{eq: full_framework2}.
In a splitting approach, the auxiliary variable $\zeta$ evolves using \eqref{eq: zeta_eq} with fixed $x$ and $p$, thus the solution is just
\begin{equation}
\zeta(\tau)=e^{-\alpha\tau}\zeta(0)+g(x,p)\int_{0}^{\tau}e^{-\alpha(\tau-s)}{\rm d}s.\label{eq:zeta_solution}
\end{equation}
Hence we may introduce an additional map 
\begin{equation}\label{eq: Z-map}
\Phi^{\rm Z}(x,p,\zeta)=\left (x,\ p,\ 
\exp(-\Delta\tau\alpha)\zeta+
\frac{1}{\alpha}\left ( 
1-e^{-\Delta\tau\alpha}
\right )
g(x,p) \right ).
\end{equation}
Once $\zeta$ has been calculated, the new stepsize $\Delta t$ is  updated via $\Delta t:=\psi(\zeta)\Delta\tau$. In Sec. \ref{sec: SamAdams} we introduce two suitable choices for the transform kernel $\psi$. 

With $x,p$ held fixed and thus $g(x,p)\equiv g$, we may rewrite the update for $\zeta$ as follows, 
\[
\zeta_{n+1} = \hat{\Phi}^{\rm Z}(x,p,\zeta_n) \equiv \rho \zeta_{n} + \alpha^{-1}(1-\rho) g(x,p),
\]
where $\rho = \exp(-\alpha \Delta \tau)$, which is fixed along a sampling path.

The choice  $\rho\approx 1$ maintains a strong dependence on the stepsize history, whereas $\rho\approx 0$ represents a rapid damping.    We also define
\[
\hat{\Phi}^{Z}_a(x,p,\zeta)=\rho^a \zeta+ \alpha^{-1}(1-\rho^a)g(x,p),
\]
to allow for taking partial steps of the $\zeta$ flow.

\subsection{The Algorithm}\label{sec:algorithm}
Using the splitting components from the previous section, we can write down algorithms to simulate SamAdams. With slight notational ambiguity, we use integer subscripts to denote timesteps in artificial time, thus $x_n \approx x_{\tau}(\tau_n)$, etc. Algorithm \ref{alg:AdamSampler} presents the ZBAOABZ method that we have used in all our numerical experiments.
\begin{algorithm}[ht]
\caption{SamAdams (ZBAOABZ)}
\label{alg:AdamSampler} 
\begin{algorithmic} 
\State{Given: $\hat{\Phi}^{\text{BAOAB}}_{\Delta t}$ the (fixed-stepsize) BAOAB integrator, ${\psi}$ a suitable Sundman transformation.}
\State{Given: parameters $n_{\rm max}$, $\Delta\tau$, and $n_{\text{meas}}$.}
\State{Given: initial conditions $x_{0}$, $p_{0}$, $\zeta_{0}$. Set $\mu_0 = {\psi}(\zeta_0)$.}
 \For {$n=0:n_{{\rm max}}$} 
 \vspace{5pt}
 \If{$n\ \text{mod}\ n_{\text{meas}}==0$} \State{Collect sample
$(x_{n}, p_n, 
\mu_n)$}.   \Comment{$\mu_n$ needed for reweighting.}
\EndIf 
\vspace{5pt}
 \State $\zeta_{n+\frac{1}{2}}:=\hat{\Phi}^Z_{1/2}(x_{n}, p_n,\zeta_{n})$.
 \Comment{Evolve $\zeta$.}
\vspace{5pt} 
\State $\Delta t_{n+1}:={\psi}(\zeta_{n+\frac{1}{2}})\Delta \tau$. \Comment{Modify stepsize.}
\vspace{5pt}
\State $({x}_{n+1},{p}_{n+1}):=\hat{\Phi}^{\text{BAOAB}}_{\Delta t_{n+1}}(x_{n},p_{n})$. \Comment{BAOAB step with stepsize $\Delta t_{n+1}$.}
\vspace{5pt}
\State $\zeta_{n+1}:=\hat{\Phi}^{Z}_{1/2}(x_{n+1},p_{n+1},\zeta_{n+\frac{1}{2}})$. 
\Comment{Evolve $\zeta$.}
\vspace{5pt}
\State $\mu_{n+1}:={\psi}(\zeta_{n+1})$. \Comment{Calculate weight.}
\vspace{5pt}
\EndFor 
\end{algorithmic} 
\end{algorithm}
The implementation of the underlying Langevin integrator, here BAOAB (denoted by $\hat{\Phi}^{\text{BAOAB}}_{\Delta t}$ in Alg. \ref{alg:AdamSampler}), is as in the fixed-stepsize setting.  There is no additional significant cost overhead above the usual cost of fixed stepsize integration, since the force evaluation needed to update Z is already performed in the BAOAB iteration. However, one could use other fixed-stepsize Langevin dynamics integrators as well.

Palindromic letter sequences like ZBAOABZ indicate a symmetric introduction of $\zeta$ updates around the fixed-stepsize method BAOAB.
Compared to the alternative ZBAOAB, the symmetric inclusion of the two Z-half-steps has a very significant consequence, the effect of which can be seen in experiments: the weights $\mu_n$ that are output with the samples are more accurate than if we for example relied on the stepsize that was used to advance the previous time-step, see Appendix \ref{sup:sec:symm_Z}.  

We note that initialization of Algorithm \ref{alg:AdamSampler} requires a choice for $\zeta_0$. There are different approaches to this depending on the particular choices for $\psi$ and $g$, and we discuss reasonable options in Sec. \ref{sec:numerics}. Independently of $\zeta_0$, it is expected that there will be some equilibration of the timestep control in the first few steps of integration.

A generic version of Algorithm \ref{alg:AdamSampler} for overdamped (first order) dynamics is given Appendix \ref{sup:sec:general_SamAdams}.

\subsection{Convergence and Order of Accuracy}
The order of accuracy of a numerical method for SDEs can be studied in different contexts.  In some disciplines, the quantity of interest is the strong accuracy, defined as the fidelity of the numerical solution to the solution of the SDE associated to some  realization of the Wiener process.  For sampling, the more relevant quantity is the accuracy of the approximation of the distribution generated by SDE solutions at a specified time (weak accuracy). An article of G. Vilmart \cite[Proposition 6.1]{Vilmart2014} provides a rigorous proof of the weak convergence of general splitting methods, which involve components which can be integrated exactly (in the weak sense) but which may contain multiplicative noise, as is the case for ZBAOABZ and other such integrators for SamAdams. Specifically, we can state a theorem regarding finite time approximation based on this work (see Theorem \ref{th:conv}).

\begin{theorem} \label{th:conv}
   Consider the system (\ref{eq: full_framework1})-(\ref{eq: full_framework2}) and assume that $\psi$, $\sqrt{\psi}$, $\nabla U$, $g$ are $C^{6}$ functions with all partial derivatives bounded, further assume that $r < \psi < M$ for some $r,  M > 0$. Then consider a splitting of the form ZBAOABZ and generating a sequence of points $(x_n,p_n,\mu_n)$ based on Algorithm \ref{alg:AdamSampler}.  Let $\chi(\cdot) \in C^{6}$ be an observable function of $x,p$ where all partial derivatives have polynomial growth, then for all $\Delta \tau k \leq T$
   \begin{align}
       \left| \frac{\mathbb{E}(\chi(x_{n},p_{n})\mu_{n})}{\mathbb{E}(\mu_{n})} - \frac{\mathbb{E}(\chi(X(n\Delta\tau),P(n\Delta\tau))\psi(\zeta(n\Delta\tau))}{\mathbb{E}(\psi(\zeta(n\Delta\tau))}\right| \leq C\Delta \tau^{2},
   \end{align}
   where $C >0$ is independent of $\Delta \tau >0$ and $(X(\cdot),P(\cdot),\zeta(\cdot))$ is the solution to \eqref{eq: full_framework1}-\eqref{eq: zeta_eq}.   
\end{theorem}
\begin{proof}
We first remark that $\mathbb{E}(\mu_{n}) > r$ and $\mathbb{E}(\psi(\zeta(n\Delta\tau))) > r$ for all $n \in \mathbb{N}$ due to the uniform bound assumption on $\psi$. Let $x = \mathbb{E}(\chi(x_{n},p_{n})\mu_{n})$, $y = \mathbb{E}(\mu_{n})$, $w = \mathbb{E}(\chi(X(n\Delta\tau),P(n\Delta\tau))\psi(\zeta(n\Delta\tau))$ and $z = \mathbb{E}(\psi(\zeta(n\Delta\tau))$. Then we have
\begin{align*}
      \left|\frac{x}{y} - \frac{w}{z}\right| &= \left|\frac{(x-w)z + w(z-y)}{yz} \right| \leq \frac{|x-w||z| + |w||z-y|}{|y| |z|},
   \end{align*}
    and we have that $|x-w| \leq C\Delta \tau^{2}$ and $|z-y| \leq C\Delta \tau^{2}$ by \cite[Proposition 6.1]{Vilmart2014} applied to the Strang splitting between $Z$ and the symmetric kinetic Langevin dynamics integrator (for example BAOAB), then due to the uniform lower bound on $\psi$ we have the required result.
\end{proof}
\begin{remark}
    Theorem \ref{th:conv} also applies to for example ZUBUZ (or replacing the middle with any other weak order two integrator for underdamped Langevin dynamics, for example, OBABO or ABOBA).
\end{remark}
\begin{remark}
If, in addition, the process (\ref{eq: full_framework1})-(\ref{eq: full_framework2}) is ergodic, then 
$$\lim_{n\to\infty}\frac{\mathbb{E}(\chi(X(n\Delta\tau),P(n\Delta\tau))\psi(\zeta(n\Delta\tau))}{\mathbb{E}(\psi(\zeta(n\Delta\tau))} = \mathbb{E}_{\pi_{\beta}}(\chi)$$
following from Section \ref{sec:computing_averages}.
\end{remark}

The weak convergence can in principle be studied in the asymptotic sense as the time interval tends to infinity, i.e., we may consider the asymptotic evolution of the weak error. The challenge then is to establish the convergence rate in a suitable framework and to study the systematic bias that is introduced due to discretization. 
Theoretical study of the geometric convergence of Alg. \ref{alg:AdamSampler} and the order of accuracy of the stationary distribution of the numerical method will be explored in future work, using techniques that are by now well developed in the setting of Langevin dynamics \cite{Bou-RabeeOwhadi2010,LeMaSt2016,durmus2021uniform,DaRi2020,monmarche2021high,leimkuhler2023contraction}.  In Appendix \ref{sup:sec:integration_order} we examine the order of accuracy numerically. 

It should be noted that ZBAOABZ does not have the additional desirable properties of BAOAB such as its property of quartic accuracy for configuration variables at high friction and its exactness for Gaussian targets \cite{benMDbook},
but the primary motivation for  sampling methods is typically to provide stability and fast exploration of the state space. In settings where the energy landscape is complex, there can be considerably higher error due to lack of exploration than due to the bias arising from the numerical discretization. We illustrate this in the examples. 

\section{Adam-inspired Monitor Function and Sundman Transform} \label{sec: SamAdams}
The time-rescaled Langevin dynamics \eqref{eq: full_framework1}-\eqref{eq: full_framework2} yields a family of samplers, each member specified by a particular choice of driving function $g(x,p)$ and Sundman transform kernel $\psi(\zeta)$. The choice of the monitor function, the auxiliary dynamics and the restriction function collectively decide the performance of the method, but have no impact on its theoretical foundation.  In this section, we propose a choice that allows to adapt some of the advantages of the Adam optimizer to the realm of sampling. 

In \cite{da2020general} it was demonstrated that the
Adam optimizer can be interpreted as the Euler discretization of a
certain system of ODEs, namely
\begin{align} 
\diff x_{\tau} &= \frac{p_{\tau}}{\sqrt{\zeta_{\tau} + \epsilon}}\diff \tau,           \label{eq:adam_ode1} \\ 
\diff p_{\tau} &= - \nabla U(x_{\tau})\diff \tau - \gamma p_{\tau}\diff \tau,                    \label{eq:adam_ode2} \\ 
\diff \zeta_{\tau} &= [\nabla U(x_{\tau})]^2\diff \tau - \alpha \zeta_{\tau}\diff \tau,          \label{eq:adam_ode3} 
\end{align} 
with $x_{\tau},p_{\tau},\zeta_{\tau}\in\mathbb{R}^{d}$,
$\alpha>0$, $\gamma>0$, and the algebraic operations are to be understood elementwise. To see how a discretization of this ODE leads to the Adam optimizer, see Appendix \ref{sup:sec:adam_derivation} and \cite{KaterinaThesis}. 

Assuming $\beta^{-1}$=0 in the time-rescaled process \eqref{eq: full_framework1}-\eqref{eq: full_framework2}, one observes that the choices
\begin{align}
\psi(\zeta)     & :=\frac{1}{\sqrt{\zeta+\epsilon}},\label{eq:Adam_transfo1} \\
g(x,p)       & :=\|\nabla U(x) \|^2,  \label{eq:Adam_transfo2}
\end{align}
mimic the Adam ODE \eqref{eq:adam_ode1}-\eqref{eq:adam_ode3}. 

\begin{remark}
        The choice \ref{eq:Adam_transfo2} for the function $g$ is not globally Lipschitz, and does not satisfy the assumptions of Theorem \ref{th:conv} because of the square. However, one could also consider $g$ to be the norm of $\nabla U$ with power $r$ (which is then globally Lipschitz for $r\leq 1$, but has a singularity at points where $\nabla U$ vanishes). To fulfil the assumptions, you could also use instead the smoothed version $(\|\nabla U\|^2+\epsilon)^{s/2}$,  $s \leq 1$.
        
        In practice, we could instead consider $\psi(\zeta)$ rather than $\zeta$ as a variable in the reweighting and the SDE, which is uniformly bounded from above and below. The time-rescaled process \eqref{eq: full_framework1}-\eqref{eq: full_framework2} could be expressed in this way.
\end{remark}

With \eqref{eq:Adam_transfo2} the $\zeta$-dynamics \eqref{eq: zeta_dynamics} becomes
\begin{align}\label{eq:zeta_dynamics_final}
\diff \zeta_{\tau} = - \alpha \zeta_{\tau} \diff \tau+  \|\nabla U(x_{\tau}) \|^2\diff \tau.
\end{align}
Similar to Adam, the solution $\zeta(\tau)$ then computes the weighted average of $\|\nabla U\|^{2}$  over recent history, as determined by the exponential weight $\alpha$, see \eqref{eq:zeta_solution}-\eqref{eq: Z-map} and the discussion in Appendix \ref{sup:sec:stepsize_variation}. 

The resulting dynamics differs in structure from that of Adam in two key ways:
\begin{enumerate}
    \item 
In Adam, the choice is made to leave the ODE for the momentum untransformed. In other words, the Adam ODE \eqref{eq:adam_ode1}-\eqref{eq:adam_ode3} is an incomplete Sundman transform.
Viewing the dynamics \eqref{eq:adam_ode1}-\eqref{eq:adam_ode3} in real time $t$ leads to
\begin{align*} 
\diff x_{t} & = p_t\diff t,                                            \\ 
\diff p_{t} & = - \frac{1}{\psi}\nabla U(x_t)\diff t - \frac{\gamma}{\psi} p_t \diff t,  
\end{align*} 
which implies that Adam uses an effective configuration-dependent momentum evolution.  By contrast, our method adapts the timescales of position and momentum components in a symmetric way.
\item Adam uses a vectorial $\zeta\in\mathbb{R}^{d}$, which in the language of time-rescaling corresponds to an individual adaptive stepsize per degree of freedom. The squared Euclidean norm in \eqref{eq:Adam_transfo2}
is replaced by an elementwise squaring in \eqref{eq:adam_ode3}.  Incorporating such individual timesteps in a sampling context introduces additional complication in both the theoretical foundation and practical implementation, and is left for future work.
\end{enumerate}

While we could employ the Sundman kernel \eqref{eq:Adam_transfo1}, we follow the idea of \cite{Ben_adaptive_verlet,alix} and introduce a filter in $\psi$ in order to restrict the value to a specified interval.  Two (similar) choices for Sundman transformation are 
\begin{equation}\label{sundman}
{\psi}^{(1)}(\zeta)=m\frac{\zeta^r+M}{\zeta^r+m}, \hspace{0.3in} 
{\psi}^{(2)}(\zeta)=m\frac{\zeta^r+M/m}{\zeta^r+1},
\end{equation}
for two constants $0<m<M<\infty$. We see that, for either choice, 
${\psi}(0)=M$ and $\psi(\infty)=m$ such that $m$ and $M$ serve as bounds on the Sundman transform kernel and hence on the adaptive stepsize, which then satisfies $\Delta t \in [\Delta t_{\rm min}, \Delta t_{\rm max}]$, where $\Delta t_{\rm min} =m\Delta \tau$ and $\Delta t_{\rm max} = M\Delta \tau$. In particular, $M/m$ gives the maximum factor by which the timestep can be dilated relative to the minimum stepsize. Using $\psi^{(1)}$ or $\psi^{(2)}$ can improve stability compared to \eqref{eq:Adam_transfo1} and allows the user to exert more control on the effective stepsizes $\Delta t$. The power $r>0$ can additionally be used to adjust the dependency to influence the distribution of stepsizes used in simulation. 
We note that for $\zeta\gg 0$,\footnote{To see this, divide numerator and denominator of $\psi^{(2)}$ by $\zeta^r$ and expand in a geometric series.}
\[
\psi^{(2)}(\zeta) \sim  m + (Mm-m^2)/\zeta^r.
\]
Hence for $m\approx 0$, $Mm=1$, we see that this is asympotically related to $\zeta^{-r}$ for large $\zeta$, i.e., where the smaller steps are needed. For further discussion of how the $\zeta$-dynamics \eqref{eq:zeta_dynamics_final} together with the transform kernel $\psi$ influences the stepsize adaption, we refer to Appendix \ref{sup:sec:stepsize_variation}.  


An Adam-like choice of $g$ is $g(x,p)=\|\nabla U(x)\|^2$, i.e., monitoring the force norm.  The closest fidelity to \eqref{eq:Adam_transfo1} is thus found with $r=1/4$ in the Sundman transformation, since scaling the full vector field by $a$ can be related to scaling half the system only by $a^2$.  
For flexibility, we propose to use 
\begin{equation}\label{eq:monitoring_function}
g(x,p)=\Omega^{-1} \|\nabla U(x)\|^s, 
\end{equation}
monitoring the force norm raised to some possibly fractional, positive power $s$ and scaled by a normalization factor $\Omega^{-1}$. Effectively, this choice combined with raising $\zeta$ to the $r$th power  in  the restriction function is tantamount to controlling the timestep based on the $r\cdot s$ power of the gradient norm.  
 
Where it is not a physical parameter of the modelling task, the temperature $\beta^{-1}$ allows the method to behave more like an optimizer ($\beta^{-1}\approx 0$) or a sampling scheme ($\beta^{-1} > 0$). In a Bayesian sampling context, $U(x)$ is the negative log-posterior and we may take $\beta^{-1}=1$ to sample the posterior or $\beta^{-1} <1$ to implement annealed importance sampling \cite{Neal2001}.  Due to the various points mentioned above, SamAdams with $s\cdot r=1/2$ does not strictly reduce to Adam for $\beta^{-1} \rightarrow 0$. In our experiments we found that the optimal choices of $r$ and $s$, as well as other coefficients, were problem-class dependent, but within a specific class of models those selections were relatively easy to decide.

\subsection{Other Monitor Functions}
The choice to make SamAdams relate to the Adam optimizer is not unique: possible alternatives include basing $g$ on only the prior in a Bayesian sampling setting or  basing $g$ on some subset of the variables (which are known to exhibit high levels of variability).  It would also be possible to introduce higher derivative information derived from the potential energy function such as the trace or determinant of the Hessian matrix. We have not substantially explored such options. 
For noisy gradient evaluations, one might also design a monitor function based on the estimation of gradient noise such that stepsizes are decreased whenever the injected gradient noise is large. This idea is motivated in Appendix \ref{sup:sec: LogisticRegression}, which shows results for logistic regression with stochastic gradients, during which the adaptive stepsize strongly reacts to the employed batch size. 
\section{Numerical Experiments\label{sec:numerics}}
In all our numerical experiments, we use the monitor function \eqref{eq:monitoring_function} for different $\Omega$ and $s$ (mostly $s\in\{1,2\}$), and the Sundman kernel $\psi$ may either be $\psi^{(1)}$ or $\psi^{(2)}$ from \eqref{sundman}, with different values for $m$, $M$, and $r$. We note that the detailed finetuning of these hyperparameters is usually not necessary. As discussed in Sec. \ref{sec:algorithm}, there are various ways in which $\zeta_0$ may be initialized, and we make use of both presented variants.
We examine the behavior of SamAdams in terms of individual trajectories as well as ergodic averages. In the case of the latter, we typically run multiple independent trajectories and first compute the intra-trajectory averages (i.e., time averages using the reweighting from Sec. \ref{sec:computing_averages}) and then the inter-trajectory average. To measure sampling biases, we compute weak errors of observables $\phi(q)$, $q:=(x,p)$, defined by
\begin{equation}\label{eq:weak_errors}
\lim_{n\to\infty}\Big|\mathbb{E}\big[\phi\big(q_n\big)\big]-\mathbb{E}\big[\phi\big(q(n\Delta t)\big)\big] \Big|,
\end{equation}
where $q_n$ gives the iterates obtained by the simulation. The expectation of the numerical iterates, $\mathbb{E}\big[\phi\big(q_n\big)\big]$, is estimated by time- and trajectory averages obtained from simulation, where we sometimes use the notation $\langle \phi \rangle$ to denote that empirical average. We also write $\langle\Delta t\rangle$ to denote the empirical mean of the adaptive stepsize used by SamAdams in a given experiment. We make use of the two temperature observables often monitored in molecular dynamics\cite{benMDbook}:
\begin{align}
T_{\text{kin}}  &:= \frac{1}{N_{\text{dof}}} \|p\|^2, 
  && \makebox[6.5cm][l]{(kinetic temperature)} \\
T_{\text{conf}} &:= \frac{1}{N_{\text{dof}}} x \cdot \nabla U(x), 
  && \makebox[6.5cm][l]{(configurational temperature)}
\end{align}
with $N_{\text{dof}}$ as the number of degrees of freedom. The expectation of these quantities is known and given by the system temperature, $\mathbb{E}(T_{kin})=\mathbb{E}(T_{conf})=T=\beta^{-1}$.
We also use other observables in our low-dimensional experiments whose expectations can be computed via numerical quadrature. Thus, the ground truths $\lim_{n \to \infty}\mathbb{E}[\phi(q(n\Delta t))]$ of the observable means are either known or easily obtainable. For the Bayesian neural network problems we do not have access to the true posterior means, but we measure sampling quality in terms of negative loss likelihood (given by the loss function) and posterior-averaged classification accuracies.

We integrate the SamAdams dynamics using the  ZBAOABZ scheme (see Alg. \ref{alg:AdamSampler}). When using our proposed Sundman kernels \eqref{sundman}, the algorithm has several hyperparameters on top of the usual Langevin friction $\gamma$ and temperature $T$, namely $\alpha$, $\Delta \tau$, $m$, $M$, $r$, and indeed the monitor function $g$. The values for $\Delta \tau$, $m$, and $M$ are relatively easy to choose based on the desired stepsize range (since $\Delta t\in (m\Delta\tau,\ M\Delta \tau]$). The parameter $r$ also has intuitive meaning: it governs the sensitivity of the transform kernel $\psi$ (and hence the stepsize $\Delta t$) to changes in $\zeta$ (see Fig. \ref{sup:fig:sundman_transform} in Appendix \ref{sup:sec:stepsize_variation}). We set it to values in $\{0.25,\ 0.5,\ 1 \}$ and did not need to finetune it on a given experiment. The $\zeta$ relaxation rate $\alpha$ and the choice of the monitor function $g$, however, are critical for good performance. When using our proposed monitoring function class \eqref{eq:monitoring_function}, one needs to specify the exponent $s$ and scaling factor $\Omega$. Once again, $s$ could just be taken to be 1 or 2 in practice without the need to finetune it, but picking the right scale $\Omega$ requires some care. In \ref{sec:alpha_omega_study} and Appendix \ref{sup:sec:stepsize_variation} we discuss the influence of $\alpha$ and $\Omega$ in more depth.

As mentioned in Sec. \ref{sec:algorithm}, there are various ways to initialize $\zeta$. is by simply setting $\zeta_0=0$. Since $\psi(\zeta_0)=M$ for either of the given choices of the filter function, we have $\Delta t_0=M\Delta \tau$=$\Delta t_{\max}$. The advantage of this approach is that the initial stepsize is known, independently of the value of the driving function $g(x_0,p_0)$. The disadvantage is that the stepsizes $\Delta t_n$ can then never become larger than $\Delta t_0$, which may prevent efficiency gains in settings where the simulation is started in a critical area, requiring small stepsizes, and then moves into more benign areas, allowing for larger stepsizes. This approach is particularly suitable for cases in which one can pick the initial configuration $(x_0,p_0)$ such that $g(x_0,p_0)$ is small. In low-dimensional models this is often easy, and we have followed this procedure below. For higher-dimensional models one may first need to run an optimization scheme. Another way of initializing $\zeta$ is by setting $\zeta_0=g(x_0,p_0)$. That way, the initial stepsize $\Delta t_0$ is fully determined by the driving function at the initial conditions, which is natural given that $\Delta t_n$ with $n>0$ will depend on the current and past values of the driving function. This allows the algorithm to pick the ``right'' initial stepsize automatically, which will typically not lead to $\Delta t_0=\Delta t_{\max}$. The downside of this approach is that the user has no way to enforce a certain initial stepsize. It's all up to the algorithm to pick one in the interval $(m\Delta \tau,\ M\Delta \tau]$. 

The experiments compare ZBAOABZ to the constant-stepsize equivalent BAOAB. Since BAOAB is a high-quality sampling method in and of itself, able to outperform many of the conventionally used unadjusted schemes\cite{LeMa2013}, showing that ZBAOABZ is able to outperform BAOAB is equivalent to pushing the limits of state-of-the-art sampling.

\subsection{Asymmetric Double Well}\label{sec:toymodel}
We first demonstrate the stepsize adaptation of SamAdams on a one-dimensional double well problem where one well is much narrower than the other. The potential is given by $U(x)=\frac{b}{L}(x+1)^2(x-L)^6$ with $b=1.5$ and $L=2$. We pick a small temperature $T\equiv \beta^{-1}=0.4$ to make the transition across the barrier reasonably rare. Fig. \ref{fig:toy_results}a) shows the potential and the density. The hyperparameters of monitoring function \eqref{eq:monitoring_function} are $s=2$ and $\Omega=1$. We initialize $(x_0,p_0)=(2,0)$ and SamAdams variable $\zeta_0=g(x_0,p_0)=\|\nabla U(x_0) \|^2$. We use Sundman kernel $\psi=\psi^{(1)}$ with $m=0.1$,  $M=10$,  $r=0.25$. Langevin friction and $\zeta$ friction are given by $\gamma=\alpha=1$. 
Since a Langevin path will encounter greater forces in the narrower potential well, the aim is for SamAdams to decrease its stepsize $\Delta t$ accordingly. That this is indeed the case can be seen in Fig. \ref{fig:toy_results}b), which shows the $x$- and $\Delta t$-values of a single SamAdams trajectory. It can be seen that occupation in the narrow well (negative $x$-values) leads to a restriction of the stepsize $\Delta t$ to small values. 
\begin{figure}[!htb]
\begin{centering}
\includegraphics[width=1\textwidth]{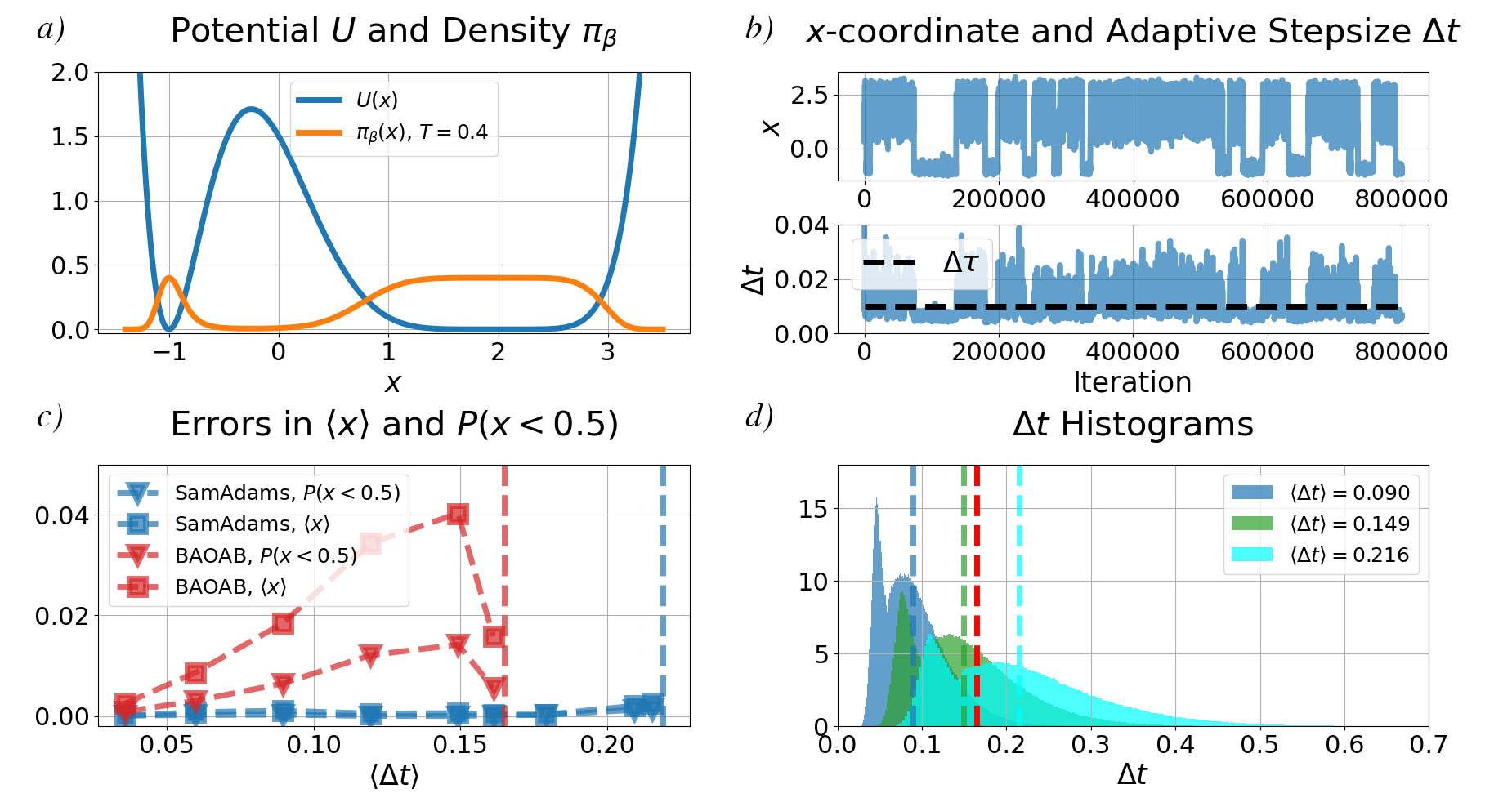} 
\end{centering}
\caption{\label{fig:toy_results} Sampling experiments on a 1D toy model (see text). \textbf{a)} Potential and Gibbs density for employed temperature $T=0.4$. \textbf{b)} $x$-coordinate and adaptive stepsize $\Delta t$ for SamAdams along a single trajectory. The black dashed line gives the value of virtual stepsize $\Delta \tau$ which is adaptively increased or reduced to yield the real stepsize $\Delta t$. \textbf{c)} Absolute mean errors of two observables, the $x$-coordinate and the occupation frequency of area $x<0.5$ against (mean) stepsize $\Delta t$. The different values for SamAdams were obtained by varying $\Delta \tau$ from 0.03 to 0.2. \textbf{d)} $\Delta t$ histograms for SamAdams run at three different $\Delta \tau$. }
\end{figure}
Fig. \ref{fig:toy_results}c) shows weak errors of two observables, the $x$-coordinate and indicator function of the domain $x<0.5$ (which roughly corresponds to the occupation probabilities of the narrow well and the barrier), against stepsize (mean stepsize for the adaptive scheme).  We first ran SamAdams for different values $\Delta \tau$, measured the mean adaptive stepsizes $\langle \Delta t \rangle$ for each of these runs, and then ran BAOAB at stepsizes fixed to these values to obtain canonical averages at the same mean stepsize and compute cost.
Each point in the figure was generated by averaging over 300 independent trajectories for $5\cdot 10^7$ iterations (discarding the first 100,000 iterations as burn-in). We also plot vertical lines denoting the maximum (mean) stepsize at which the corresponding algorithm became unstable in at least one of the 300 trajectories (i.e., the stability threshold). We see that BAOAB does significantly worse than SamAdams, only reaching similar accuracies for the smallest stepsize examined. SamAdams' performance barely depends on $\langle \Delta t\rangle$ until very close to its stability threshold. We also observe that it is able to use larger steps than BAOAB. 

Fig. \ref{fig:toy_results}d) shows the $\Delta t$ histograms for three of the SamAdams runs. They have a bimodal structure, as might be expected from an asymmetric double well. All three of them are able to use stepsizes larger than the stability threshold of BAOAB, supporting the idea that the stability threshold for fixed-stepsize schemes depends on local variation of the loss landscape, rather than on the landscape as a whole. As for  the star potential of Fig. \ref{fig:star_trajectory}, it is enough to use small stepsizes in critical areas. Note how one of the $\Delta t$ histograms has a mean of $\langle \Delta t \rangle=0.216$, which is 31\% larger than BAOAB's threshold, which together with the low observable errors even at that stepsize implies a substantial increase in computational efficiency compared to BAOAB. The ability to use larger stepsizes (in some cases {\em much} larger stepsizes) than constant stepsize schemes while preserving sampling quality will also be observed in the examples of the next subsection.

\subsection{Planar Systems}
In the introduction, we already mentioned an example involving the ``star potential'', which has narrow corridors that can be difficult to sample efficiently.  Here we consider several other examples of 2-degree of freedom problems and explore the accuracy and stability of the new method in comparison with fixed stepsize integration.

\subsubsection{Neal's Funnel}
We consider a 2D variant of Neal's funnel \cite{Neal2001} (a 9D version will be taken up in the following subsection).   The potential energy function is
\[
U_{\rm Neal} (x,\theta) = 
\frac{x^2}{2e^\theta} + \frac{\epsilon}{2}(x^2+\theta^2).
\]
Our goal is to sample the canonical distribution at temperature $T=1$.  We used $\gamma=5$ and discarded $10^5$ steps as burn-in (equilibration). The initial position was taken to be $(0,5)$ in a relatively flat zone, and we set initial momenta and $\zeta_0$ to zero. 

Because canonical sampling may be intractable due to unconstrained domains, one often incorporates a term in the form of a prior (and associated potential) to maintain confinement of solutions; we have done this here by using a simple harmonic restraint.

In the funnel problem the domain shrinks to a narrow neck as $\theta\rightarrow -\infty$.  This creates a numerical challenge as the trajectory rattles back and forth against the walls of the channel; for a fixed moderate stepsize, at some point the solution will become unstable and jump out of funnel.  In practice these escaping trajectories often re-enter the larger domain (with $\theta>0$) and re-equilibrate at the target temperature, but the unstable behavior can damage the computation of observables.     An example of this type of behavior can be seen in Fig. \ref{funnel_instability} where a trajectory is shown together with the evolution of kinetic temperature and average potential.  As we can see, the two observables are severely degraded when the instability is encountered which would ultimately be seen as poor convergence.  In the right panel of Fig. \ref{funnel_instability}, we also see what happens when the stepsize is halved ($\Delta t=0.02$).  The instabilities are still very much in evidence if the trajectory is long enough (here $N=10^8$).  The stepsize would need to be below $\Delta t=0.01$ to completely eliminate the instability.\footnote{Discrete Langevin trajectories for potentials that are not globally Lipschitz are inherently unstable, due to the use of normally distributed random variables (see \cite{Talay2002,MiTr2005_paper} for some discussion); that said, in our experience, for typical systems the frequency of long excursions in fixed stepsize trajectories decreases rapidly as the stepsize drops below a certain well defined stability threshold. }  

\begin{figure}[!htb]
   \centering
\includegraphics[width=0.95\textwidth]{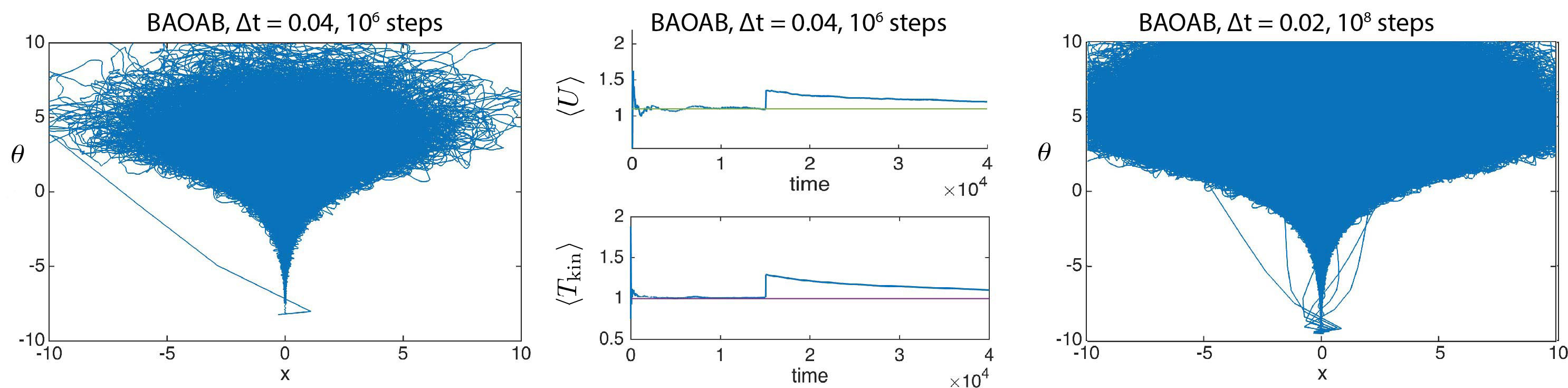}
\caption{\label{funnel_instability} \textbf{Left, Center:} A BAOAB trajectory with stepsize $\Delta t=0.04$ shows an unstable evolution in $10^6$ steps. The descent into the funnel leads to a spike in both kinetic temperature and mean potential energy.  Although shortlived, this type of event can, as here, corrupt long term averages. \textbf{Right:} At longer times, these instabilities are inevitable,  for stepsizes above or equal to $\Delta t= 0.015$. }
\end{figure}

By contrast, SamAdams (ZBAOABZ)  produces reliable, stable trajectories with much larger mean stepsize than any fixed stepsize method.  In Fig. \ref{funnel_SA} we show a trajectory with mean stepsize $\langle \Delta t \rangle = 0.16$ ($N=10^7$). (The details of stepsize variation are as follows:  $\Delta t_{\rm min} = m\Delta_{\tau} = 10^{-4}$, $\Delta t_{\rm max}=M\Delta_{\tau} = \Delta t_0 = 0.6$, $\zeta_0=0$,  $\alpha = 0.1$,  $r=0.5$ and $g(x,p)=\|\nabla U\|$.  We used the second form of the filter function $\psi=\psi^{(2)}$.)

A histogram computed using the samples obtained from SamAdams is virtually identical to the target distribution (Fig. \ref{funnel_dists}), despite requiring around 10\% of the computational effort needed if the corresponding fixed-stepsize method was used.  Note that since a histogram can be interpreted as an observable itself, the samples obtained via SamAdams have to be reweighted as described in Sec. \ref{sec:computing_averages} before plotting the histogram. Observables like the kinetic temperature or the mean potential energy as shown in Fig.\ref{funnel_SA} are approximated to three significant digits.
\begin{figure}[!htb]
   \centering
\includegraphics[width=0.7\textwidth]{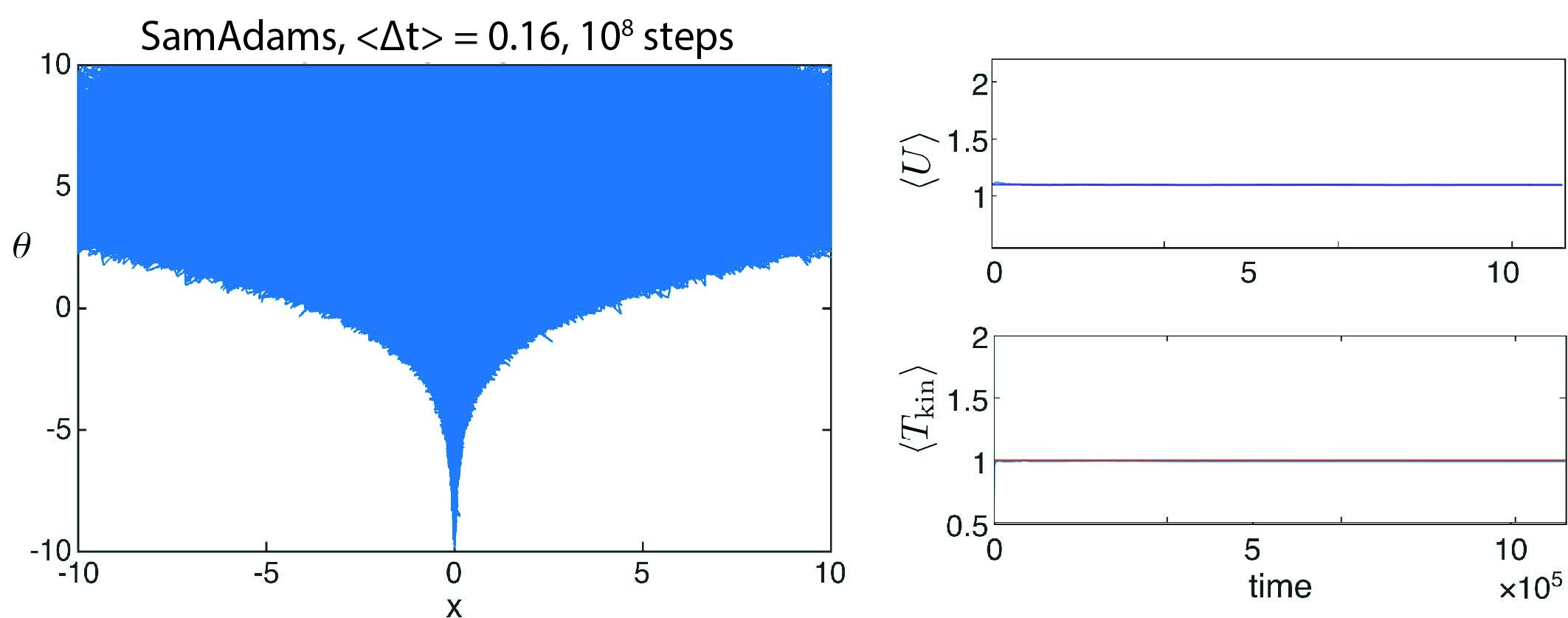}
\caption{\label{funnel_SA}A SamAdams trajectory with mean stepsize $\langle \Delta t\rangle =0.16$ corrects the instability of the fixed stepsize method.  The kinetic temperature and potential energy average converge to three significant digits of accuracy.}
\end{figure}
\begin{figure}[!htb]
   \centering
\includegraphics[width=0.7\textwidth]
{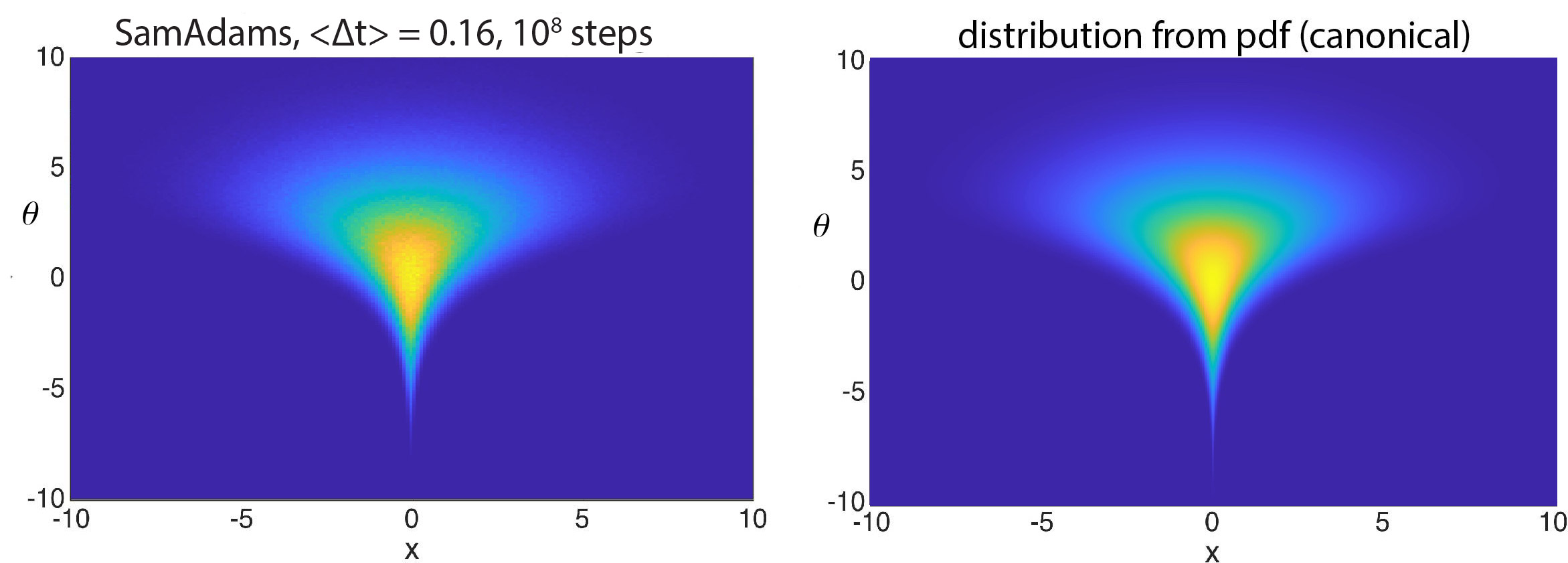}
\caption{\label{funnel_dists}Here we compare the (reweighted, see text) histogram of trajectory data for SamAdams to the actual canonical distribution. The distributions are visually very similar.}
\end{figure}

\subsubsection{Entropic Barrier}
Rare event sampling problems in many fields can be characterized by low energy basins connected by thin channels.  Diffusion through the narrow corridors typically requires small stepsizes since too-large stepsizes either lead to expulsion from the corridor or numerical instability. An entropic barrier problem was constructed to illustrate the challenge of such tasks.   The potential function  is
\[
U_{\rm channel}(x,y)= \frac{y^2}{1+10x^4}+0.001(x^2-9)^2.
\]
We simulate this with temperature set to $0.05$ to create a  challenging test.  We used friction $\gamma=5$ in this example which gave approximately optimal results. The initial point was taken at $(x,y)=(3,0)$ near the minimum on the right side, and the initial momenta were set to zero. The SamAdams parameters were like in the previous example, but we set $\Delta t_{\rm min}=0.0001$, $\Delta t_{\rm max}=\Delta t_0= 0.5$.
\begin{figure}[!htb]
   \centering
\includegraphics[width=0.9\textwidth]{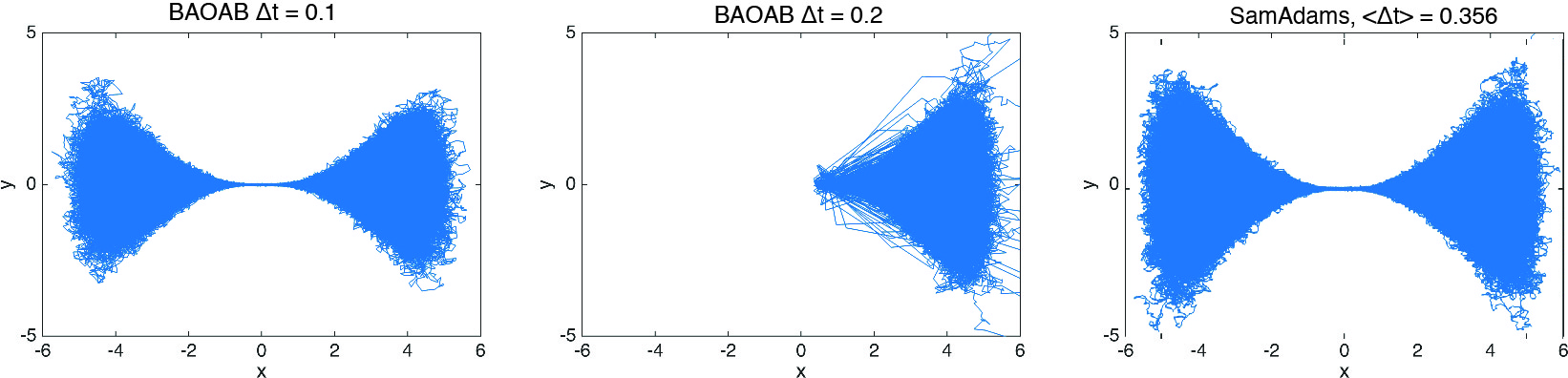}
\caption{\label{entropic_compare_1} Comparison of trajectories obtained using fixed and variable stepsize.  \textbf{Left:} BAOAB with fixed stepsize $\Delta t=0.1$  converges as expected (with around 2-3 barrier crossings per 1M steps). \textbf{Center:} BAOAB becomes unstable above $\Delta t=0.15$ and at $\Delta t=0.2$ shows no diffusion over the barrier. \textbf{Right:} A variable stepsize trajectory with $\langle \Delta t\rangle =0.356$, restoring the performance of small fixed stepsize.}
\end{figure}

\begin{center}
\begin{figure}[!htb]
\includegraphics[width=0.55\textwidth]{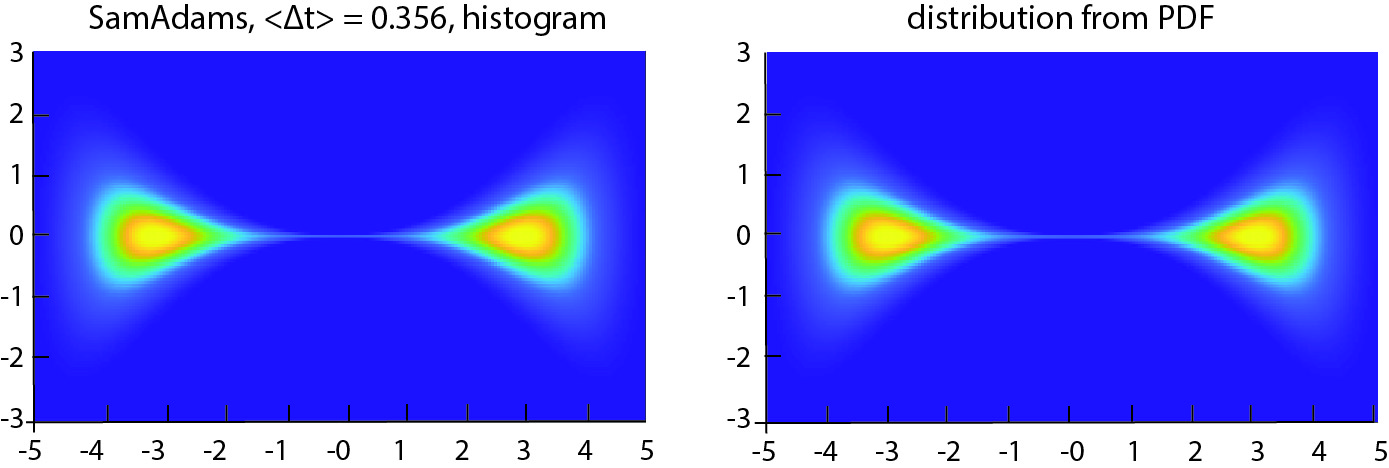}\hspace{0.1in}
\includegraphics[width=0.4\textwidth]{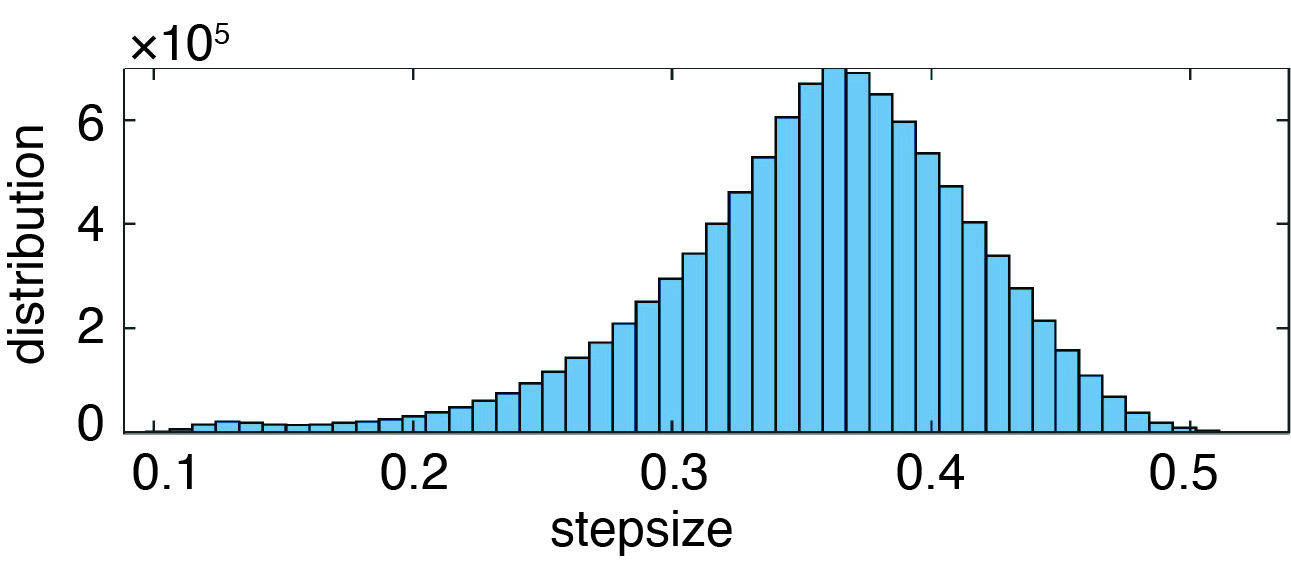}
\caption{\label{entropic_compare_2} SamAdams trajectory with large mean stepsize ($\langle \Delta t\rangle = 0.356$).  \textbf{Left:}  
Histogram of trajectory data. \textbf{Center:} Exact distribution. \textbf{Right:} 
 Distribution of adaptive stepsize $\Delta t$. Only a vanishingly small number of steps require a small (below 0.2) stepsize. Not visible in the histogram is the fact that a very small number of steps used a stepsize below 0.01.}   
\end{figure}
\end{center}

We found that the fixed stepsize integrator was stable up to a maximum stepsize of $\Delta t\approx 0.2$, for many trajectories,  but above $\Delta t=0.15$ there is a lot of error in the narrowest part of the channel and the number of full crossings from one basin to the other is significantly reduced.  Fixed stepsize trajectories with $\Delta t=0.1$ and $\Delta t=0.2$ are shown in Fig. \ref{entropic_compare_1}.  Also shown is a typical SamAdams trajectory with mean stepsize $\langle \Delta t\rangle=0.356$.   

In the left panel of Fig. \ref{entropic_compare_2} we show a histogram of the computed states for $\langle \Delta t\rangle=0.356$  which can be compared with the exact distribution in the central panel. In this example, with $10^8$ steps, the kinetic and configurational temperatures are accurate to within 1\% ($T_{\rm kin} = 0.04947$, $T_{\rm conf} = 0.04954$).  Of particular interest is the fact that the stepsize distribution shown at right in Fig. \ref{entropic_compare_2} has barely any mass below $\Delta t=0.2$, meaning that the small stepsizes are only needed at very rare instances of barrier crossing (precisely in a small neighborhood of the origin).

\subsubsection{Beale Potential}
In this 2D model there are again two basins, but they have a complicated curved shape. The rarity of transitions is actually still more extreme than in the entropic barrier problem. The two wells have very different depths and shapes.  Transitions between the wells happen in a narrow region around the origin, although we have noticed that numerical error can either eliminate corridors or create new ones.
With a 6th order exponential confinement term, the Beale potential takes the form
\begin{eqnarray*}
U_{\rm Beale}(x,y) & = & (1.5 - x + xy )^2 + ( 2.25 - x + xy^2)^2+ (2.625 - x+ xy^3)^2 \\
& & \hspace{0.3in} +0.3\exp \left (0.00001(x^6 + y^6) \right ).
\end{eqnarray*}
We used a temperature of $T=3$ to obtain sufficient barrier transitions.  The initial point was $(x,y)=(3,0)$, with zero momenta, and we use $\gamma=1$.

For fixed stepsize integration, the step needs to be below about 0.003.  Running at or above this threshold almost all runs of $10^8$ steps result in failure.   At the small value $\Delta t=0.0025$, $10^6$ steps are not enough to cover the distribution, as shown in the center panel of Fig. \ref{beale_dists}.
By contrast (see the right panel of Fig. \ref{beale_dists}) we were able to use SamAdams reliably with a mean step of $\langle \Delta t \rangle = 0.022$ which indicates a stability improvement of nearly an order of magnitude.  While there is some visible error in  the histogram for SamAdams (Fig. \ref{beale_dists}) it is much less than for BAOAB operating at its much smaller stable stepsize.  

\begin{center}
\begin{figure}[!htb]
\includegraphics[width=0.85\textwidth]{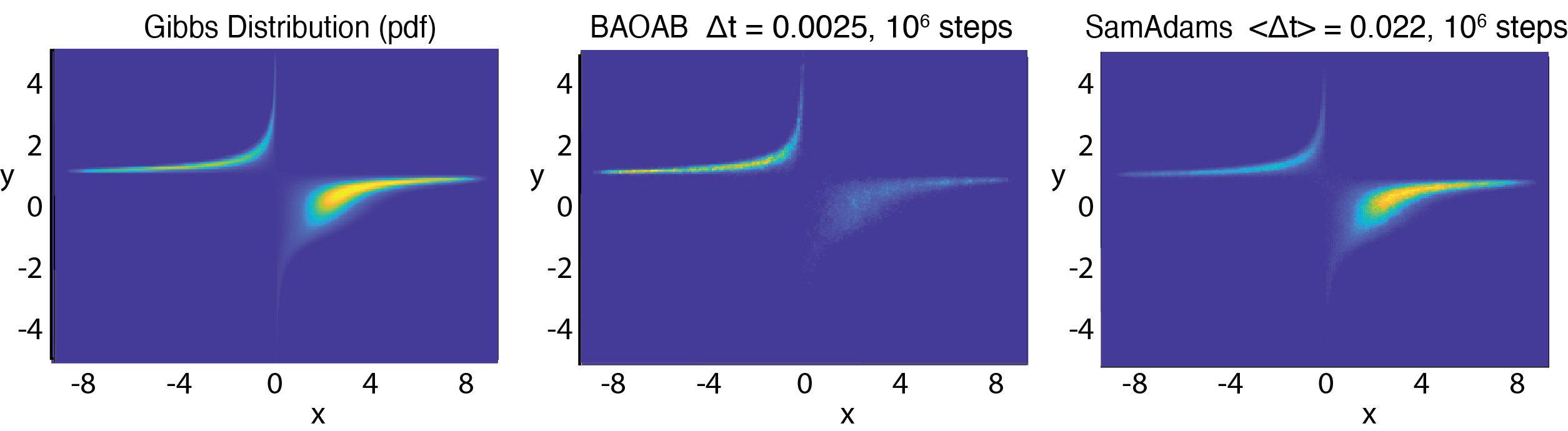}
\caption{\label{beale_dists}Sampling results for the Beale potential. \textbf{Left:} Gibbs distribution (ground truth). \textbf{Center:} Histogram for a fixed stepsize BAOAB run with $\Delta t=0.0025$ (larger stepsizes are unreliable). \textbf{Right:} Histogram obtained from a SamAdams trajectory with $\langle \Delta t\rangle = 0.022$. Settings: $\psi=\psi^{(2)}$ with $\Delta t_{\rm min}=m\Delta \tau=0.001$, $\Delta t_{\rm max}=M\Delta \tau=0.1$, $\Delta t_{0}=0.1$, $r=0.5$,  $\alpha=1$, $\zeta_0=0$, and monitor function $g(x,p) = \|\nabla U(x)\|$. }   
\end{figure}
\end{center}


We also perform more extensive runs to resolve the bias in canonical averages.  The procedure is similar to the case of the 1D toymodel in Sec. \ref{sec:toymodel}. To obtain different $\langle \Delta t \rangle$ for SamAdams, we varied $\Delta \tau$, keeping everything else fixed, and then ran BAOAB at fixed stepsizes set to the obtained $\langle \Delta t \rangle$ values for comparison. We averaged over 200 independent trajectories, where the number of iterations scaled linearly with $\Delta \tau$ (we used $10^8$ iterations at $\Delta \tau=0.001$, and varied $\Delta \tau$ from 0.00025 to 0.06). The computational cost of each BAOAB simulation was the same as for the corresponding SamAdams run. Fig. \ref{fig:beale_errors} shows the bias against (mean) stepsize $\Delta t$ for the two temperatures and coordinates.
While SamAdams and BAOAB show comparable errors for small stepsizes, BAOAB becomes unstable at $\Delta t=0.0025$. SamAdams can still be run at $\langle \Delta t \rangle =0.01$ without any decline in accuracy. If a modest decline in accuracy is acceptable, it can be run at substantially larger steps. In fact, while the BAOAB curves in Fig. \ref{fig:beale_errors} stop at the last stable stepsize,  SamAdams remained stable until the last examined $\Delta \tau$.  
\begin{center}
\begin{figure}[!htb]
\includegraphics[width=0.85\textwidth]{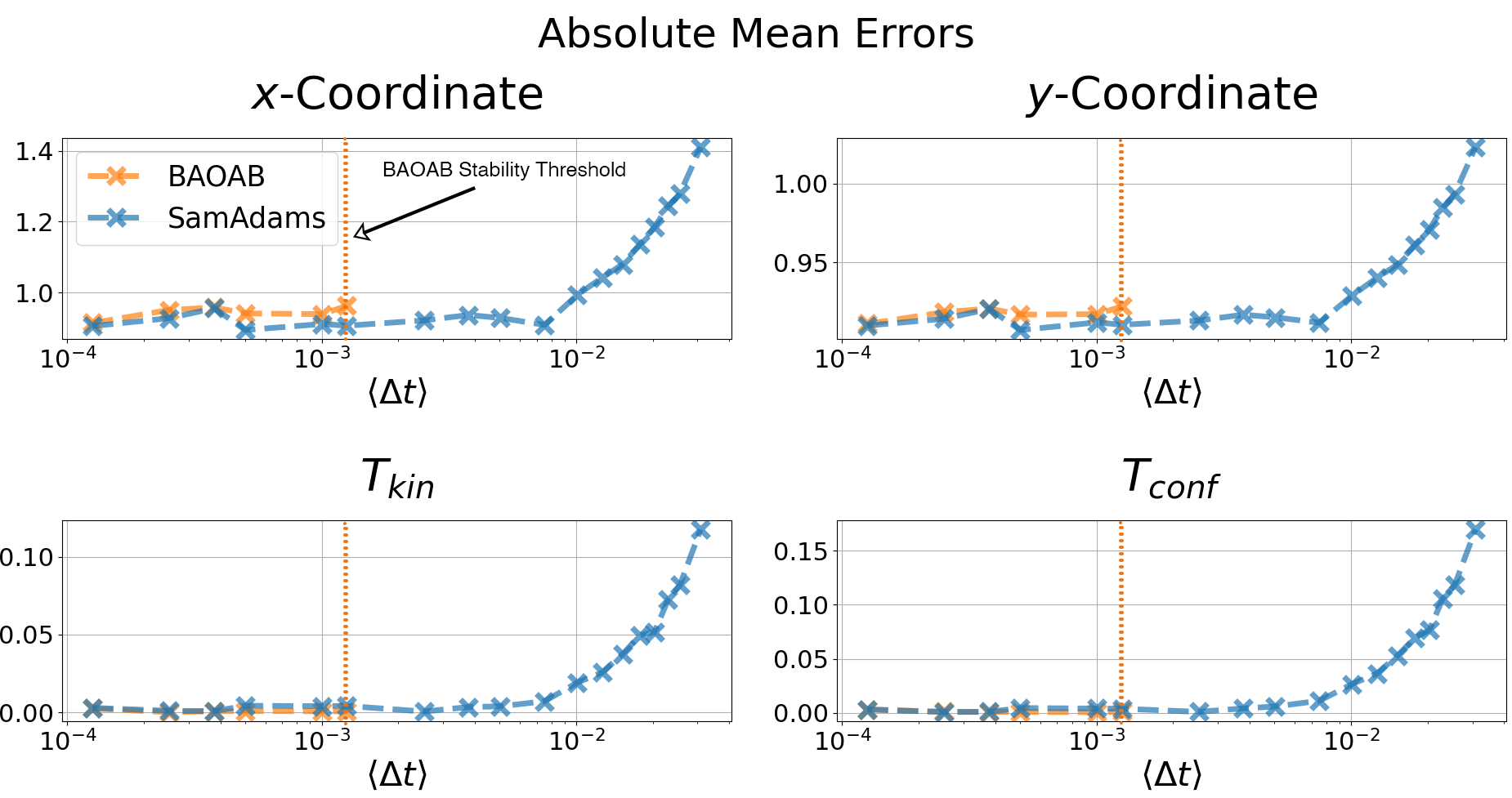}
\caption{\label{fig:beale_errors}Absolute errors of the means of four observables on the Beale potential against (mean) stepsize $\Delta t$. The ground truths for the coordinates were obtained via numerical quadrature. Simulation settings: $\gamma=1, \ T=3, \ \alpha=1, \ g(x,p)=0.1\|\nabla U\|^2, \ \psi=\psi^{(1)} \text{ with } m=0.1, \ M=10, \ r=0.25.$ See text for more information. }  
\end{figure}
\end{center}

\subsubsection{Auxiliary Variable Control Dynamics} \label{sec:alpha_omega_study}
Until now we have not yet explored the role of the hyperparameter $\alpha$ in \eqref{eq:zeta_dynamics_final} nor the influence of the scale of the monitor function $g(x,p)$. In the case of $g(x,p)=\Omega^{-1}\| \nabla U(x)\|^s$ (our choice for all of our experiments in this article), the scaling factor $\Omega$ governs how strongly $\zeta$ reacts to a given force value $\nabla U(x)$. Since the range of typical values of $\nabla U$ strongly depends on the problem at hand, the scaling parameter $\Omega$ is useful to adjust the range of $g(x,p)$ and hence the range of $\zeta$ and $\Delta t$. For the one- and two-dimensional examples considered thus far, setting $\Omega=1$ serves as a good starting point, whereas on the classification tasks of the the next section we use $\Omega=N_{\text{D}}$ with $N_{\text{D}}$ the size of the dataset. Since the force $\nabla U$ is written as a sum over the model likelihoods of all data points, this choice for $\Omega$ encourages $\zeta$-values close to the fixed point of the Sundman transform $\psi$. The hyperparameter $\alpha$ governs the damping of the $\zeta$-dynamics as well as its moving average behavior, i.e., how quickly past force values are 'forgotten' (see Appendix \ref{sup:sec:stepsize_variation} for more details).

As we have shown thus far, one fundamental benefit of the adaptive stepsize is the ability to use overall large stepsizes while using small stepsizes selectively. If the resulting $\langle \Delta t \rangle$ is larger than the maximum stable stepsize of a conventional scheme (assuming the accuracies remain acceptable), SamAdams leads to an improvement of stability and computational efficiency. An important question is then how difficult it is to tune the SamAdams-specific hyperparameters $\alpha$ and $\Omega$.
To examine this, we chose the example of the star potential.   We create a grid of different $(\alpha, \Omega)$-values and run SamAdams trajectories for each grid point. For each grid point, we experimentally find the stability threshold of SamAdams, i.e., the largest $\langle \Delta t\rangle$ that still leads to stable trajectories. We do this by running 100 independent trajectories for successively increasing values of $\Delta \tau$, the base stepsize, until the trajectories become unstable. We vary $\Delta \tau $ in $[0.04, 0.08]$, a range in which the resulting $\langle \Delta t \rangle $-values become roughly comparable to the stability threshold of BAOAB for most ($\alpha, \Omega$)-combinations \footnote{BAOAB's threshold was found to be $\Delta t=0.01275$, defined as the smallest stepsize for which 100 independent trajectories remain stable for $5\cdot 10^6$ iterations.}. The number of iterations per tested $\Delta \tau$ is scaled via $N=N_0\Delta \tau_0/\Delta \tau $ with $N_0=2\cdot 10^6$ and $\Delta \tau_0=0.06$. If an $(\alpha, \Omega)$-combination leads to unstable trajectories even at the smallest $\Delta \tau$ tested, we deem the combination `highly unstable'. The results are shown in Fig. \ref{fig:star_grid_study}.

\begin{figure}[!htb]
\includegraphics[width=0.85\textwidth]{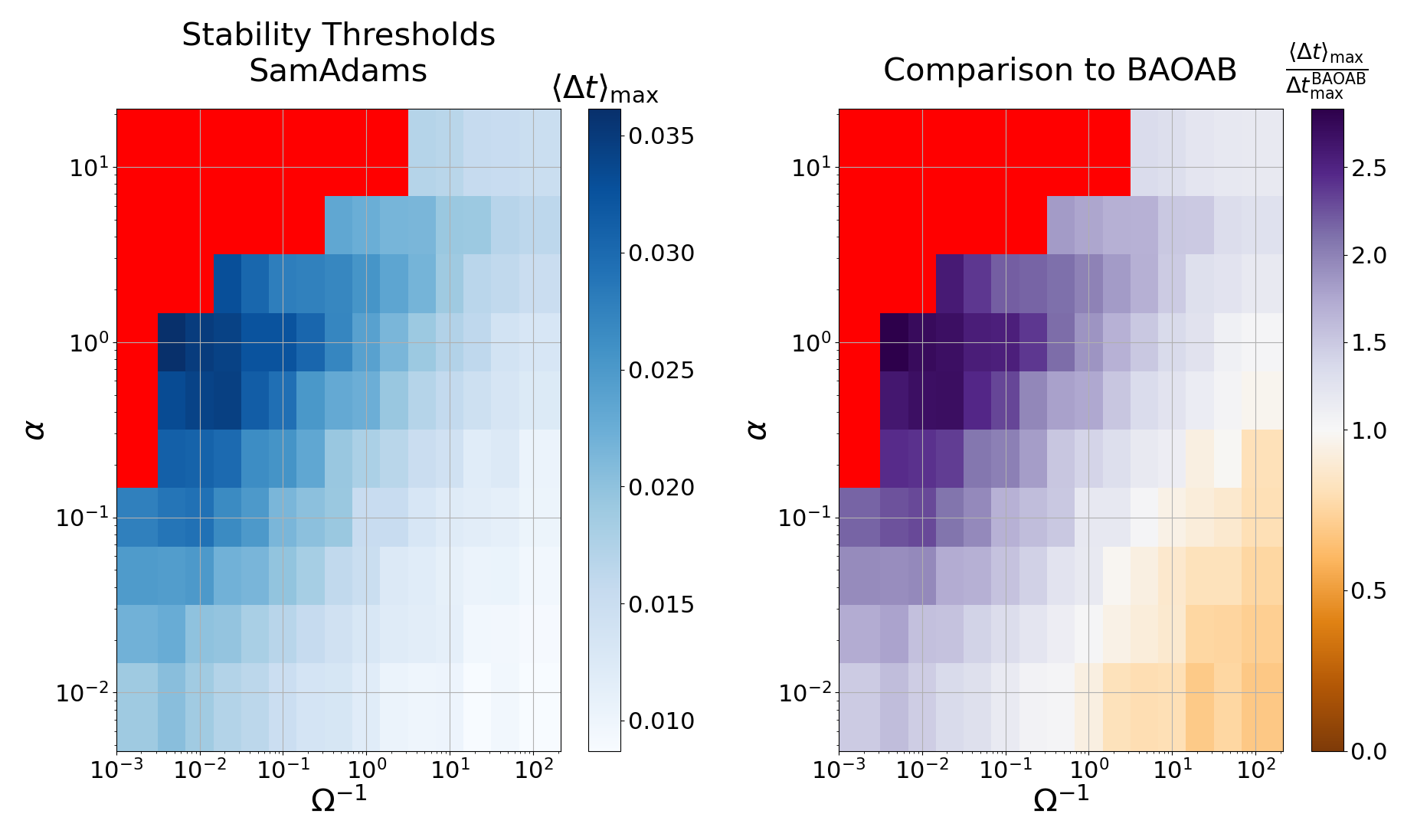}
\caption{\label{fig:star_grid_study} 
SamAdams stability thresholds $\langle \Delta t \rangle_{\text{max}}$ in dependency of  $\alpha$ (attack rate) and $\Omega$ (scale coefficient of force norms) in the case of the star potential. \textbf{Left:} Plain $\langle \Delta t \rangle_{\text{max}}$. A well defined optimal zone appears around $\alpha = 1$, $\Omega = 100$. \textbf{Right:} Fraction of SamAdams threshold and BAOAB's threshold $\Delta t^{\text{BAOAB}}_{\text{max}}$. The brown areas show parameter regions where BAOAB is more stable than SamAdams, the purple area regions where SamAdams is more stable. The red areas denote 'highly unstable' regions (see main text). Other settings: $T=\gamma=1, \ s=2, \ \psi=\psi^{(1)} \text{ with } m=0.1, \ M=10, \ r=0.25$.} 
\end{figure} 
From the left-hand figure, we see how SamAdams' stability threshold varies with $(\alpha, \Omega)$, with an optimum region  visible for $\alpha \simeq 1$ and $\Omega^{-1} \simeq 0.01$. The highly-unstable area (red) is mainly due to too large values of $\alpha$. As shown in Appendix \ref{sup:sec:stepsize_variation}, too large $\alpha$ will lead to $\Delta t = M\Delta \tau $. For our choice $M=10$ and the smallest $\Delta \tau$-value tested, $\Delta \tau=0.04$, this would correspond to a constant-stepsize scheme running at stepsize 0.4, more than 30 times larger than BAOAB's stability threshold. Note that even for these highly unstable ($\alpha, \Omega)$-values, one could still enforce stability by picking $\Delta \tau$ and $M$ such that $M\Delta \tau$ is smaller than BAOAB's threshold. On the right-hand side in Fig. \ref{fig:star_grid_study}, we plot the same grid but color according to the fraction of SamAdams' stability threshold and BAOAB's threshold. We observe that there is a large area in parameter space (shaded purple) in which SamAdams is more stable than BAOAB, with a remarkable improvement of up to 300\%.

While these experiments show the need to pick admissible $(\alpha,\Omega)$-values, they also demonstrate that there is a wide range of settings in which SamAdams is more stable than BAOAB, reducing the need for time-consuming hyperparameter finetuning.

\subsection{Higher-dimensional Models}
We now explore several higher-dimensional cases using fixed stepsize Langevin as well as the SamAdams scheme. Our goal here is to show that there is a good basis for believing the method will be effective for applications in statistics and machine learning.  

\subsubsection{Neal's Funnel (multi-dimensional case).}
The 2D funnel example of the last subsection is a simplified version of a model with multiple latent variables that was originally proposed as a surrogate for problems in Bayesian hierarchical modeling with nonlinear dependencies \cite{Neal2001}; for a related model arising in ecology, see \cite{Mo2017_growth}.  We test SamAdams on the original 9-dimensional model modified with a confining prior with a large variance ($\sigma_x^2 = 20$) to ensure sufficiently frequent trips into the funnel neck.  We used the following setup in our simulations: 
\[
\theta \sim {\mathcal{N}}(0,3), \hspace{0.2in} \left . x_i \right | \theta \sim  {\mathcal{N}}(0,\exp(\theta))\times \mathcal{N}(0,\sigma_x^2),
\]
with p.d.f.
\[
p(\theta, x_1,x_2, \ldots, x_n) = \frac{1}{\sqrt{(2\pi)^n \cdot 6}} \exp\left(-\frac{\theta^2}{6}\right) \cdot \exp\left(-\frac{n}{2}\theta - \sum_{i=1}^n x_i^2\left (\frac{1}{2e^\theta}+\frac{1}{2\sigma_{x}^2}\right)\right ).
\]

\begin{center}
\begin{figure}[!htb]
\includegraphics[width=1.0\textwidth]{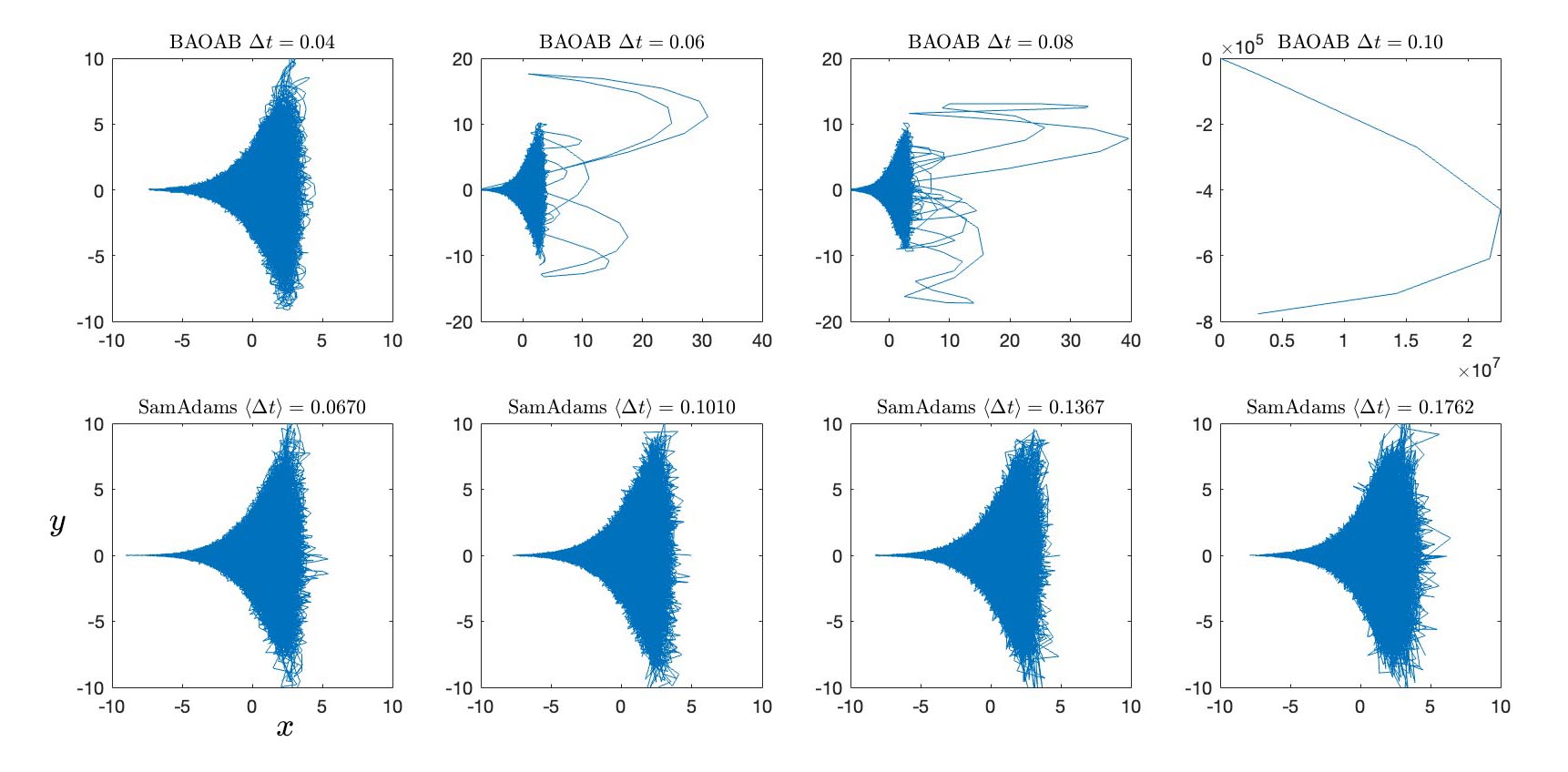}
\caption{\label{fig:multiD_compare} Comparison of trajectory graphs in the $\theta,x_1$ projection obtained for different methods and stepsizes.  \textbf{Top:} fixed stepsize BAOAB runs with $\Delta t$ increasing left to right until explosion is encountered at $\Delta t=0.1$. \textbf{Bottom:} results of different variable stepsize runs using SamAdams. While all these trajectories remain bounded, the last one shows some visible thickening in the $x$ direction, indicative of discretization bias.}   
\end{figure}
\end{center}
We examined the stability of fixed stepsize (BAOAB) Langevin dynamics by running trajectories with a series of increasing timesteps $\Delta t= 0.04, 0.06, 0.08, 0.1$. The system was initialised with a large enough value of $\theta$ ($\theta_0=5$) and all $x_i=0$ so as to avoid any initial high gradient, and we sample at $\gamma=1$.  The graphs of the $(\theta,x_1)$-projection are shown in the upper row of Fig. \ref{fig:multiD_compare}. All simulations involved $N=10^7$ steps.  We then ran SamAdams with various $\Delta \tau$, generating the figures shown in the second row, also using $10^7$ steps for each. The SamAdams monitor function is given by  $g(x,p) \equiv \frac{1}{100}\|\nabla U(x)\|$, and the filter kernel is $\psi=\psi^{(1)}$ with $m=0.01,\ M=1.0,\ r=1$. We use $\alpha=1$ for the $\zeta$ dissipation and set $\zeta_0=0$.
We label these runs by the mean stepsize used, which is reciprocally related to the computational work to reach a given fixed time.  We observe that BAOAB only leads to accurate sampling at its lowest stepsize. Given the mild confining potential, the BAOAB excursions to large positive $x$ at stepsize greater than $0.04$ represent high energy (far from equilibrium) events; in a more complicated model those trajectories would lead to blow-up, or, perhaps more seriously, subtle degradation of cumulative averages; here the confining potential quickly returns them to the domain of interest. By contrast, SamAdams is stable and retains the trajectories within a similar size and shape region of the $(\theta,x_1)$-plane, except for a slight increase in the width of the sampled region at the largest mean stepsize ($\langle \Delta t \rangle = 0.1762$). This is noteworthy given the much larger stepsizes SamAdams is able to use.
Since the largest stable mean stepsize of SamAdams is approximately four times the largest stable stepsize for Langevin dynamics with fixed stepsize, SamAdams promises substantial computational speedups. 

In Fig. \ref{fig:stepsizes_funnel}, we show a SamAdams trajectory with the points colored by the stepsize used.  We see that the smallest stepsizes are only used at the narrowest part of the funnel neck and the stepsize quickly resets as we leave these locations.  Stepsize distributions are shown in the right panel of Fig. \ref{fig:stepsizes_funnel}, with the mean stepsizes used in the four SamAdams runs of Fig. \ref{fig:multiD_compare}.   The parameters $m$ and $M$ partly govern the shape of these distributions, which are also defined by features of the problem itself.

\begin{figure}[!htb]
\includegraphics[width=0.32\textwidth]{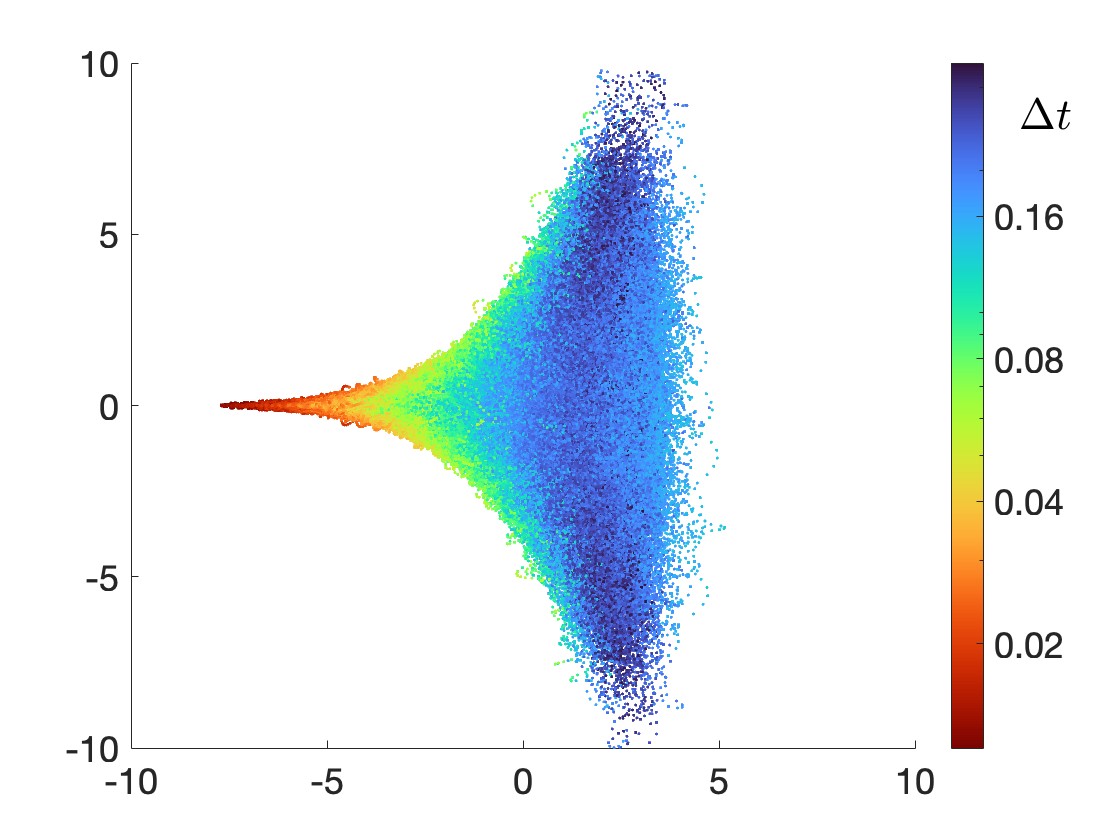}
\includegraphics[width=0.6\textwidth]{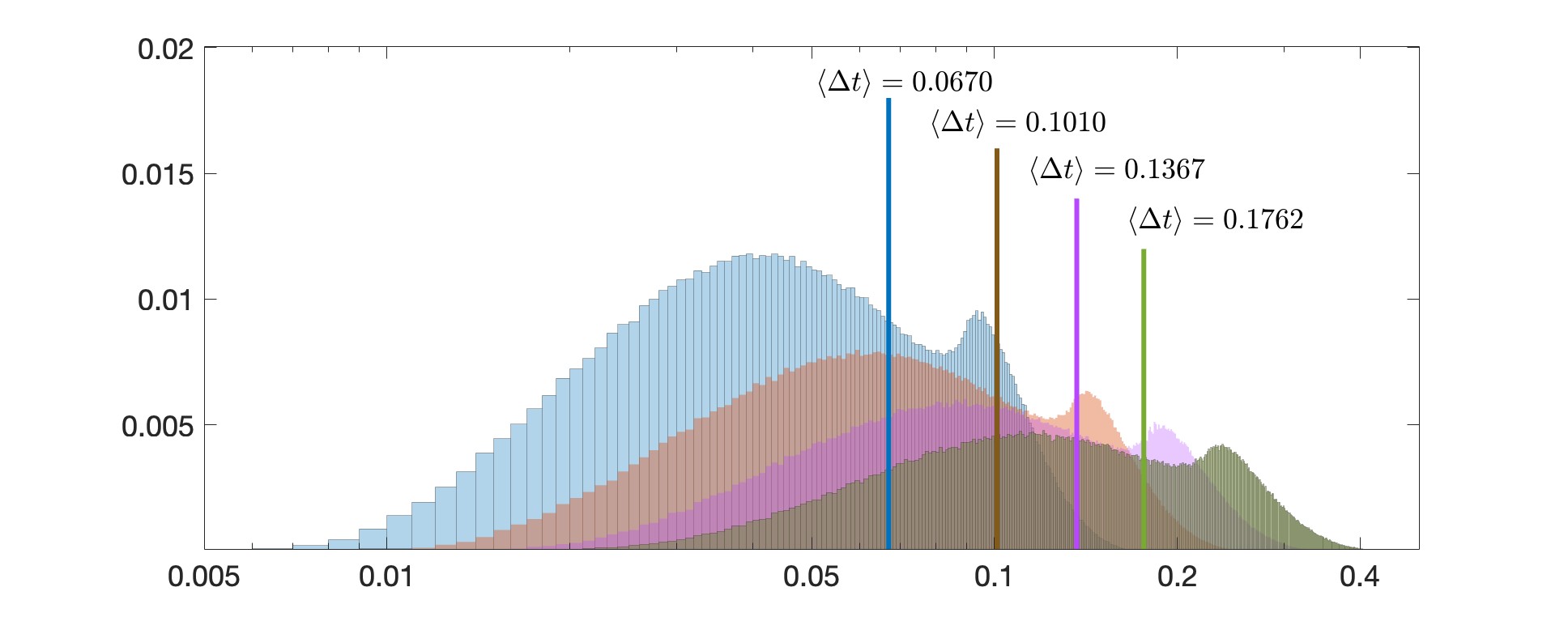}
\caption{\label{fig:stepsizes_funnel}\textbf{Left:} points along a trajectory of the 9D Neal funnel (here for $\langle \Delta t\rangle =0.1367$) are colored by stepsize used.  \textbf{Right:} the actual stepsize distributions for four SamAdams trajectories are shown; note that the stepsize scale at right is logarithmic, meaning that there is less density associated to the smaller stepsizes than is apparent from the areas of the respective bins.} 
\end{figure} 

An important question is whether the time-series corresponding to different timesteps explore the space with similar efficiency.  We settle this for the 9D funnel example in Fig. \ref{fig:funnel_autocorrelation} where we see that relative to the elapsed time, the trajectories diffuse at similar rates.\footnote{Computing autocorrelation functions and ESS requires first interpolating the non-uniformly spaced time-series data to a uniform mesh; this interpolation can be avoided when computing expectations by using the method of Section \ref{sec:computing_averages}.}    Effective sample sizes  per sample \cite{MCMC_practice} are as shown in Table \ref{tab:esss}.    These results indicate that the trajectories are similar in terms of the rates of exploration per unit time, but that the large mean stepsize SamAdams trajectories are much more efficient in terms of  the samples gained per timestep taken.

Mean log posterior values (from 10M step runs) for the different stepsizes are also given in the table and suggest that bias may be becoming more noticeable at the larger stepsize.  The result $-10.46$ for mean log posterior is accurate, verified using a small stepsize of $0.01$ and 1 billion steps.  Accuracy (and ESS) appear to fall off at the largest fixed stepsize, with $\Delta t=0.02$ being approximately optimal.  This can be compared to SamAdams with a mean stepsize of $\langle \Delta t\rangle = 0.0998$, indicating an improvement of over 400\% in sampling efficiency without any reduction in accuracy.

Note that with small stepsize $0.01$, fixed stepsize runs of 10M steps as in the table generate similar size errors (in this case due to Monte Carlo error due to the higher correlation of samples) to the largest variable stepsize integration (in that case,  due to  sampling bias), thus demonstrating a clear trade-off between sampling error and bias.
\renewcommand{\arraystretch}{1.2} 
\begin{table}
 \begin{tabular}{|l|c|c|c|c|c|c|c|}
 \hline
 & \multicolumn{3}{c|}{BAOAB} &
 \multicolumn{4}{c|}{SamAdams}\\
 \hline
 mean $\Delta t$ & 0.01  & 0.02  & 0.04  & 0.0664  & 0.0998  & 0.1336& 0.1676 \\
 \hline
 ESS/sample ($\times 10^{-4}$) & 6.6 & 13.2 & 22.4& 41.2 & 62.9 & 84.1 & 108.9\\
 \hline
mean log posterior ($10^7$ samples) & -10.55 & -10.47 & -10.42 & -10.46  & -10.46 & -10.50 & -10.53 \\ 
 \hline
 \end{tabular} 
 \caption{\label{tab:esss} Table showing the effective sample size per sample and the expected log posterior, for different variable stepsizes.}
\end{table}

\begin{figure}[!htb]
\includegraphics[width=0.49\textwidth]{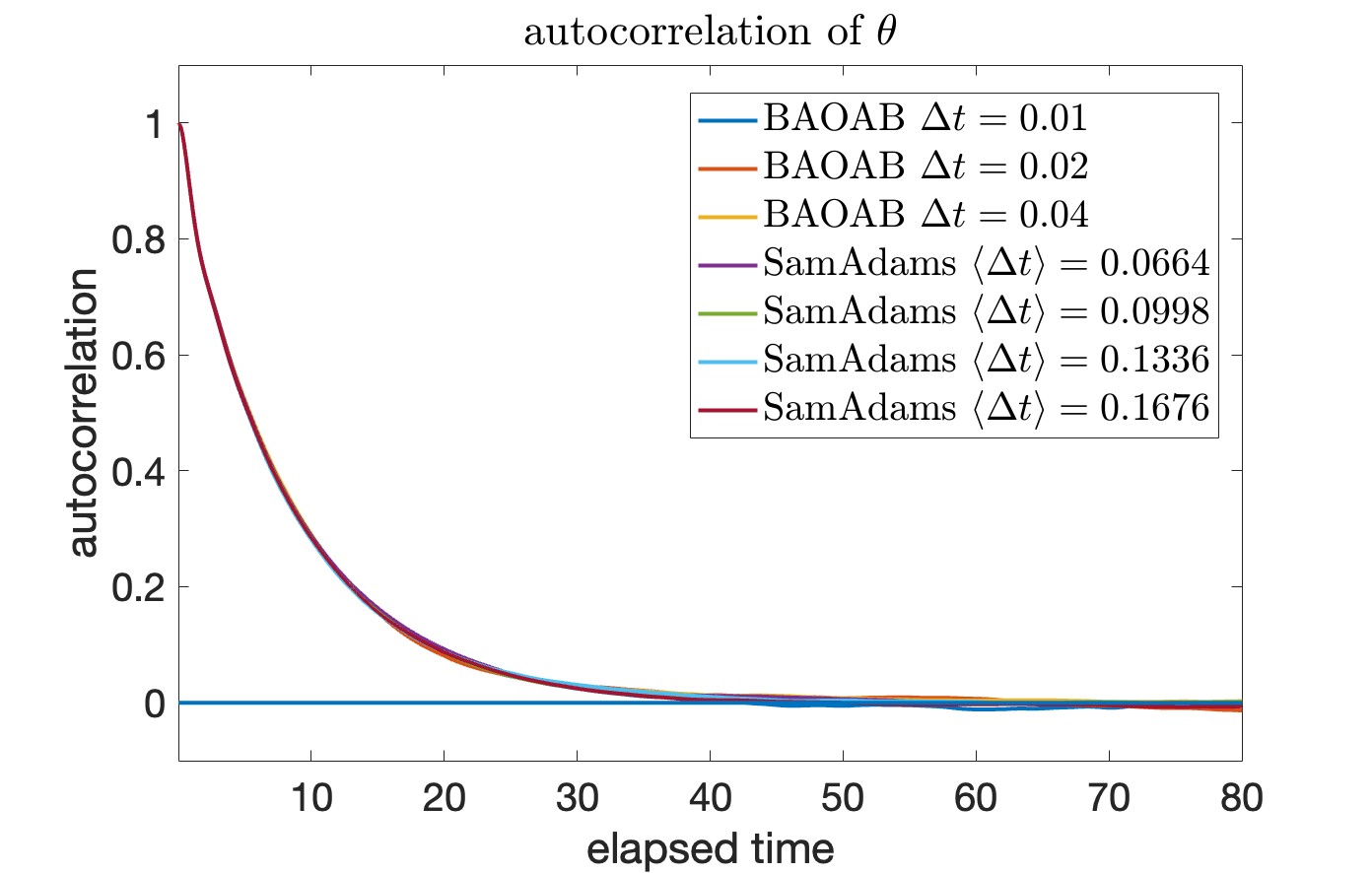}
\includegraphics[width=0.49\textwidth]{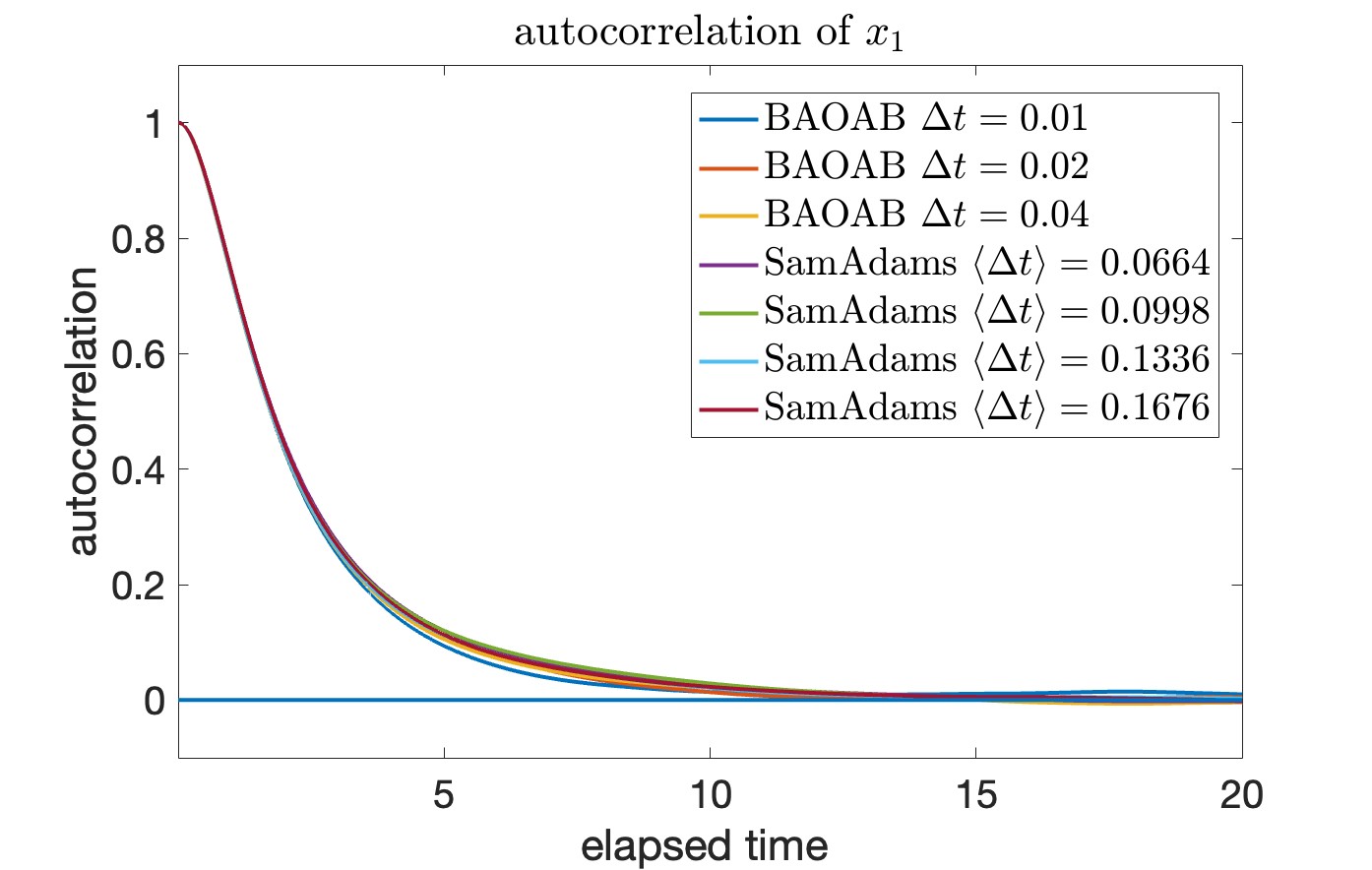}
\caption{\label{fig:funnel_autocorrelation} Autocorrelation functions with respect to time generated by different mean stepsizes.  \textbf{Left:} autocorrelation of $\theta$; \textbf{Right:} autocorrelation of $x_1$.  For each variable, the decay rates corresponding to different schemes are similar, regardless of the stepsize (and are similar for both the fixed and variable stepsize runs).} 
\end{figure}

\subsubsection{MNIST Image Classification on an MLP}
We next apply our new scheme to the MNIST digit classification dataset \cite{mnist}, a standard benchmark of computer vision applications consisting of 60,000 training and 10,000 test examples of handwritten digits, stored as $28 \times 28$ pixel grayscale images. The goal in this study is to illustrate the potential benefits of using adaptive-stepsize methods when exploring neural network loss landscapes compared to constant-stepsize methods.  As is customary, we normalize the data to mean 0.5 and standard deviation 0.5. We use a multi-layer perceptron (MLP) with two hidden layers with 800 and 300 nodes, respectively. While conventional neural network training is typically done with small batches, we use large batch sizes ($B=10,000$) in the experiments reported here to prevent additional gradient noise from obscuring the effect of stepsize adaptation. (For study of the effect of varying batch size in logistic regression, refer to Appendix \ref{sup:sec: LogisticRegression}.) 

 For the experiments in this and the next subsection, we use transform kernel $\psi^{(1)}$ with $m=0.1$, $\ M=10$, $\ r=0.25$.  The monitor function is taken as $g(x,p)=N_{D}^{-1}\|\nabla U(x) \|^2$ with $N_{D}=60,000$ the number of examples in the training dataset. We use $\Delta \tau=0.0002$, $\alpha=50$, and sample at $T=\gamma=1$. All trajectories are initialized through PyTorch's default (the momenta being drawn from their invariant Gaussian measure), with $\zeta_0=g(0,0)$.  
 
 Fig. \ref{fig:MNIST_spike_reduction} shows typical trajectories for BAOAB and SamAdams, where the former was run at the mean adaptive stepsize used by the latter.  We observe a large loss spike in the constant-stepsize scheme which is not present for the adaptive scheme.  Before and after the spike the dynamics seem to align to high degree, which implies that it is indeed the force-sensitive stepsize adaptation of SamAdams that prevents the spike from forming.
\begin{figure}[H]
\begin{centering}
\includegraphics[width=0.8\textwidth]{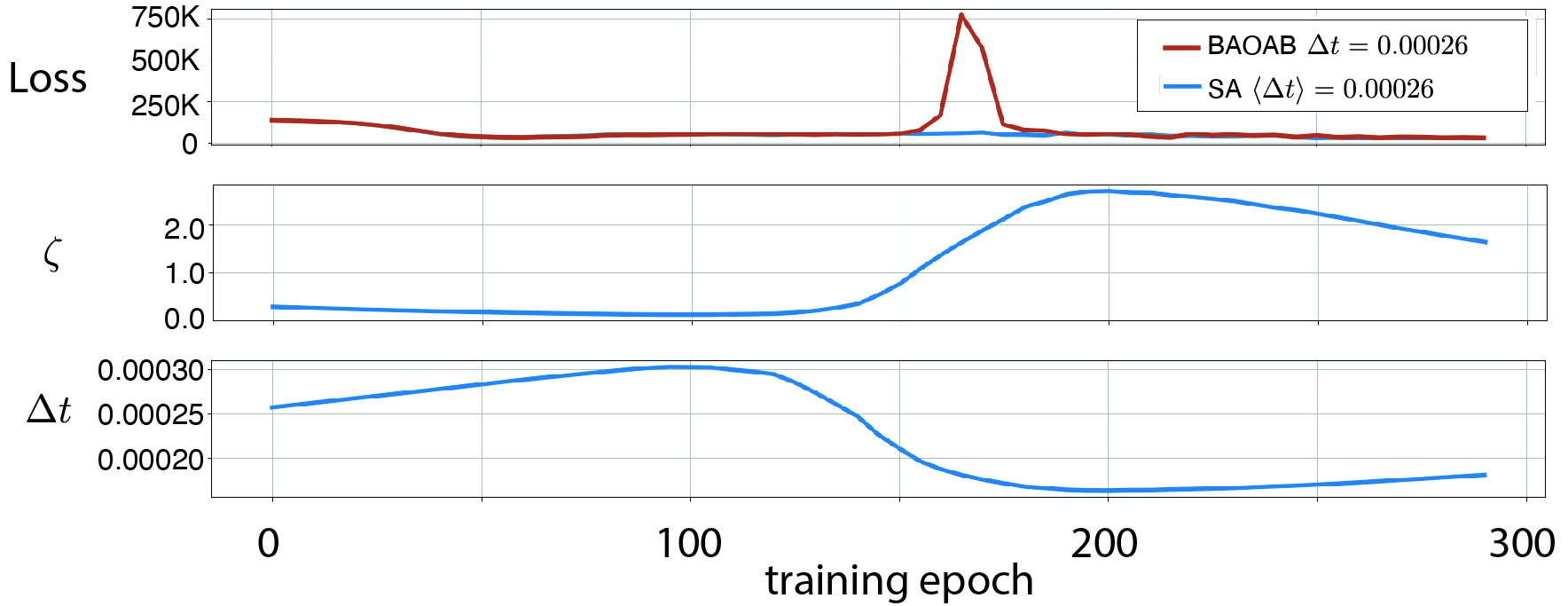} 
\end{centering}
\caption{\label{fig:MNIST_spike_reduction} Single trajectory results for an MLP on MNIST. \textbf{Top:} Training loss for both BAOAB and SamAdams. The constant stepsize scheme BAOAB becomes unstable during training, observable as a large spike. \textbf{Middle and bottom:} Corresponding $\zeta$-dynamics and resulting adaptive stepsize $\Delta t$ of SamAdams. }
\end{figure}

This is also clearly evident from the sudden increase of $\zeta$ and corresponding decrease of $\Delta t$ at that point. Observe also that we pick a comparatively large value of $\alpha$ on this example, which before and after the spike leads to a continuous increase of the stepsize because the forces in these regions are sufficiently small. 

While even the constant-stepsize scheme is able to recover from the instability, the appearance of spikes like this can lead to decreased performance in actual posterior sampling experiments, in which one averages over many different trajectories. To demonstrate this point, we draw 100 independent trajectories for both SamAdams and BAOAB, and average the resulting final accuracies. For SamAdams we use the same settings as in Fig. \ref{fig:MNIST_spike_reduction}. For BAOAB, we perform two experiments. In the first round, we set its learning rate to the mean of the stepsizes adopted by SamAdams, pooled from all 100 trajectories. In the second round, we take the mean of the pooled stepsizes again, but only consider the first 10\% of $\Delta t$ samples adopted by SamAdams on each trajectory. The latter setup is based on the idea that, for a fair comparison, the choice of stepsize for BAOAB should only be allowed to rely on information of the early stage of the SamAdams run. In contrast, using SamAdams' mean adaptive stepsize uses information of all the stepsizes used by SamAdams, i.e., knowledge of the force evolution across the whole trajectory. This information would usually not be available to someone wanting to set the stepsize of a constant-stepsize scheme, so using the mean adaptive stepsize for the BAOAB runs gives the benefit of the doubt to the fixed-stepsize method. 
At the same time, using the mean adaptive stepsize for BAOAB is the choice that leads to similar computational cost when run for the same number of iterations, which can also be seen as the basis for a 'fair' comparison. 

The final mean train and test accuracies together with 95\%-confidence intervals are given by Table \ref{tab:MNIST_table}.\footnote{Note: these runs were performed without time-averaging.}  We observe that SamAdams significantly outperforms both BAOAB setups. 
\renewcommand{\arraystretch}{1.2} 
\begin{table}[H]
    \centering
    \resizebox{1\textwidth}{!}{
    \begin{tabular}{|c|c|c|c|}
        \hline
                            & SamAdams & BAOAB (mean of all $\Delta t$) &   BAOAB (mean of first 10\% of $\Delta t$)\\ \hline
        Train Accuracy (\%) & $\textbf{94.0}\pm \textbf{0.1}$     &  $92.3\pm0.7$                           &  $92.9\pm0.5$   \\ \hline
        Test Accuracy  (\%) & $\textbf{93.6}\pm \textbf{0.1}$     &  $91.0\pm0.7$                           &  $91.5\pm0.6$    \\ \hline
    \end{tabular}
    }
    \caption{Final accuracies averaged over 100 independent trajectories. Mean values and 95\%-confidence intervals. For the hyperparameters used and a description of the difference of the two BAOAB runs, see main text.}
    \label{tab:MNIST_table}
\end{table}

\subsubsection{MNIST Image Classification on a CNN}\label{sec: MNIST_CNN}
As image classifiers usually adopt convolutional neural networks (CNN) rather than fully-connected ones, we repeat the experiment from the previous section on a simple CNN using three convolutional layers. The network architecture is given in Appendix \ref{sup:sec:CNN_architecture}. Unless explicitly restated, the hyperparameters are the same as in the previous section.

Fig. \ref{fig:MNIST_CNN_loss} (left) motivates the use of variable stepsizes.  It shows the results of a single SamAdams trajectory compared to three different BAOAB trajectories, each one corresponding to a different stepsize: (i) the smallest stepsize adopted by SamAdams, (ii) the mean stepsize, and (iii) the largest stepsize. This way of choosing the stepsizes of BAOAB to compare it to SamAdams is different from that used in the last section, where we took the point of view that a fair comparison between adaptive-stepsize and constant-stepsize schemes is made by considering the first 10\% of the obtained $\Delta t$ as well as the mean of $\Delta t$. 
\begin{figure}[htbp]
    \centering
    \includegraphics[width=\textwidth]{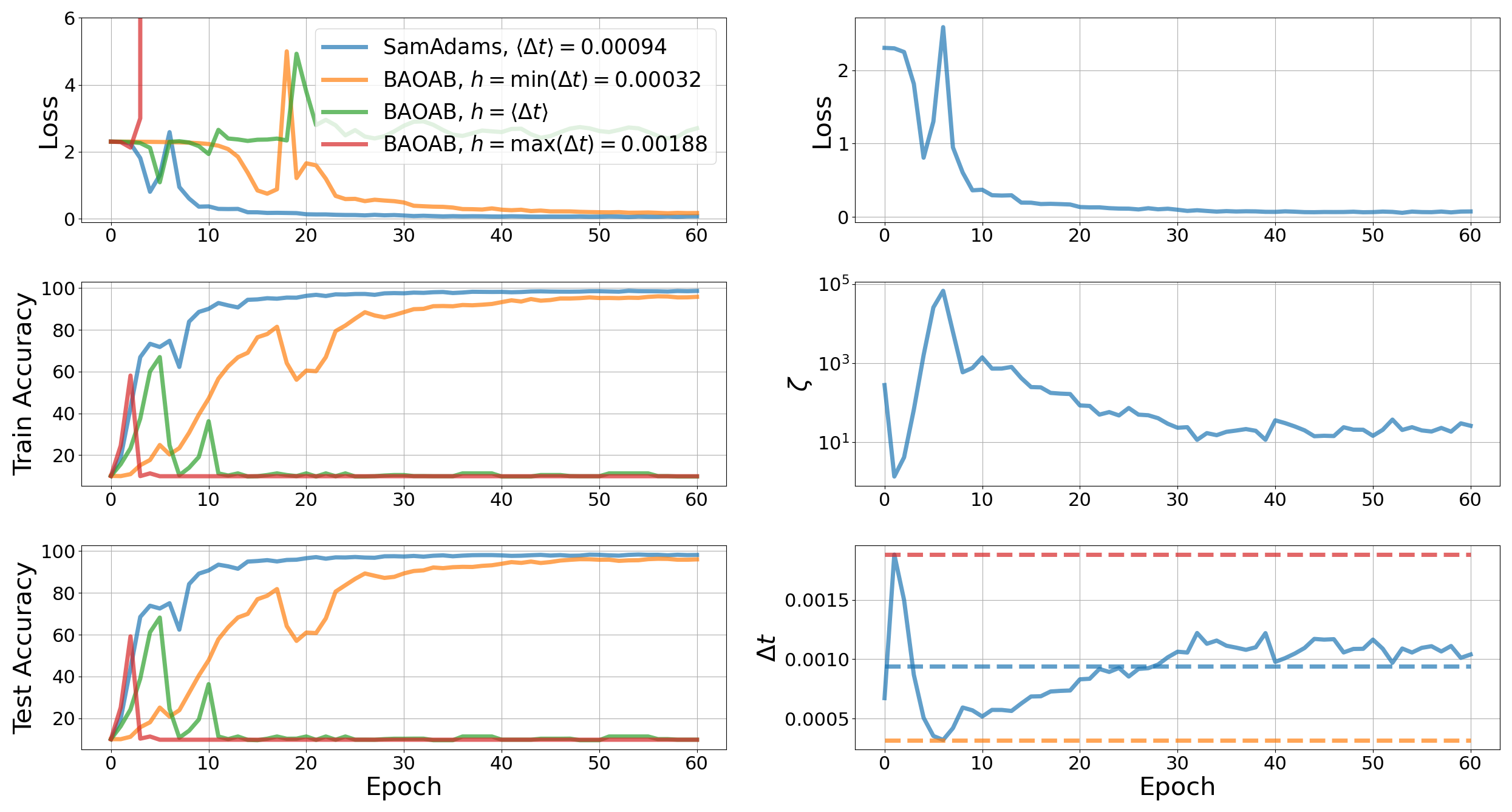}
    \caption{Training of a CNN on MNIST. \textbf{Left:} Train loss, train- and test accuracies. BAOAB was run at three different stepsizes: the smallest, largest, and mean stepsize used by SamAdams. \textbf{Right:} SamAdams results for loss (same as on the left), $\zeta$, and $\Delta t$. The dashed lines correspond to the stepsizes used by BAOAB. Hyperparameters: $\Delta \tau=0.002$, $\alpha=500$.}
    \label{fig:MNIST_CNN_loss}
\end{figure}
The two larger BAOAB stepsizes fail completely, while the smallest stepsize leads to reasonable results, but with a convergence speed (as measured in number of epochs, i.e., compute time)  substantially reduced in comparison to SamAdams, implying enormous computational speed-ups when using the latter. Fig. \ref{fig:MNIST_CNN_loss} (right) shows the evolution of $\zeta$ and $\Delta t$ of the SamAdams run compared to the loss. One observes that the algorithm reacts to the instabilities visible in the loss during the early training phase by rapidly reducing the stepsize $\Delta t$, well below the level of the learning rates used by the two unstable BAOAB runs. Only the BAOAB run using the smallest $\Delta t$ adopted by SamAdams remained stable, implying that small stepsizes are necessary to make it through the early stage of training. However, from epoch 5 onward SamAdams slowly increases its stepsize again, until it stabilizes at $\sim 0.0011$, more than three times the size of the stepsize used by the stable BAOAB run. This explains SamAdams faster convergence in loss and accuracy. 
Looking at the loss curves and the obtained $\Delta t$ by SamAdams, it appears that the trajectories start on a plateau (allowing for a rapid increase in $\Delta t$ at the very beginning), then descend through an irregular landscape (leading to rapid damping in $\Delta t$ and breakdown of two of the three BAOAB runs), before reaching a widened basin (allowing for moderate increase of $\Delta t$ again). 
This challenges the conventional wisdom in deep learning to use large learning rates at the start of training and to decrease of learning rates during later phases \cite{decreasing_learning_rates1,dlr2,dlr3}. Potentially, one could use adaptive stepsize schemes in place of conventional learning rate schedulers; we leave the detailed study of this intriguing possibility for future work.

We now examine accuracies obtained by the two schemes via posterior sampling and averaging. As in the previous section, we run 100 independent trajectories, initialized as before. The hyperparameters are the same as in Fig. \ref{fig:MNIST_CNN_loss}. Each trajectory is run for 60 epochs. Since the two larger fixed stepsize choices in Fig. \ref{fig:MNIST_CNN_loss} led to failure, we run BAOAB at three smaller stepsizes: One given by the mean of the smallest 10\% of the SamAdams stepsizes pooled from all 100 runs, $h=0.00046$, and two even smaller ones given by $h=0.0003$ and $h=0.0002$\footnote{We use $h$ here to denote BAOAB's stepsize to avoid confusion with the variable stepsize of SamAdams.}.  Fig. \ref{fig:MNIST_CNN_Ntraj_vs_t} shows the individual trajectories (without average), where we plot loss and accuracies against dynamics time $t$ and accuracy also against epoch count. The plots against time $t$ can be used to examine whether the spikes and instabilities in the trajectories for different stepsizes (see Fig. \ref{fig:MNIST_CNN_loss} again) occur at the same time, which would imply that the trajectories follow qualitatively similar paths, independently of the stepsize used. From this it would follow that the ability to use overall larger stepsizes in a stable manner would lead to efficiency gains. 
\begin{figure}[htbp]
    \centering
    \includegraphics[width=\textwidth]{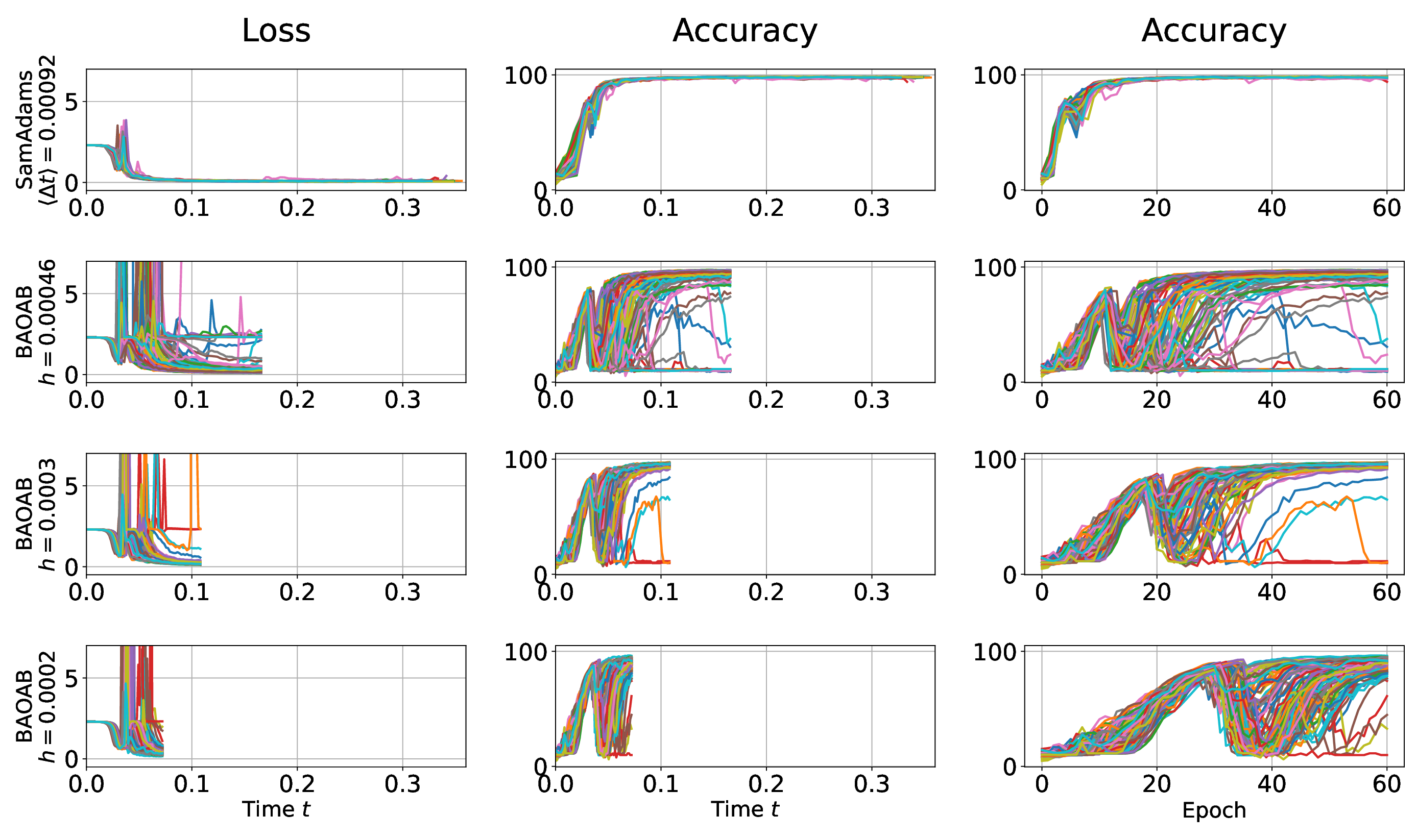}
    \caption{Training of a CNN on MNIST. The top row shows SamAdams results. The bottom three rows show BAOAB results for different stepsizes. Each figure shows 100 trajectories. \textbf{Left:} Loss vs. dynamics time $t$. \textbf{Middle:} (Test) accuracy vs. dynamics time $t$. \textbf{Right:} (Test) accuracy vs. epochs.}
    \label{fig:MNIST_CNN_Ntraj_vs_t}
\end{figure}
The first two columns in Fig. \ref{fig:MNIST_CNN_Ntraj_vs_t} show that all trajectories of both samplers and all examined stepsizes experience instabilities during the early stages of training, at $t \approx 0.03$.  This shows that the trajectories indeed follow similar paths\footnote{To be more precise, they encounter similar landscape topologies at similar times. They don't actually follow similar paths as they are initialized randomly and independently.}, roughly independently of the stepsize. Naturally, this independency will only hold for a certain stepsize range. The BAOAB runs in particular become unstable; while many of them are able to recover from this and begin to converge to high accuracies, some trajectories fail completely. The fraction of trajectories that are able to recover and their convergence speed both decrease with increasing stepsize. In contrast, SamAdams is able to  mitigate these issues as evident by the much smaller spikes in the curves, and all trajectories converge to high accuracies, barely being affected by the instabilities. This is particularly impressive given its mean stepsize being twice as large as the largest tested BAOAB stepsize. In order for BAOAB to traverse the unstable early-training region as unobstructed as SamAdams, it would require a smaller stepsize than the ones used in the figure, which would come at the cost of needing to execute many more iterations/epochs. The plots against epoch demonstrate the efficiency gains from using SamAdams, which provides high accuracy as early as epoch 20. 

We now compute posterior averages using the samples obtained from the trajectories. Fig. \ref{fig:MNIST_CNN_Nraj_accus} shows the accuracies of the different samplers and stepsizes at certain training epochs. At each such epoch, a time average is performed over the previous 10 epochs (where each epoch yields one sample), followed by the average over the 100 trajectories. The shaded areas denote 95\%-confidence intervals (based on t-statistics). 
\begin{figure}[htbp]
    \centering
    \includegraphics[width=\textwidth]{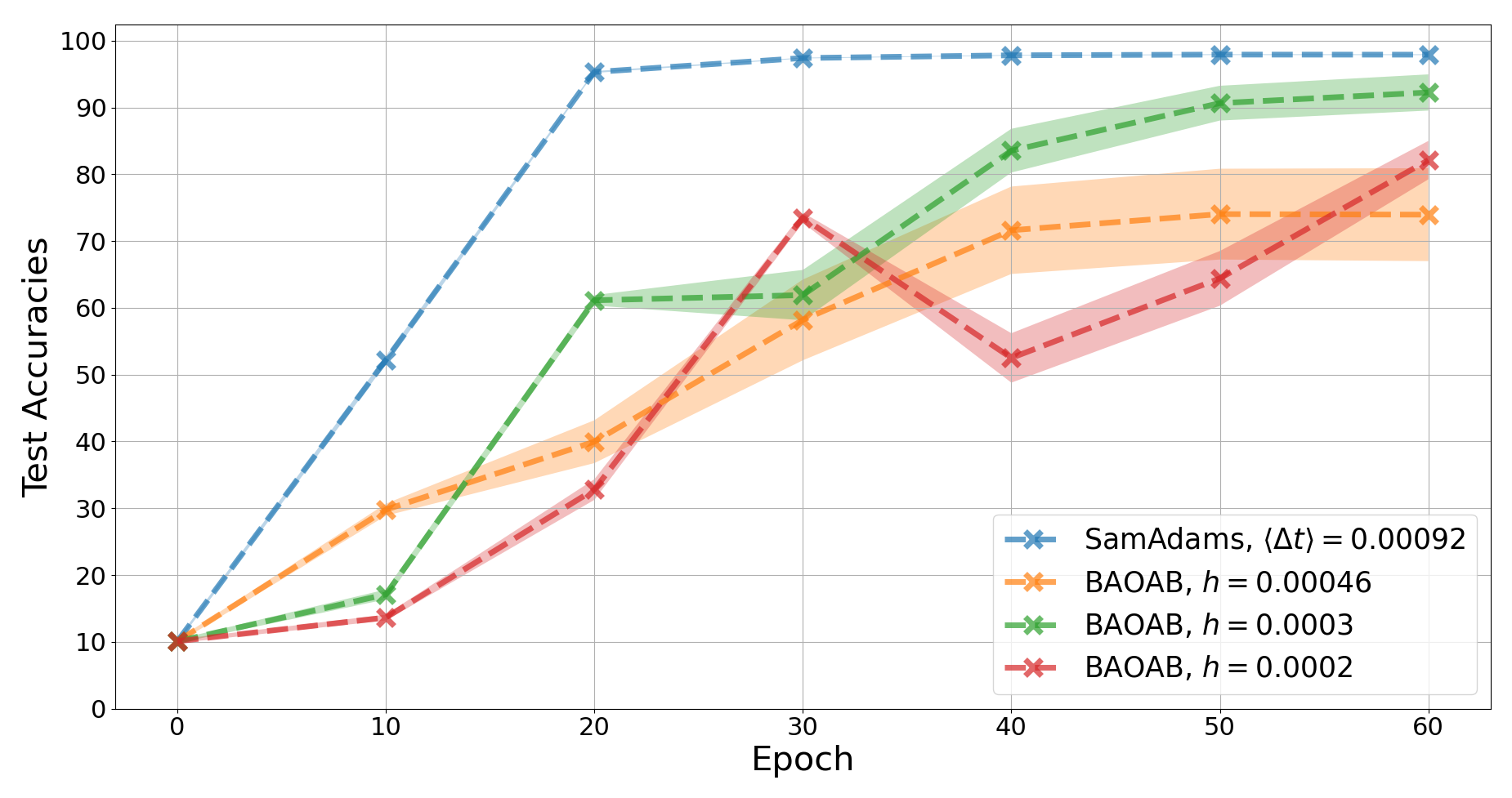}
    \caption{Test accuracies of a CNN on MNIST obtained through Bayesian sampling. Every 10 epochs, the samples of the last 10 epochs were averaged (time average) and the results were averaged across 100 independent trajectories for each of the curves. The shaded areas denote 95\%-confidence intervals.}
    \label{fig:MNIST_CNN_Nraj_accus}
\end{figure}
We see that SamAdams converges to high accuracies much faster than all BAOAB runs, none of which reach comparable accuracies even at the end of training. Furthermore, due to several BAOAB trajectories failing (as shown in Fig. \ref{fig:MNIST_CNN_Ntraj_vs_t}) the variance of the BAOAB results are much larger than those of SamAdams; the extremely narrow confidence interval, the shaded area behind the blue curve, is almost too narrow to see. 
The final accuracies and confidences are shown by Table \ref{tab:MNIST_CNN_table}.
\renewcommand{\arraystretch}{1.2} 
\begin{table}[htbp]
    \centering
    \resizebox{0.5\textwidth}{!}{
    \begin{tabular}{|@{\hspace{2pt}}l|c@{\hspace{2pt}}|}
        \hline
        Sampler, Stepsize                               & Test Accuracy  (\%) \\
        \hline
        SamAdams, $\langle\Delta t \rangle=0.00092$     & $\textbf{97.90}\pm\textbf{0.06}$ \\
        BAOAB, $h=0.00046$                              & $73.95\pm6.96$ \\
        BAOAB, $h=0.0003$                               & $92.24\pm2.71$ \\
        BAOAB, $h=0.0002$                               & $82.11\pm2.87$ \\
        \hline
    \end{tabular}
    }
    \caption{Mean accuracies and 95\%-confidence intervals obtained by sampling the posterior of a CNN on the MNIST dataset. Values correspond to the epoch 60 results in Fig. \ref{fig:MNIST_CNN_Nraj_accus}. }
    \label{tab:MNIST_CNN_table}
\end{table}
The particularly large variance for BAOAB at $h=0.00046$ comes from the large number of trajectories that fail as a consequence of the early-stage instabilities, see again Fig. \ref{fig:MNIST_CNN_Ntraj_vs_t} and also Fig. \ref{fig:MNIST_CNN_histos} which shows the histograms of obtained accuracies after 60 epochs for that BAOAB setting compared to SamAdams.

These experiments demonstrate that the enhanced stability and computational efficiency of SamAdams observed on the artificial planar problems of the previous sections extend to complex models on real-world data, and even establish that the mechanism is similar.

We note that the benefits of the adaptive stepsizes when sampling the posterior of the CNN of this section mainly stem from the early stages of training, during which the sampler equilibrates from the random initial conditions to lower regions of the loss landscape. There are alternative approaches to BNN sampling. For example, one can perform BNN experiments by first running an optimizer to find a local mode, then running an MCMC scheme starting from within that mode to sample the mode vicinity. This can either be done by adding a confining potential as in \cite{SymmetricMiniBatchSplitting} or by using low temperatures or stepsizes \cite{cold_posteriors,zhang2020cyclicalstochasticgradientmcmc}. Our adaptive stepsize framework can easily be used in conjunction with the ideas in those works. In our experiments, from looking at the comparably small fluctuation of $\zeta$ and $\Delta t$ in the later stages of training (see Fig. \ref{fig:MNIST_CNN_loss}), the necessity for adaptive stepsizes in the mode vicinity is much smaller than during the early stages of equilibration. It is easily conceivable, however, that more complex models and datasets lead to basin topologies which would benefit from adaptive stepsizes as well. We leave the exploration of larger models and different approaches to BNN sampling with adaptive stepsizes for future work.

\begin{figure}[htbp]
    \centering
    \includegraphics[width=\textwidth]{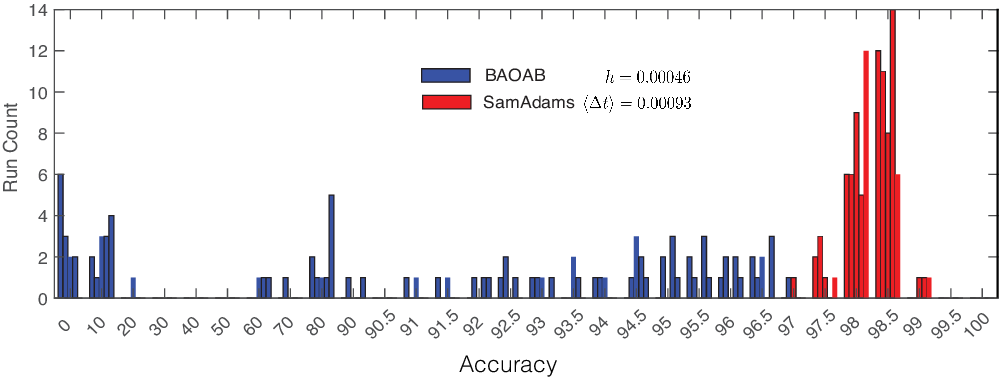}
    \caption{Histograms of time-averaged train accuracies for 100 independent trajectories of SamAdams (red) and BAOAB (blue) corresponding to the results in Table \ref{tab:MNIST_CNN_table}. BAOAB was run at $h=0.00046$.}
    \label{fig:MNIST_CNN_histos}
\end{figure}

\section{Conclusion}
We have presented a flexible integration framework for adaptive-stepsize Langevin sampling algorithms based on an auxiliary monitor variable. In particular, we have shown how the SamAdams algorithm, inspired by the Adam optimizer, shows superior behavior in terms of both stability and convergence speed compared to fixed stepsize alternatives. While we have provided various numerical experiments, we believe that there are many more settings in which the force-sensitive stepsize adaptation can greatly enhance sampling performance. 

The method can be adapted to large-scale Bayesian machine learning, and is likely to show advantages in  relation to models that are typically trained using Adam and its derivatives in the deep learning context \cite{training_survey},  whether for natural language processing \cite{openai}, diffusion models \cite{training_diffusion_models} or some other type of machine learning application.  As sampling and BNN frameworks are introduced to address a wider range of AI challenges, we expect algorithms such as the one described here will be in high demand.  Manual adjustment of the stepsize (learning rate scheduling) is often used in machine learning applications \cite{decreasing_learning_rates1,dlr2,dlr3}; the flexible nature of our framework suggests the possibility of  an automated approach which can simplify training workflows if the right choice of monitor function and other aspects can be identified (which may not always be as simple as the norm of the gradient). The strong relation between the stepsize used in training, the batch size (amount of gradient noise), and the generalization error makes the new method interesting for settings with batch sizes that vary in time, such as active learning \cite{active_learning}. Although it is not the main target of this work, we believe that SamAdams (or a similar method based on the adaptation framework presented in this article) might be of interest to computational scientists simulating physical models in which forces may increase preciptously during integration,  often requiring the use of small stepsizes compared to the long simulation times that have to be realized.

Finally, we note that the SamAdams framework can easily be combined with other sampling procedures based on SDE discretization, since, as we have written it in Algorithm \ref{alg:AdamSampler}, the timestep adaptation is implemented separately from the propagation of state variables. It could also be combined with debiasing techniques (see \cite{UBUBU}) to produce unbiased estimates from the target measure whilst still avoiding Metropolis-Hastings accept-reject steps.
\section*{Acknowledgements}
This research was supported by the MAC-MIGS Centre for Doctoral Training (grant EPSRC EP/S023291/1).
The authors wish to acknowledge Katerina Karoni for helpful discussions at an early stage of the project and Daniel Paulin for advice on the neural network studies.  We  thank Gilles Vilmart for generous assistance in establishing the weak convergence of our method and Michael Tretyakov for helpful input on numerical stability of SDE discretizations.

\bibliographystyle{plain}
\bibliography{Arxiv.bbl}
\appendix
\section{General Form of the SamAdams Algorithm}\label{sup:sec:general_SamAdams}
It is possible to SamAdams-ize any integrator by wrapping it in Z-steps.   For a standard form SDE
\begin{equation} \label{eq:gen_sde}
{\rm d} x_t = a(x_t, t) {\rm d} t + \sigma(x_t,t) {\rm d}W_t,
\end{equation}
we can introduce auxiliary variable $\zeta$ controlled by
\[
\frac{{\rm d} {\zeta}_{\tau}}{{\rm d} \tau} = f(x_{\tau}, \zeta_{\tau}),
\]
with, for example, $ f(x_{\tau},\zeta_{\tau}) = -\alpha \zeta_{\tau} + g(x_{\tau})$, then introduce a Sundman transformation 
\[
{\rm d} t = \psi(\zeta) {\rm d}\tau,
\]
with $\psi$ a suitable uniformly positive, bounded, smooth function.  Finally the rescaled equations corresponding to  (\ref{eq:gen_sde}) become
\begin{eqnarray} \label{eq:gen_sde_rescaled-1}
{\rm d} x_{\tau} & = & \psi(\zeta_{\tau}) a(x_{\tau}, t(\tau)) {\rm d}{\tau}  + \sigma(x_{\tau},t(\tau))\sqrt{\psi(\zeta_{\tau})} {\rm d}W_{\tau},\\
{\rm d} \zeta_{\tau} & = & f(x_{\tau}, \zeta(\tau)){\rm d} \tau,\\
{\rm d} t & = & \psi(\zeta_{\tau}) {\rm d}\tau.
\label{eq:gen_sde_rescaled-3}
\end{eqnarray}

In case the simple relaxation equation $f(x,p,\zeta)=-\alpha \zeta + g(x)$ is adopted, one could adopt any fixed stepsize integrator $\hat{\Phi}_{\Delta t}$, as the basic method and turn it into a variable stepsize scheme which might be denoted Z$\hat{\Phi}$Z outputting states $\{x_n\}$ and weights $\{\mu_n\}$ computed via the following step sequence:
\begin{eqnarray} \label{eq:ZPhiZ_1}
\zeta_{n+1/2} & = & \rho^{1/2} \zeta_n + \alpha^{-1}(1-\rho^{1/2})g(x_n),\\
\Delta t_{n+1} & = & \psi(\zeta_{n+1/2}) \Delta \tau,\\
x_{n+1} & = & \hat{\Phi}_{\Delta t_{n+1}}(x_n),\\
\zeta_{n+1} & = & \rho^{1/2} \zeta_{n+1/2} + \alpha^{-1}(1-\rho^{1/2})g(x_{n+1}),\\
\mu_{n+1} & = & \psi(\zeta_{n+1}) \label{eq:ZPhiZ_5}, 
\end{eqnarray}
with $\rho = \exp(-\alpha \Delta \tau)$ (compare Algorithm \ref{alg:AdamSampler}). 
In case $g$ depends on a force, that force calculation performed at the end of a step can be reused in the following step during the initial $Z$-half-step and in the state propagation.  Like Alg. \ref{alg:AdamSampler}, this method uses the two half Z-steps to produce more accurate weights at the step endpoints.
\section{Adam Dynamics}\label{sup:sec:adam_derivation}
In \cite{da2020general} and \cite{KaterinaThesis} it was demonstrated that the
Adam optimizer can be interpreted as the Euler discretization of a
certain system of ODEs, given by
\begin{align} 
\diff x_{t} &= \frac{p_{t}}{\sqrt{\zeta_{t} + \epsilon}} \diff t          \label{eq:x_dot_adam}, \\ 
\diff p_{t} &= - \nabla U(x_{t})\diff t - \gamma p_{t} \diff t                   \label{eq:p_dot_adam}, \\ 
\diff \zeta_{t} &= [\nabla U(x_{t})]^2\diff  t - \alpha \zeta_{t} \diff t .         \label{eq:zeta_dot_adam} 
\end{align} 
Here we have $x_{t},p_{t},\zeta_{t}\in\mathbb{R}^{d}$,
$\alpha>0$, $\gamma>0$, and the algebraic operations are to be understood
elementwise. Applying an Euler discretization with step size ${\Delta t}$
leads to 
\begin{align}   
p_{n+1}       & =  (1- \gamma {\Delta t}) p_n- {\Delta t} \nabla U(x_n),  \\ 
\zeta_{n+1}   & =  (1- \alpha {\Delta t}) \zeta_n + {\Delta t} [\nabla U(x_n)]^2, \\
x_{n+1}       & =   x_n + {\Delta t} \bigg(\frac{p_{n+1}}{\sqrt{\zeta_{n+1} + \epsilon}}\bigg) .
\end{align}
Setting $\beta_{1}:=1-\gamma{\Delta t}\Rightarrow$ ${\Delta t}=(1-\beta_{1})/\gamma$
and $\beta_{2}:=1-\alpha{\Delta t}\Rightarrow{\Delta t}=(1-\beta_{2})/\alpha$,
this becomes 
\begin{align}
p_{n+1}     & = \beta_1 p_n - \frac{1-\beta_1}{\gamma} \nabla U(x_n),  \\ 
\zeta_{n+1} & = \beta_2 \zeta_n + \frac{1-\beta_2}{\alpha}  [\nabla U(x_n)]^2, \\
x_{n+1}     & = x_n + {\Delta t} \bigg(\frac{p_{n+1}}{\sqrt{\zeta_{n+1} + \epsilon}}\bigg). 
\end{align} 
Multiplying the $p$-equation by $\gamma$ and the $\zeta$-equation
by $\alpha$, setting $\tilde{p}:=-\gamma p$ and $\tilde{\zeta}:=\alpha\zeta$,
we obtain
\begin{align} 
\tilde{p}_{n+1}     & = \beta_1 \tilde{p}_n + (1-\beta_1) \nabla U(x_n) \label{sup:eq:p_n+1_adam},  \\ 
\tilde{\zeta}_{n+1} & = \beta_2 \tilde{\zeta}_n + (1-\beta_2) [\nabla U(x_n)]^2 \label{sup:eq:zeta_n+1_adam}, \\ 
x_{n+1}             & = x_n - \widetilde{{\Delta t}} \bigg(\frac{\tilde{p}_{n+1}}{\sqrt{\tilde{\zeta}_{n+1} + \tilde{\epsilon}}}\bigg) \label{sup:eq:x_n+1_adam},
\end{align}
where we also set $\widetilde{{\Delta t}}:=\frac{{\Delta t}\sqrt{\alpha}}{\gamma}$
and $\tilde{\epsilon}:=\alpha\epsilon$. Note that the new momentum variable is sign-flipped which leads to a plus sign in front of the force in the  $\tilde{p}_{n+1}$-equation and a minus sign in front of the momentum in the $x_{n+1}$ equation, which is conventionally reversed for Langevin dynamics-based schemes. Apart from an additional scaling of $\tilde{p}_{n+1}$ and $\tilde{\zeta}_{n+1}$ (see the end of this section), \eqref{sup:eq:p_n+1_adam}-\eqref{sup:eq:x_n+1_adam} form the Adam scheme as introduced in \cite{Adam} and employed by widely used machine learning frameworks such as PyTorch \cite{pytorch} or Tensorflow \cite{tensorflow}. 
According to \eqref{sup:eq:zeta_n+1_adam}, the $n$-th iterate of $\tilde{\zeta}$ is given by 
\begin{equation}\label{sup:eq:adam_average_zeta}
\tilde{\zeta}_n=\beta_2^n\tilde{\zeta}_0+(1-\beta_2)\sum_{i=0}^{n-1}\beta_2^{n-1-i}[\nabla U(x_i)]^2, 
\end{equation}
which takes the form of an exponentially weighted average over the squared gradient components, where smaller weights are assigned to values further in the past.  Since Adam is an optimizer, not a sampler, it is inherently deterministic. However, if one assumes the loss landscape itself is subject to noise (e.g., from dataset subsampling), such that in any given iteration, $U(x_t)$ is an unbiased estimator of the true loss at that point, $\tilde{U}(x_t)$, one observes that \eqref{sup:eq:adam_average_zeta} estimates the second moments of the gradients, i.e., the uncentered variances\footnote{This is only approximately true as the distribution of $U(x_t)$ will be different for different points $x_t$ along the trajectory. However, as mentioned in Sec. 3 in \cite{Adam}, due to the decaying weights assigned to gradients further in the past in the sum in \eqref{sup:eq:adam_average_zeta}, the error can be assumed to be small.}. Similarly, the  momentum accumulates an estimate of the first moment, i.e., the expectation of the gradient. Without gradient noise, the moving averages $\tilde{p}_n$ and $\tilde{\zeta}_n$ can be interpreted as estimates of the averaged loss gradient and gradient variance.
In both cases, the parameter update \eqref{sup:eq:x_n+1_adam} approximates
\begin{equation}
x_{n+1}             = x_n - \widetilde{{\Delta t}} \bigg(\frac{\mathbb{E}\Big[\nabla U(x_n)\Big]}{\sqrt{\text{Var}_0(\nabla U(x_n)) + \tilde{\epsilon}}}\bigg),
\end{equation}
where the $\text{Var}_0(X)$ denotes the uncentered variance of $X$. One therefore obtains larger steps in regions of large gradients but small gradient variances (curvatures). 

From \eqref{sup:eq:x_n+1_adam} it can be seen that the adaptive stepsize in step $n$ given by 
\begin{equation}
\Delta t_n=\frac{\widetilde{{\Delta t}}}{\sqrt{\tilde{\zeta}_{n}+\tilde{\epsilon}}}, 
\end{equation}
such that  larger values of the moving averages of the squared gradients lead to smaller adaptive stepsizes.

We note that Adam typically contains two additional steps: Before inserting the momentum $\tilde{p}_{n+1}$ and stepsize-scaling variable $\tilde{\zeta}_{n+1}$ into \eqref{sup:eq:x_n+1_adam}, they are rescaled according to 
\begin{equation}
\hat{p}_{n+1}=\frac{\tilde{p}_{n+1}}{1-\beta_1^{n+1}}, \hspace{1cm} 
\hat{\zeta}_{n+1} = \frac{\tilde{\zeta}_{n+1}}{1-\beta_2^{n+1}}.
\end{equation}
This is done in order to remove the bias due to the initial conditions from the estimates of the gradient moments (see Algorithm 1 and Sec. 3 in \cite{Adam}). Since both $\beta_1^n\to0$ and $\beta_2^n\to0$ for $n \to \infty$, these steps can often be skipped in practice, which is why we will not consider them in the rest of this work.

\section{Further Details of Stepsize Control Mechanism}\label{sup:sec:stepsize_variation}
The stepsize adaptation mechanism used by the SamAdams process presented in Sec. \ref{sec: SamAdams} in the main text is based on three components. The first is the evolution of the force-sensitive $\zeta$-variable, given by 
\begin{equation}\label{sup:eq:zeta_dynamics}
\diff \zeta_{\tau}= - \alpha \zeta_{\tau} \diff \tau + \Omega^{-1} \|\nabla U(x_{\tau}) \|^s\diff \tau,  
\end{equation}
with the two hyperparameters $\alpha>0$ and $\Omega>0$, whose role will be expanded upon below. The second component is given by the choice of a Sundman transform kernel $\psi$ as a function of $\zeta$. For example, we set $\psi(\zeta)\equiv\psi^{(1)}(\zeta)$ with $\psi^{(1)}$ from Sec. \ref{sec: SamAdams}, given by
\begin{equation}
    \psi^{(1)}(\zeta)=m\frac{\zeta^r+M}{\zeta^r+m}, \hspace{10pt} \text{with } \hspace{2pt} 0<m<M<\infty.
\end{equation}
As mentioned in the main text, this form resembles the term used by Adam but allows for more flexibility via the scaling hyperparameter $r>0$ and stability due to its boundedness, $\psi^{(1)}(\zeta)\in (m,M)$ for all $\zeta>0$ (note that once $\zeta$ is initialized to $\zeta_0>0$, it will always remain positive due to  \eqref{sup:eq:zeta_dynamics}). The third component is the time rescaling relationship $\Delta t=\psi(\zeta)\Delta \tau$, which  scales the constant stepsize in virtual time, $\Delta \tau$,  with the transform kernel evaluated at $\zeta$  to yield the adaptive stepsize in real time, $\Delta t$. Fig. \ref{sup:fig:sundman_transform}
 shows the transform kernel $\psi^{(1)}(\zeta)$ with $m=0.1$ and $M=10$ (the values that were used in most of our experiments) and for various $r$. Since $m$ and $M$ are the bounds of $\psi^{(1)}$, they also specify the bounds on the adaptive stepsize, $\Delta t \in (m\Delta \tau, M\Delta \tau)$. As denoted by the black arrows in the figure, larger forces tend to increase the value of $\zeta$ which in turn decreases the value of $\psi^{(1)}(\zeta)$, leading to a decrease in adaptive stepsize.
 The red line denotes the value of $\zeta$ for $\psi^{(1)}(\zeta)=1$, i.e., where $\Delta t=\Delta \tau$. It thus gives the boundary between the $\zeta$-regions where $\Delta t$ is smaller or greater than $\Delta \tau$. For $m=1/M$, this boundary is exactly at $\zeta=1$, i.e., $\psi^{(1)}(\zeta)=1$ (the red dot). Note that this changes for different $m$ and $M$. The exponent $r$ governs the sensitivity of the overall mechanism.
\begin{figure}[htb!]
\begin{center}
\includegraphics[width=5in]{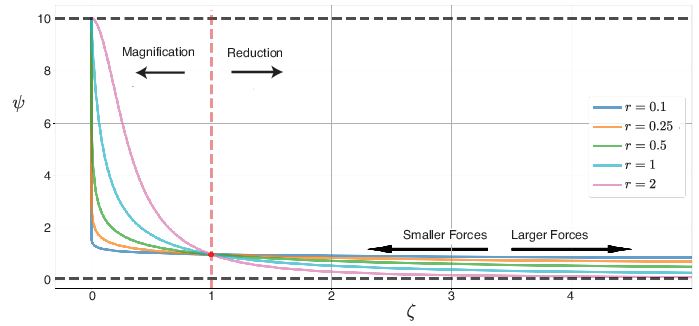}
\caption{Time transform kernel $\psi (\zeta)\equiv \psi^{(1)}(\zeta)$ as a function of $\zeta$ for $m=0.1$,  $M=10$, and different exponents $r$. The black dashed lines denote the bounds $m$ and $M$. The red point at (1,1) and red dashed line separate the $\zeta$-region in which the basic stepsize $\Delta \tau$ is magnified from the region where it is reduced (i.e., the regions of $\psi^{(1)}(\zeta)>1$ and $\psi^{(1)}(\zeta)<1$, resp.). The arrows denote the direction of the  $\zeta$-evolution dependent on the experienced forces $\|\nabla U \|^s$.}
\label{sup:fig:sundman_transform}
\end{center}
\end{figure}

How large the forces need to be to lead to a decrease of $\Delta t$ below $\Delta \tau$ (for a fixed set of transform kernel parameters $r,m,M$) can be controlled with the hyperparameters $\alpha$ and $\Omega$ in \eqref{sup:eq:zeta_dynamics}. From looking at the solution of \eqref{sup:eq:zeta_dynamics} given by (see\eqref{eq:zeta_solution} in the main text) 
\begin{equation}
\zeta(\tau)=e^{-\alpha\tau}\zeta(0)+\Omega^{-1}\int_{0}^{\tau}e^{-\alpha(\tau-s)}\|\nabla U(x_s)\|^s{\rm d}s,\label{sup:eq:zeta_solution}
\end{equation}
we see that that $\zeta$ is identical to an exponentially weighted moving average over the past force magnitudes raised to the power $s$, where $\alpha$ governs how strongly past values are suppressed. Pulling $\Omega^{-1}$ into the integral, it is clear that we actually average over $\Omega^{-1}\|\nabla U\|^s$,  which means that $\Omega^{-1}$ linearly scales the obtained $\zeta$ values, directly influencing the size of $\Delta t$. For a more concrete view on the influence of the two parameters, we look at the discretized version of \eqref{sup:eq:zeta_solution} employed by our splitting integrators described in Sec. \ref{sec: integration} in the main text, given by 
\begin{equation}\label{sup:eq:Z-step}
\zeta_{n+1}=\Phi_{\Delta\tau}^{{\rm Z}}(\zeta_n,x_n)=e^{-\Delta\tau\alpha}\zeta_n+\frac{1}{\Omega \alpha}(1-e^{-\Delta\tau\alpha})\| \nabla U(x_n)\|^s.
\end{equation}
From this, it follows that the $n$-th iterate is given by
\begin{equation}\label{sup:eq:zeta_iterates}
\zeta_n=p^n\zeta_0+q\bigg(\sum_{i=1}^np^{n-i} \|\nabla U(x_i)\|^s\bigg),
\end{equation}
with $p:=\exp(-\Delta \tau \alpha)$,  $q:=\frac{1}{\Omega \alpha}(1-p)$. This form is equivalent to the one used in Adam, see \eqref{sup:eq:adam_average_zeta}. As mentioned in the main text, the fundamental difference is that in Adam the $\tilde{\zeta}$ in \eqref{sup:eq:adam_average_zeta} is a vector and the squaring of the force is an elementwise operation (leading to an adaptive stepsize per degree of freedom), which in our case is replaced by the Euclidean norm of the force, leading to a scalar $\zeta$ (and hence a single adaptive stepsize for all degrees of freedom). 
We confirm again that $\alpha$ through $p$ governs the influence of past forces, suppressing the ones further in the past more strongly. Larger values for $\alpha$ lead to smaller weights assigned to past forces leading to less "memory" in the $\zeta$-dynamics. 
The parameter $q$ linearly scales the contribution of the sum and thus the overall magnitude of $\zeta$ (and thus the obtained adaptive stepsize $\Delta t$). 
Note that $q$ depends on both $\alpha$ and $\Omega$, which means varying $\alpha$ will not only change the influence of past force values but also change the overall magnitude of $\zeta$ . If we simply wanted to change the moving average behavior without changing the average $\Delta t$-level, we have to adjust $\alpha$ such that $\mathbb{E}(\zeta_n)$ is kept constant, where the expectation is with respect to the (unknown) evolving law of $\zeta_n$. We care only about what happens at long times. Assume we have $(x_0,p_0,\zeta_0)$ sampled from the (unknown) invariant measure of the SamAdams dynamics. 
Since our proposed integrator (Algorithm 1 in the main text) approximates the invariant measure, we may assume  $\mathbb{E}(\zeta_n)\approx E_\zeta$  and $\mathbb{E}\Big[\|\nabla U(x_n)\|^s\Big] \approx E_{U}$ for all $n$.  
Taking the expectation of \eqref{sup:eq:zeta_iterates}, we then have
\begin{equation}
E_\zeta \approx p^nE_{\zeta}+q\bigg(\sum_{i=1}^np^{n-i} E_{U}\bigg).
\end{equation}
Using properties of the geometric series and the definition of $q$, we obtain 
\begin{equation}
E_\zeta \approx \frac{1}{\Omega \alpha}E_U.
\end{equation}
Thus, if we change $\alpha$ to influence the amount of memory in the system but we want to preserve the mean adaptive stepsize $\mathbb{E}(\Delta t)$, we need to set

\begin{equation}\label{sup:eq:change_a1_a2}
\Omega^{\text{new}}=\frac{\alpha^{\text{old}}}{\alpha^{\text{new}}}\Omega^{\text{old}},
\end{equation}
keeping the product $\Omega \alpha$ constant. 
Meanwhile, if we simply want to change the average magnitude of $\zeta$ (and hence $\Delta t$) without changing the weights of the moving average, we simply change $\Omega$ while keeping $\alpha$ fixed.  
In general terms, we obtain the following rules of thumb.
\begin{enumerate}
\item For fixed $\alpha$, smaller values of $\Omega$ lead to larger $\zeta$ and hence smaller $\Delta t$.
\vspace{0.2cm}
\item For fixed $\Omega$, larger values of $\alpha$ lead to smaller $\zeta$ and hence larger $\Delta t$, and also to less memory in $\zeta$.
\vspace{0.2cm}
\item Larger values of $\alpha$ with $\Omega$ scaled according to \eqref{sup:eq:change_a1_a2} leads to less memory in $\zeta$ while keeping $\mathbb{E}(\Delta t)$ fixed.
\end{enumerate}
In Fig. \ref{sup:fig:dt_rules}, we demonstrate the effects of changing $\alpha$ and $\Omega$ on the obtained stepsizes on the 2-dimensional star potential (see Fig. \ref{fig:star_trajectory} in the main text).
\begin{figure}[htb!]
\begin{centering}
\includegraphics[width=1\textwidth]{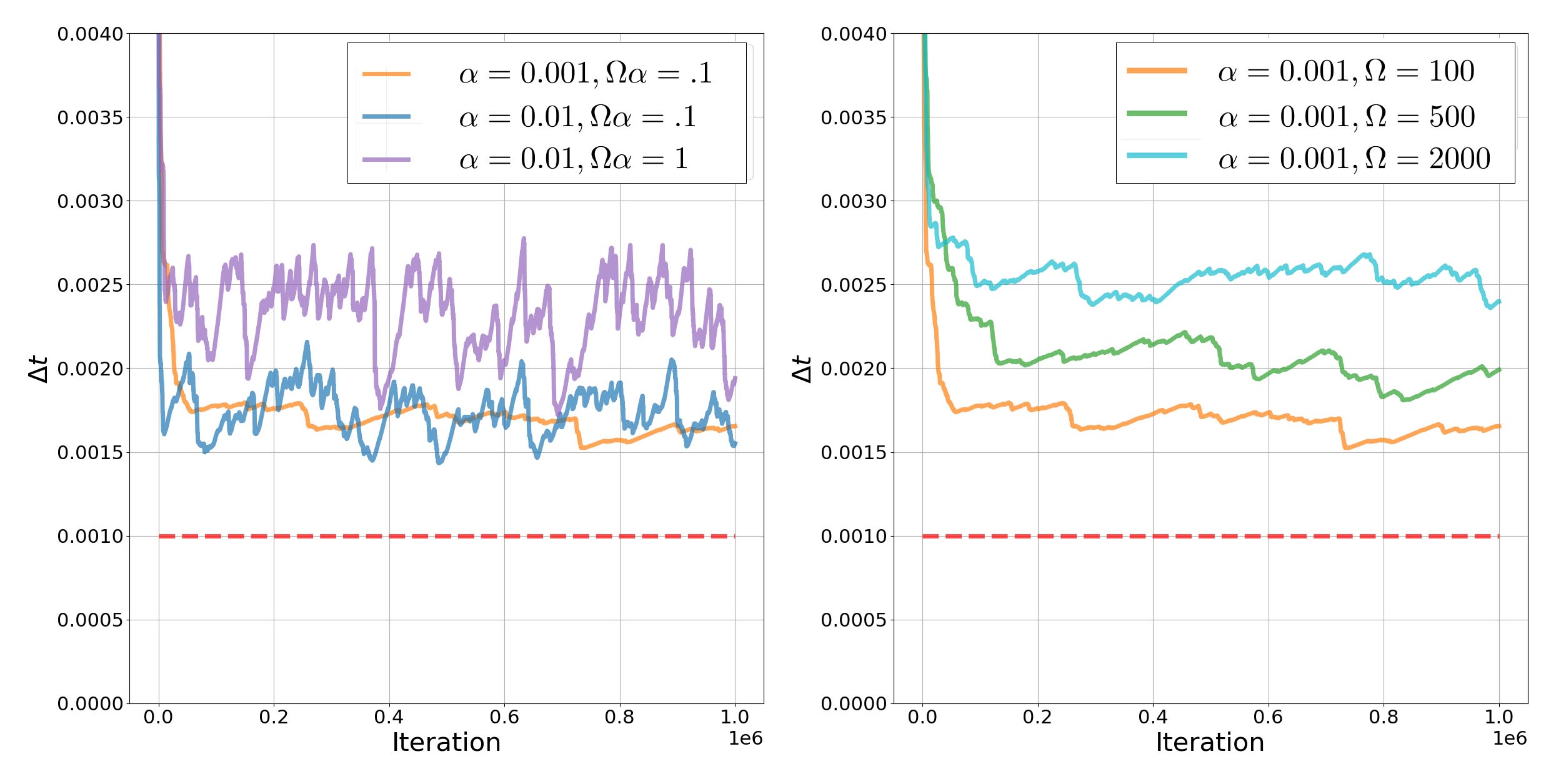} 
\end{centering}
\caption{\label{sup:fig:dt_rules} Adaptive stepsize obtained by SamAdams on a 2D test problem. \textbf{Left:} Changing $\alpha$ changes the smoothness of $\Delta t$. If $\Omega \alpha$ is identical for two different values of $\alpha$, $\mathbb{E}(\Delta t)$ is roughly identical as well. \textbf{Right:} Changing $\Omega$ while keeping $\alpha$ fixed changes  $\mathbb{E}(\Delta t)$ without influencing its smoothness. The red dashed lines denote the minimum admissible stepsize, i.e., $\Delta t=m\Delta \tau$ (here $\Delta \tau=0.01, \ m=10$.) }
\end{figure}
In the left-hand plot, we start with the orange curve obtained by $\alpha_1=0.001$, $\Omega=100$, and hence $\Omega \alpha=0.1$. Increasing the value of $\alpha$ by a factor of 10 and simultaneously changing $\Omega^{-1}$ by the same factor, such that $\Omega \alpha$ remains the same, we decrease the memory in the $\zeta$-dynamics leading to more spontaneous changes in $\Delta t$ while keeping the overall magnitude of $\Delta t$ fixed (blue curve). Increasing $\alpha$ without decreasing $\Omega$ will also remove memory from the system but in a less controlled way, as one also obtains different  $\mathbb{E}(\Delta t)$ (purple curve). On the right-hand side, starting again with the orange curve, we see that an increase of $\Omega$ while keeping $\alpha$ fixed leads to a systematic increase of $\Delta t$ without changing the moving average properties of the system. These mechanisms also work on more complex examples such as neural networks, see Fig. \ref{sup:fig:spiral_dt_variation}.
\begin{figure}[H]
\begin{centering}
\includegraphics[width=1\textwidth]{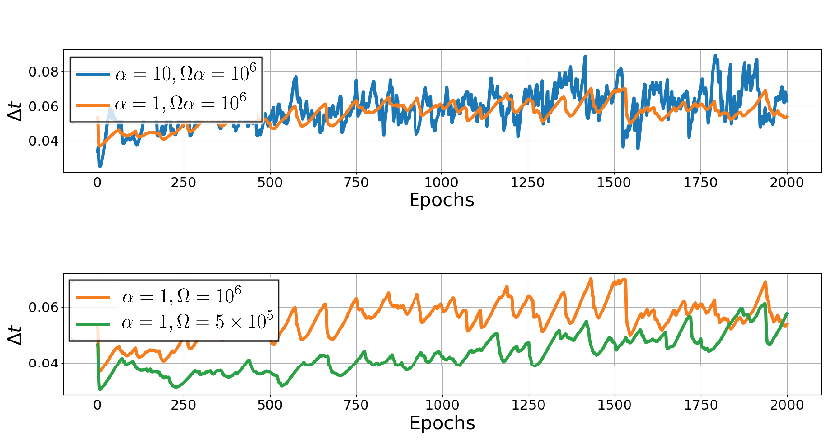} 
\end{centering}
\caption{\label{sup:fig:spiral_dt_variation} Adaptive stepsize trajectories of SamAdams for a simple fully connected neural network (one hidden layer)  on a spiral/Swiss-roll-type binary classification task. \textbf{Top:} Effect of changing $\alpha$ on $\Delta t$ if $\Omega \alpha$ remains fixed. \textbf{Bottom:} Effect of changing $\Omega$ if $\alpha$ remains fixed. Hyperparameters: $T=\gamma=1$, $\Delta \tau=0.03$, $m=0.1$, $M=10$, $r=0.25$, $s=2$. }
\end{figure}
Note that the rules (1)-(3) above can only be regarded as rules of thumb. They do not strictly hold near the stability threshold (close to which anything can happen). Even for stepsizes well below the threshold, the rules only hold for changes in $(\alpha, \Omega)$ that don't lead to qualitatively different trajectories. For example, increasing $\Omega$ while keeping $\alpha$ fixed will tend to increase $\Delta t$. A trajectory with higher $\Delta t$ might then experience larger forces compared to before, which will then lead to a decrease of $\Delta t$ again by virtue of larger obtained $\zeta$-values.

We further remark that on data science problems where the force $\nabla U$  is given by a sum over the data points, we find that $\Omega$ should often be chosen as $\mathcal{O}(N)$ or $\mathcal{O}(N^s)$ with $N$  the number of data points. This prevents $\zeta$ from drifting off to too large values which would lead to $\Delta t_n=m\Delta \tau\ \text{for all }n.$ Note that the loss gradient is often normalized by the number of data points by default in common machine learning packages, which might render this point obsolete.

Finally, it is also instructive to consider the limiting cases of \eqref{sup:eq:Z-step}. We have:
\begin{equation}
\lim_{\alpha\to\infty}\zeta_{n+1}=0,
\end{equation}
which leads to $\Delta t_{n+1}=M\Delta \tau$  for all $n$ (maximum $\Delta t$),
\begin{equation}
\lim_{\alpha \to 0}\zeta_{n+1}=\zeta_n + \Omega^{-1}\Delta \tau \|\nabla U(x_n) \|^s, 
\end{equation}
such that $\zeta_n \to \infty$ for $n\to \infty$ and hence $\Delta t_n\to m\Delta \tau$ (minimum $\Delta t$),
\begin{equation}
\lim_{\Omega\to 0}\zeta_{n+1}=\infty,
\end{equation}
and hence  $\Delta t_{n+1} = m\Delta \tau$ for all $n$ (minimum $\Delta t$), and
\begin{equation}
\lim_{\Omega \to \infty}\zeta_{n+1}=e^{-\Delta \tau \alpha}\zeta_n,
\end{equation}
leading to $\zeta_n\to0$ for $n\to \infty$ and hence $\Delta t_n \to M\Delta \tau$.

An interesting case arises for the limit $\alpha\to\infty$ such that $\Omega \alpha$ is kept constant,
\begin{equation}
\lim_{\substack{\alpha \to \infty, \\ \Omega \alpha = \text{const}}}\zeta_{n+1}=\frac{1}{\Omega \alpha}\| \nabla U(x_n) \|^s.
\end{equation}
This is a case with no memory but finite $\zeta$-values, where the adaptive stepsize is obtained via 
\begin{equation}
\Delta t = \psi(\| \nabla U(x)\|^{sr})\Delta \tau,
\end{equation}
which corresponds to the method used in \cite{alix}.
 
\FloatBarrier
\section{Symmetric Z-Steps}\label{sup:sec:symm_Z}
Algorithm \ref{sec:algorithm} in the main text consists of a symmetric integrator of the time-rescaled Langevin dynamics \eqref{eq: full_framework1}-\eqref{eq: full_framework2}, named ZBAOABZ, that adds two Z-steps as defined in Sec. \ref{sec: integration} before and after a constant-stepsize LD integrator, here BAOAB. Right after each of the Z-steps, the stepsize $\Delta t$ is updated. The BAOAB iteration is then executed with a stepsize given by the one obtained from the first Z-step.  While the symmetry of the LD integrator used, BAOAB, has beneficial impact on the integration error, the symmetric placing of two Z-steps in combination with two evaluations of $\psi(\zeta)$ is less obvious. Since $\hat{\Phi}^{Z}_{\frac{a}{2}}\circ \hat{\Phi}^{Z}_{\frac{a}{2}}=\hat{\Phi}^{Z}_a$, the dynamics would not change if the two Z-steps were merged and the Sundman transform $\psi$ was only evaluated once (namely right before the BAOAB iteration to update $\Delta t$). However, as described in Sec. \ref{sec:computing_averages}, the values $\psi(\zeta)$ are needed to reweight the collected samples of an observable $\phi(x_n,p_n)$ in order to obtain averages with respect to the correct measure. It is not obvious which value of $\psi(\zeta)$ should be used to reweight a sample generated by a BAOAB iteration: the value of $\psi(\zeta)$ before the BAOAB iteration, i.e., the value that governs the stepsize used by BAOAB, or the value of $\psi(\zeta)$ after the BAOAB iteration, namely the one that is obtained by updating $\zeta$ through the freshly obtained $(x,p)$-samples. Using two half-steps in $\zeta$ is a compromise between the two approaches which we found to outperform either of the two, see Fig. \ref{sup:fig:symm_Z_steps}.
\begin{figure}[]
\begin{centering}
\includegraphics[width=1\textwidth]{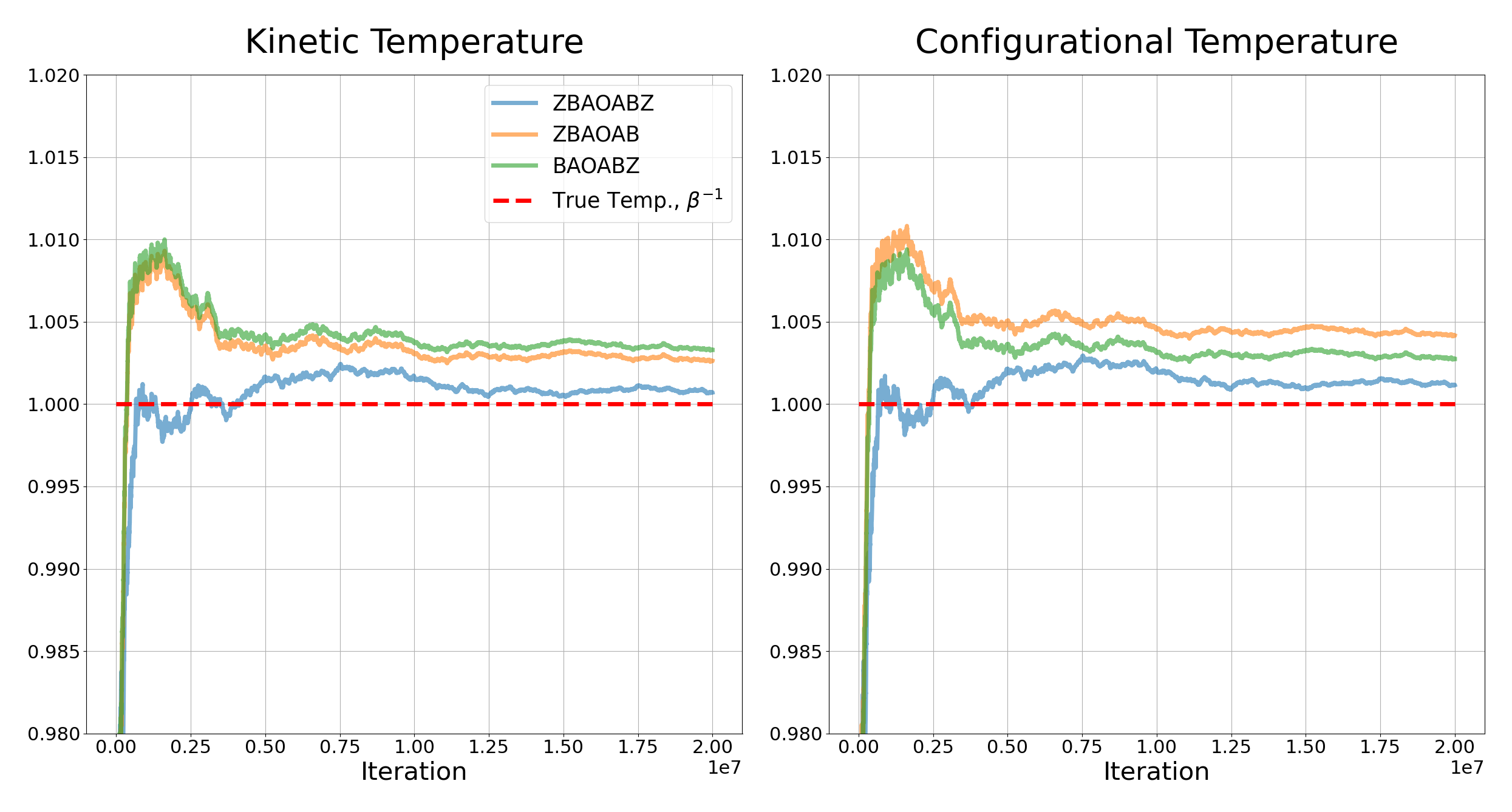}
\end{centering}
\caption{\label{sup:fig:symm_Z_steps} Comparison of three different Z-step placings concerning the computation of canonical averages of temperature observables. The symmetric placing, ZBAOABZ, corresponds to Alg. \ref{sec:algorithm} in the main text. The results were obtained on the star potential (see main text), and averaged over time and 200 trajectories. Simulation settings: $\beta^{-1}=\gamma=\alpha=1,\  \Delta \tau=0.01,\  g(x,p)=\|\nabla U(x) \|^2,\  \psi=\psi^{(1)} \text{ with } m=0.1,\  M=10,\ r=0.25$. }
\end{figure}

\FloatBarrier
\section{ZBAOABZ Observed Integration Order}\label{sup:sec:integration_order}
While the integrator ZBAOABZ (see Sec. \ref{sec: integration}) can be interpreted as solving Langevin Dynamics with variable stepsizes in time $t$, it is a conventional constant-stepsize scheme in time $\tau$ (the process given by \eqref{eq: full_framework1}-\eqref{eq: full_framework2}), for which conventional error analysis may be used. As mentioned in the main text and shown in \cite{LeMa2013}, the BAOAB integrator is of weak order 2 in the ergodic limit \footnote{The weak order even becomes 4 when considering configurational observables and working in the high-friction limit, $\gamma \to \infty$.}. This means that for a given observable $\phi(q)$, $q:=(x,p)$,  and stepsize $\Delta t$, we have
\begin{equation}
\lim_{n\to\infty}\Big|\mathbb{E}\big[\phi\big(q_n\big)\big]-\mathbb{E}\big[\phi\big(q(n\Delta t)\big)\big] \Big| < C\Delta t^2,
\end{equation}
with constant $C>0$ and $q_n$ the iterates obtained through the integrator.
The error analysis of ZBAOABZ is more complicated than BAOAB because of the presence of multiplicative noise, but from our analysis (see Sec. 3) we expect weak second order in the ergodic limit.
Here we verify this hypothesis numerically. We run ZBAOABZ for different values of stepsize $\Delta \tau$ and plot the weak errors of several observables against $\Delta \tau$ on the 2D star potential (see main text for its definition). For each $\Delta \tau$, we run $N_{tr}$ independent trajectories for $N_{iter}$ iterations, with $(N_{tr},\ N_{iter})=(300,\ 25\cdot 10^6)$ for $\Delta \tau>0.02$ and $(N_{tr},\ N_{iter})=(600,\ 100 \cdot 10^6 )$ for smaller $\Delta \tau$. The means were then obtained by first computing the time-average per trajectory (using the reweighting process described in Sec. \ref{sec:computing_averages}) and then by averaging over trajectories. The results are shown in Fig. \ref{sup:fig:integrator_order} and confirm the second order accuracy.

\begin{figure}[]
\begin{centering}
\includegraphics[width=1\textwidth]{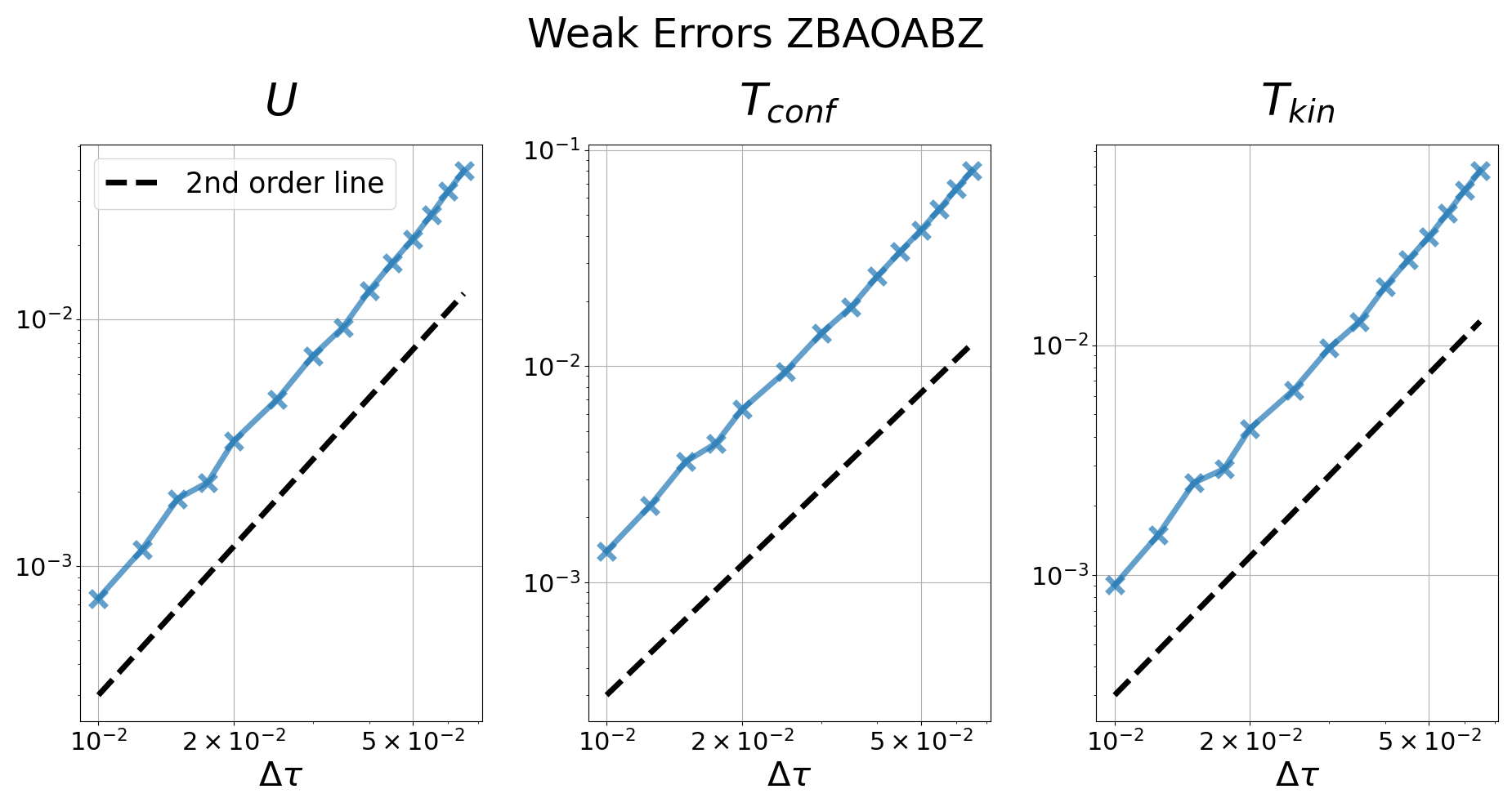}
\end{centering}
\caption{\label{sup:fig:integrator_order} Weak error of ZBAOABZ against stepsize $\Delta \tau$, obtained by averaging over time and independent trajectories. The ground truth of the potential energy was obtained via numerical quadrature. Simulation settings: $\gamma=0.1,\ \beta^{-1}=\alpha=1,\ g(x,p)=\|\nabla U(x) \|^2,\ \psi=\psi^{(1)}$ with $m=0.1,\ M=10,\ r=0.25$. }
\end{figure}

\FloatBarrier
\section{Example with Gradient Noise in Logistic Regression}\label{sup:sec: LogisticRegression}
We sample from the posterior of a logistic regression model on the forest covertype dataset \cite{covertype_31}. This time, we examine the effect of using stochastic gradients with various batch sizes.
We run SamAdams with various batch size and measure the resulting training accuracies. To compare it with BAOAB, we measure the mean adaptive stepsize of SamAdams on the full-batch run and execute all BAOAB runs at this fixed stepsize. This setting captures how the two samplers cope with changing batch sizes (leaving all other hyperparameters fixed). Fig. \ref{fig:LogReg_train_accuracies} shows the results for 4 different values of $\alpha$. 

\begin{figure}[h]
\begin{centering}
\includegraphics[width=1\textwidth]{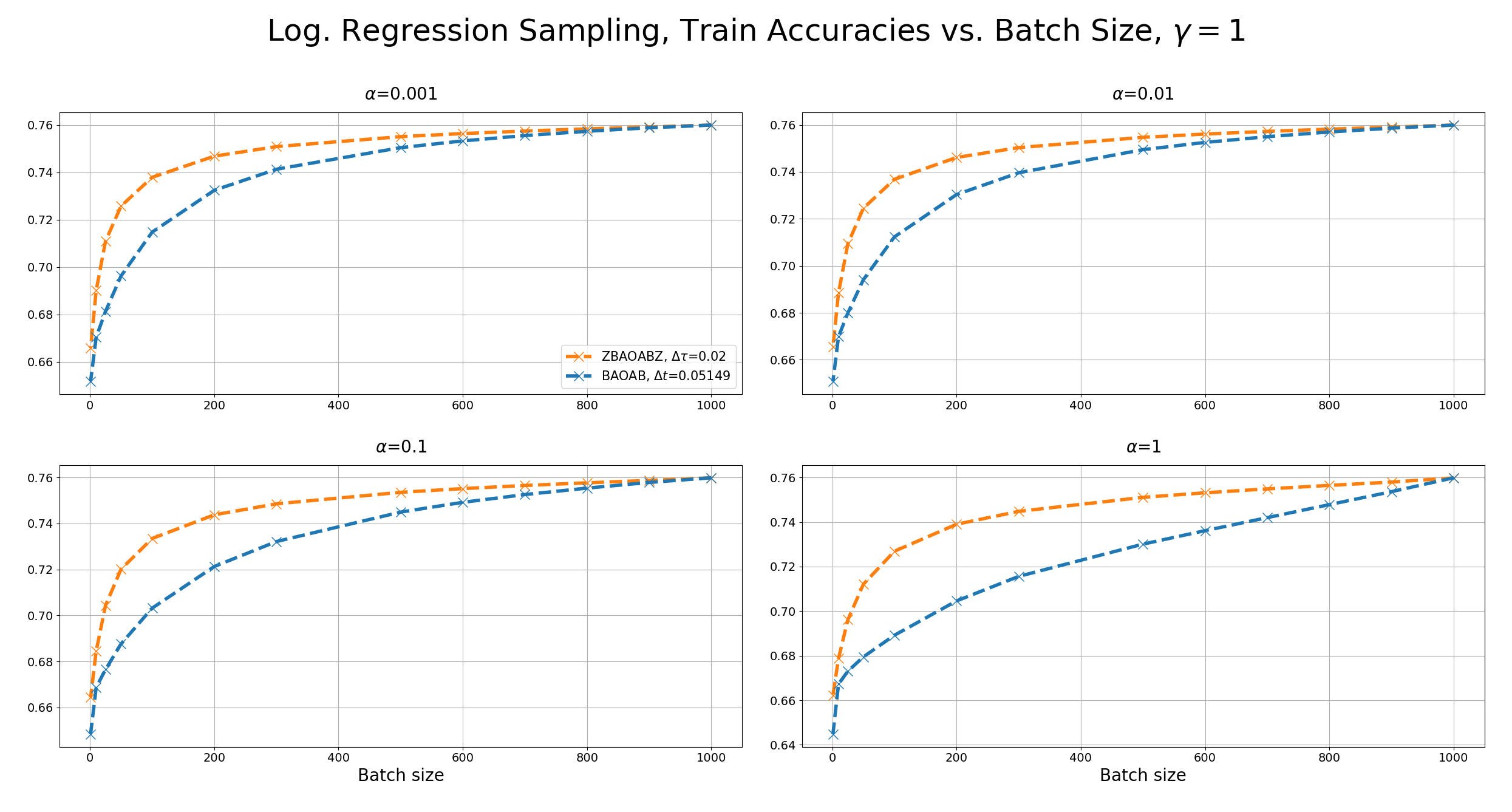}
\end{centering}
\caption{\label{fig:LogReg_train_accuracies} Training accuracies against stochastic gradient batch size $B$ for 4 different values of SamAdams hyperparameter $\alpha$. The stepsize of BAOAB was chosen to be identical to the mean of the adaptive stepsizses adopted by SamAdams on the full batch run on the given $\alpha$. 
}
\end{figure}
As the batch size decreases, the accuracy in both methods drops. However, due to the gradient noise, the otherwise simple dynamics\footnote{The logistic regression posterior is convex.} becomes more unstable, leading to automatic stepsize reduction in SamAdams. The sampling bias introduced by gradient noise from data subsampling scales with the stepsize, so a reduction of stepsize for decreasing batch size can restore accuracy. Note that the performance of BAOAB can be restored by picking a smaller stepsize as well. We merely wish to highlight the benefit of SamAdams to automatically react to induced gradient noise. Fig. \ref{fig:LogReg_dt_histograms} shows the batch size-dependent $\Delta t$-histograms obtained by SamAdams. 
\begin{figure}[]
\begin{centering}
\includegraphics[width=1\textwidth]{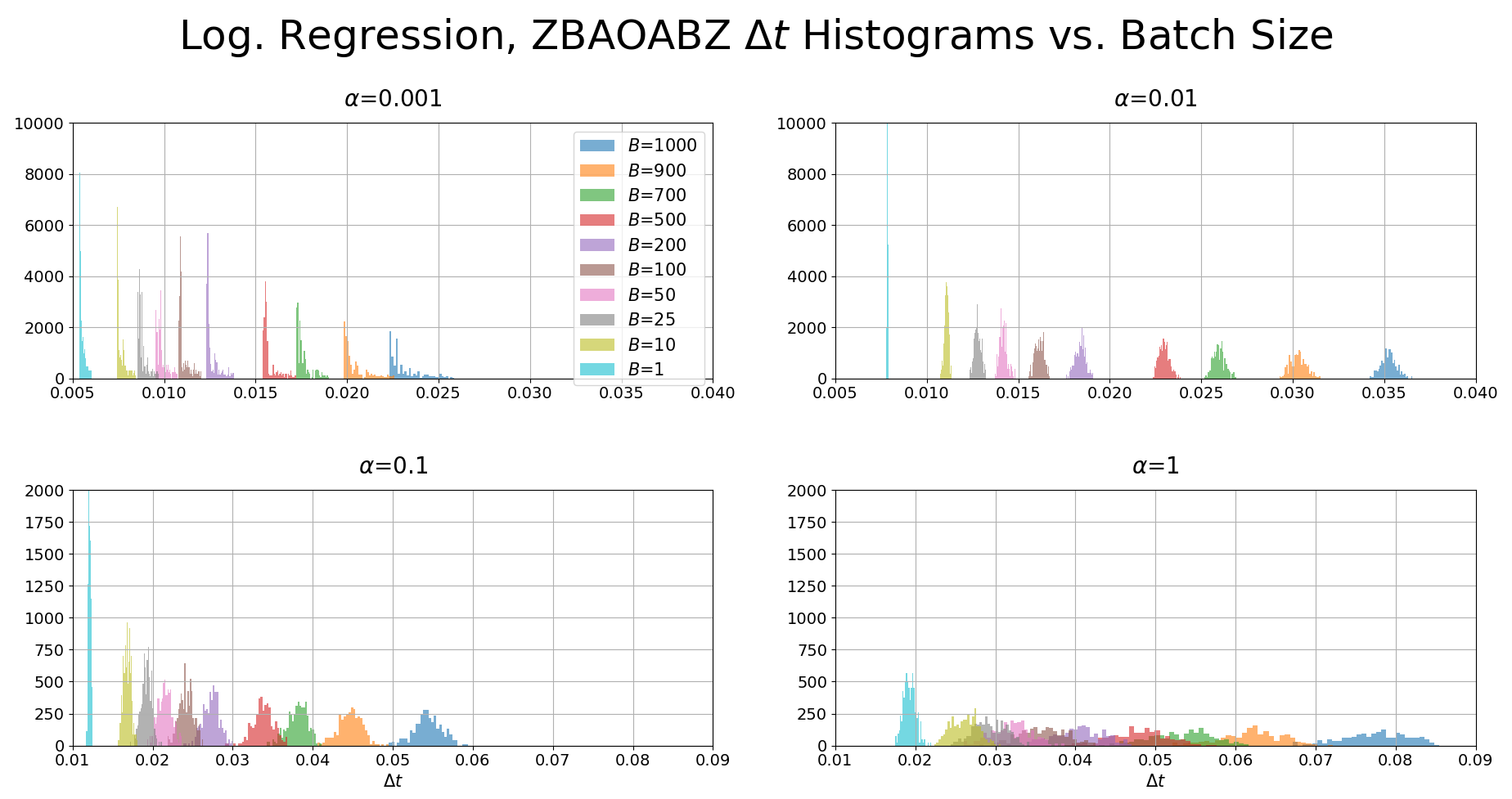}
\end{centering}
\caption{\label{fig:LogReg_dt_histograms} SamAdams $\Delta t$ histograms against batch size for the experiment of Fig. \ref{fig:LogReg_train_accuracies}.}
\end{figure}
We see a decrease of the average stepsize with decreasing batch size, consistent across all $\alpha$ values. The mean and variance (and even the type of the distribution) depends on the hyperparameters $\alpha$ and $\Omega$ (with $\Omega$ kept constant here).
To illustrate the sensitivity of the stepsize adaptation on gradient noise, we dynamically switch the batch size between 1 and 1,000 samples and inspect the behavior of the stepsize $\Delta t$. Fig. \ref{fig:LogReg_dt_alternating_B} shows the result, together with curves for constant batch size $B=1$ and $B=1,000$.
\begin{figure}[]
\begin{centering}
\includegraphics[width=1\textwidth]{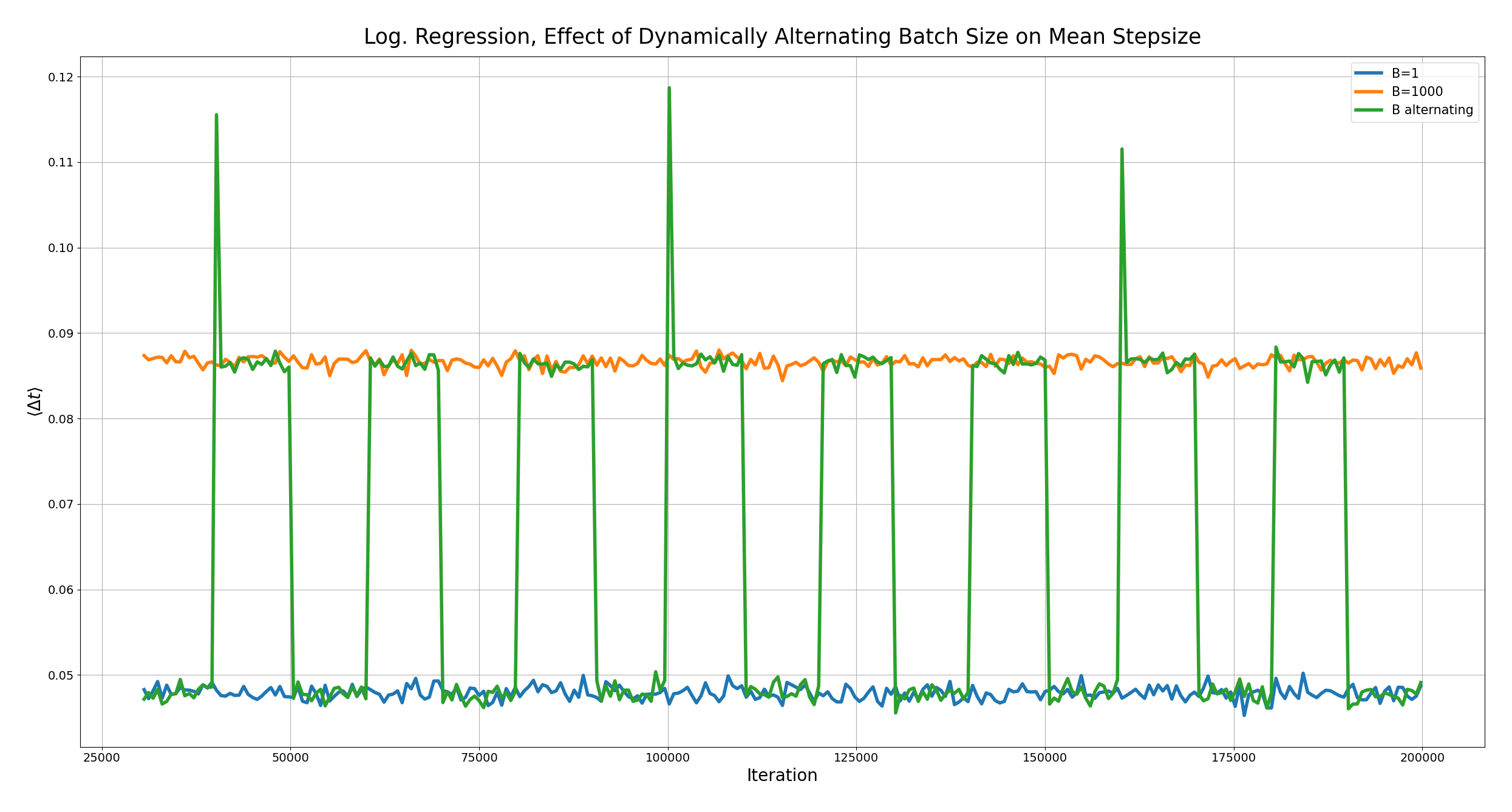}
\end{centering}
\caption{\label{fig:LogReg_dt_alternating_B} Adaptive stepsize of SamAdams when run with different batches. The green curve uses alternating batch sizes with $B$ changing between 1 and 1,000. The curves were averaged over 100 independent trajectories.}
\end{figure}
The stepsize adaptation is able to change $\Delta t$ dynamically and almost instantaneously in accordance with the batch size. While we did not find significant differences in the resulting accuracies of SamAdams and BAOAB when using dynamically changing batch sizes, we believe that this is a consequence of the simplicity of the problem. Logistic regression on comparably small datasets does not require stochastic gradients, neither for efficiency nor to escape local minima. It would be interesting to examine the combined impact of stepsize adaptation and gradient noise on large-scale deep learning models, where gradient subsampling is quintessential for efficient training.  For example, there have been several recent works in different contexts about the relationship between stepsize and batch size and about dynamically changing batch sizes (see, e.g., \cite{batch_size_schedule,large_batch_training,dynamic_transformer_batch_size,active_learning}). A potential direction of future research could then aim to answer the question whether the ideas in those works would benefit from a batch size sensitive stepsize adaptation.

\FloatBarrier
\section{CNN Architecture}\label{sup:sec:CNN_architecture}

The following code specifies the simple convolutional neural network used for the MNIST experiments in Sec. \ref{sec: MNIST_CNN} in the main text, implemented in PyTorch.
\begin{lstlisting}
class SimpleCNN(nn.Module):
   def __init__(self):
      super(SimpleCNN, self).__init__()
      self.conv1 = nn.Conv2d(in_channels=1, out_channels=32, kernel_size=3, padding=1)
      self.pool  = nn.MaxPool2d(kernel_size=2, stride=2)
      self.conv2 = nn.Conv2d(in_channels=32, out_channels=64, kernel_size=3, padding=1)
      self.conv3 = nn.Conv2d(in_channels=64, out_channels=128, kernel_size=3, padding=1)

      self.fc_input_size = 128 * 3 * 3   # based on MNIST image dimension.

      self.fc1 = nn.Linear(self.fc_input_size, 512)
      self.fc2 = nn.Linear(512, 256)
      self.fc3 = nn.Linear(256, 10)

   def forward(self, x):
      x = self.pool(nn.ReLU()(self.conv1(x)))
      x = self.pool(nn.ReLU()(self.conv2(x)))
      x = self.pool(nn.ReLU()(self.conv3(x)))
      x = x.view(x.size(0), -1)
      x = nn.ReLU()(self.fc1(x))
      x = nn.ReLU()(self.fc2(x))
      x = self.fc3(x)
      x = torch.log_softmax(x, dim=1)
      return x
\end{lstlisting}

\FloatBarrier
\end{document}